\documentclass[a4paper]{article}

\usepackage[utf8]{inputenc}
\usepackage{amsmath,amsthm,amsfonts,amssymb}
\usepackage{mathtools}
\usepackage{dsfont}
\usepackage{a4wide}
\usepackage[toc]{appendix}
\usepackage{natbib}
\usepackage[ruled, linesnumbered]{algorithm2e}
\usepackage{hyperref}
\usepackage{xcolor}
\usepackage{placeins}
\usepackage{authblk}
\usepackage{appendix}
\usepackage{tikz}
\usetikzlibrary{arrows.meta}

\newcommand{\Din}{\mathcal{D}_{\operatorname{in}}}
\newcommand{\Dout}{\mathcal{D}_{\operatorname{out}}}

\newcommand{\Ptest}{P_{\operatorname{test}}}

 \newcommand{\independent}{\perp\!\!\!\!\perp} 

\newcommand{\supp}[1]{\operatorname{supp}(#1)}

\newcommand{\optx}{\xi}
\newcommand{\optz}{\nu}
\newcommand{\balpha}{\boldsymbol{\alpha}}
\newcommand{\bbeta}{\boldsymbol{\beta}}

\def\methodname/{\texttt{Xtrapolation}}

\DeclareMathOperator*{\argmin}{arg\,min}

\newtheorem{theorem}{Theorem}

\newtheorem{lemma}{Lemma}

\newtheorem{corollary}{Corollary}
\newtheorem{example}{Example}
\newtheorem{proposition}{Proposition}
\newtheorem{definition}{Definition}

\newtheorem{remark}{Remark}

\title{Extrapolation-Aware Nonparametric Statistical Inference}
\date{\today}

\author[1]{Niklas Pfister}
\author[2]{Peter B\"uhlmann}

\affil[1]{Department of Mathematical Sciences, University of
  Copenhagen, Denmark}
\affil[2]{Seminar f\"ur Statistik, ETH Z\"urich, Switzerland}

\begin{document}

\maketitle

\begin{abstract}
  We define extrapolation as any type of statistical inference on a
  conditional function (e.g., a conditional expectation or conditional
  quantile) evaluated outside of the support of the conditioning
  variable. This type of extrapolation occurs in many data analysis
  applications and can invalidate the resulting conclusions if not
  taken into account.  While extrapolating is straightforward in
  parametric models, it becomes challenging in nonparametric
  models. In this work, we extend the nonparametric statistical model
  to explicitly allow for extrapolation and introduce a class of
  extrapolation assumptions that can be combined with existing
  inference techniques to draw extrapolation-aware conclusions. The
  proposed class of extrapolation assumptions stipulate that the
  conditional function attains its minimal and maximal directional
  derivative, in each direction, within the observed support.  We
  illustrate how the framework applies to several statistical
  applications including prediction and uncertainty quantification.
  We furthermore propose a consistent estimation procedure that can be
  used to adjust existing nonparametric estimates to account for
  extrapolation by providing lower and upper extrapolation bounds.
  The procedure is empirically evaluated on both simulated and
  real-world data.
\end{abstract}

\section{Introduction}

In the natural sciences the term extrapolation broadly refers to any
process that extends conclusions about observed settings to previously
unseen settings. For example, we extrapolate if we learn the
gravitational constant on earth in a controlled experiment in a lab
and later use it as part of a model (in this case based on the laws of
physics) to predict the energy required to launch a rocket into
space. While extrapolation with a known mechanistic model generally
works well, it becomes much more challenging in more noisy, complex or
chaotic settings where full mechanistic knowledge is unavailable. For
instance, given data from a randomized control trial for a specific
drug based on an adult cohort, we might want to extrapolate the study
results to how infants are affected by the same drug. In this case the
underlying mechanism is not fully understood and at least some
additional knowledge is required to draw reliable conclusions. In this
work we focus on extrapolation when only a statistical model without
mechanistic knowledge is available. To make this more precise, let
$P_0$ be a distribution over $(X, Y)\in\mathcal{X}\times\mathbb{R}$
where $X$ are covariates and $Y$ is a response variable and let
$\supp{X}$ denote the support of $X$. Furthermore, assume we are
interested in a conditional (on $X$) function
$\Phi_{0}:\mathcal{X}\rightarrow\mathbb{R}$, e.g., a conditional
expectation. Extrapolation throughout this work will then refer to any
statistical inference on $\Phi_0(x)$ for
$x\in\mathcal{X}\setminus\supp{X}$. Since it depends on the
assumptions we make on $P_0$ whether the conditional function is even
well-defined outside of $\supp{X}$, extrapolation can only be
meaningful under appropriate assumptions.

Existing works have used various different assumptions to render
extrapolation a well-defined inference task. These can be roughly
categorized into three groups. Firstly, \emph{global parametric
  assumptions}, which assume the conditional function $\Phi_0$ is
parametric on all of $\mathcal{X}$ (e.g., linear or polynomial) in a
way that ensures that the parameters can be (partially) identified
from the observed distribution $P_0$ and hence used to either identify
or bound the behavior of $\Phi_0$ on all of $\mathcal{X}$. All
parametric statistical models fall into this category, as well as many
semiparametric models. Secondly, \emph{functional constraint
  assumptions}, which assume that $\Phi_0$ has specific properties
that transfer from $\supp{X}$ to all of $\mathcal{X}$ (e.g.,
monotonicity, periodicity or smoothness). In particular for
nonparametric regression there are many methods that make such
assumptions implicitly, for example, by assuming $\Phi_0$ extrapolates
constant (e.g., tree ensembles) or linear outside of the support
\citep{li2003local, christiansen2021causal}. Other works make more
explicit assumptions on patterns or periodicity
\citep{wilson2013gaussian,wang2022local} or by assuming additive
nonlinear functions \citep{dong2022first}. Thirdly, \emph{mechanisitic
  assumptions}, which assume an underlying mechanistic model (often
causal) which ensures that $\Phi_0$ is (partially) identified on
$\mathcal{X}$. For example, physical laws can sometimes be used to
constrain the model class which can then improve how well a model
extrapolates \citep[e.g.,][]{pfister2019causalkinetix}. Recent works
have also assumed causal structure such as independent additive noise
to evaluate non-linear functions outside of the support
\citep{shen2023engression, saengkyongam2023identifying}.

While extrapolating under global parametric assumptions is
straightforward, it becomes more challenging for functional constraint
and mechanistic assumptions. Most nonparametric works that explicitly
consider extrapolation have generally focused on prediction of $Y$ at
a point $x$. We suspect one reason for this is that in nonparametric
statistics, the target of inference is conventionally a quantity that
is identified from the data generating distribution $P_0$. However,
when extrapolating, the target of inference -- $\Phi_0(x)$ for
$x\in\mathcal{X}\setminus\supp{X}$ -- is a-priori not a function of
$P_0$. Prediction is therefore a natural task to consider as it can be
viewed model-free (as commonly done in the machine learning community)
ensuring a well-defined extrapolation task without explicitly defining
$\Phi_0$. In this work, we propose a framework, based on Markov
kernels, that extends the nonparametric approach in a way that
explicitly allows us to consider extrapolation of any conditional
function $\Phi_0$ defined via the Markov kernel. We further ensure
partial identifiability by assuming that $\Phi_0$ has directional
derivatives on $\mathcal{X}$ that are bounded by its directional
derivatives on $\supp{X}$. This functional constraint type assumption
allows us to apply Taylor's theorem to construct extrapolation bounds
on $\Phi_0$ on all of $\mathcal{X}$ that are identified by $P_0$. We
show that these bounds can be useful in a range of statistical
inference tasks and can be estimated consistently from data. Our
framework for incorporating extrapolation in nonparametric statistical
inference, comes with the benefit that it can be applied with any
existing nonparametric estimate of $\Phi_0$ on $\supp{X}$ as long as
the extrapolation assumption for $\Phi_0$ holds. Importantly, the
resulting inference is extrapolation-aware in the sense that (in the
large sample limit) it remains unchanged if no extrapolation occurs,
since the lower and upper extrapolation bounds overlap on
$\supp{X}$. A further benefit of the proposed approach is that it does
not require knowledge of $\supp{X}$ which is often unknown and
non-trivial to estimate from data, in particular if $X$ is
multi-dimensional. The extrapolation bounds automatically adapt to
become wider in regions of the $X$-space with few (or no) observations
reasonably close by and where there is a large extrapolation
uncertainty in conditional function $\Phi_0$. We use this property to
derive a score that quantifies extrapolation which may be useful in
applications.

There is a large body of literature on related types of
extrapolation. The fields of domain adaptation and generalization
\citep{pan2009survey, schreiber2023encode}, which aim to find
predictive models that perform well on a test distribution that is
different from the training distribution, has obvious connections that
we discuss in Section~\ref{sec:out-of-support-prediction}. In most of
the literature, however, one assumes that the test and training
distributions have overlapping supports, which excludes extrapolation
as defined here. One exception are distributionally robust
optimization methods that use the Wasserstein distance which allows
for the supports to be disjoint \citep[e.g.,][]{sinha2018certifying}.
A further related area is causal inference, where the term
extrapolation is sometimes used to refer to the task of generalizing
from the observational to a previously unseen interventional
distribution. Again, most works in this area explicitly exclude
extrapolation as defined here by assuming overlapping
supports. Nevertheless, as we discuss in Section~\ref{sec:causality},
our framework can also be viewed from a causal
perspective. Importantly, as it extends the conventional nonparametric
model, it can also be applied to causal inference tasks with the same
causal assumptions required as in the non-extrapolation setting.

The remainder of the paper is structured as follows. In
Section~\ref{sec:taylor_extrapolation}, we use Taylor's theorem to
derive extrapolation bounds for specific types of differentiable
functions. In Section~\ref{sec:statistical_extrapolation}, we
introduce a nonparametric statistical framework that explicitly
accounts for extrapolation and discuss how the extrapolation bounds
can be applied in three statistical applications: out-of-distribution
prediction, extrapolation-aware uncertainty quantification and
quantifying extrapolation. We furthermore discuss a causal perspective
on extrapolation in Section~\ref{sec:causality}. In
Section~\ref{sec:estimation}, we propose an approach for estimating
the extrapolation bounds and prove that it is consistent. Finally, in
Section~\ref{sec:numerical_experiments}, we empirically investigate
how well the extrapolation bounds can be estimated on simulated data
and illustrate their use for providing extrapolation-aware prediction
intervals on two real-world data sets.

\paragraph{Notation} Let $\mathcal{X}\subseteq\mathbb{R}^d$ be a fixed
domain, denote by $C^{q}(\mathcal{X})$ the set of $q$-times
continuously differentiable functions and define $\mathcal{B}\coloneqq \{x\in\mathbb{R}^d\mid \|x\|_2=1\}$.  For all
$f\in C^{q}(\mathcal{X})$ and all $v\in\mathbb{R}^d$ define the
directional derivative $D_vf:\mathcal{X}\rightarrow\mathbb{R}$ for all
$x\in\mathcal{X}$ by
\begin{equation*}
    D_vf(x)\coloneqq\lim_{h\rightarrow 0}\frac{f(x+h\cdot v)-f(x)}{h}.
\end{equation*}
Moreover, for all $\ell\in\{1,\ldots,q\}$, define the $\ell$-th
directional derivative recursively by
$D_v^{\ell}f\coloneqq D_v(D_v^{\ell-1}f)$, where $D_v^0f=f$. For all
multi-indices $\balpha\in\mathbb{N}^d$ with
$|\balpha|\coloneqq \sum_{j=1}^d\balpha^j\leq q$ and all functions
$f\in C^q(\mathcal{X})$ define
$\partial^{\balpha}f\coloneqq D_{e_1}^{\balpha^1}\cdots
D_{e_d}^{\balpha^d}f$ and for all $j\in\{1,\ldots,d\}$ define
$\partial_jf=D_{e_j}^1f$, where $e_j$ denotes the $j$-th canonical
unit vector. Additionally, define the function
$\overline{v}:\mathcal{X}\times\mathcal{X}\rightarrow\mathbb{R}^d$
which maps all points $x_0,x_1\in\mathcal{X}$ to the unit vector
pointing in the same direction as the vector from $x_0$ to $x_1$, that
is, for all $x_0,x_1\in\mathcal{X}$ define
$\overline{v}(x_0, x_1)\coloneqq
\frac{x_1-x_0}{\|x_1-x_0\|_2}\mathds{1}_{\{x_0\neq x_1\}}$.  Lastly,
we denote by $\mathfrak{B}(\mathbb{R})$ the Borel-sigma algebra on
$\mathbb{R}$.

\section{Extrapolation via Taylor's theorem}\label{sec:taylor_extrapolation}

We start by considering extrapolation from a fully deterministic
perspective. Assume we are given a function $f\in C^q(\mathcal{X})$
but are only able to evaluate it on a closed domain
$\mathcal{D}\subseteq\mathcal{X}$. Since, $f$ is $q$-times
continuously differentiable, knowing the function on $\mathcal{D}$ can
help constrain how the function can behave on
$\mathcal{X}\setminus\mathcal{D}$. Formally, using Taylor's theorem
\citep{taylor1715}, we get for all $x_0\in\mathcal{D}$ and all
$x_1\in\mathcal{X}$ that there exists a $c\in [0, 1]$ such that for
$\xi\coloneqq cx_1+(1-c)x_0$ it holds that
\begin{equation}
    \label{eq:taylor_equation}
    f(x_1)= \sum_{\ell=0}^{q-1}D_{\overline{v}(x_0, x_1)}^{\ell}f(x_0)\frac{\|x_1-x_0\|_2^{\ell}}{\ell!} + D_{\overline{v}(x_0, x_1)}^{q}f(\xi) \frac{\|x_1-x_0\|_2^{q}}{q!}.
\end{equation}
The only quantity in this equation which can -- depending on $c$ --
rely on evaluating the function outside of $\mathcal{D}$ is the value
of $D_{\overline{v}(x_0, x_1)}^{q}f(\xi)$. However, if we are willing
to assume that the $q$-th directional derivative of $f$ is bounded on
all of $\mathcal{X}$, we can use \eqref{eq:taylor_equation} to
construct upper and lower bounds on $f(x_1)$.  Our approach is based
on assuming that $f$ -- the function we want to extrapolate -- behaves
at most as 'extreme' on $\mathcal{X}$ as on $\mathcal{D}$. To be at
most as 'extreme' in our setting, means that the directional
derivatives of $f$ on $\mathcal{X}$ are bounded by its directional
derivatives on $\mathcal{D}$, for all possible directions.
\begin{definition}[$q$-th derivative dominated]
  \label{def:qth_order_dominated}
  Let $f\in C^q(\mathcal{X})$ be a function and
  $\mathcal{D}\subseteq\mathcal{X}$ be a non-empty closed set. A
  function $g\in C^{q}(\mathcal{X})$ is called \emph{$q$-th derivative
    dominated by $f$ over $\mathcal{D}$}, denoted by
  $g \triangleleft_{\mathcal{D}}^q f$, if it holds for all
  $v\in\mathcal{B}$ that
  \begin{equation*}
    \inf_{x\in\mathcal{X}}D_v^qg(x)\geq \inf_{x\in\mathcal{D}}D_v^qf(x)
    \quad\text{and}\quad
    \sup_{x\in\mathcal{X}}D_v^qg(x)\leq \sup_{x\in\mathcal{D}}D_v^qf(x).
  \end{equation*}
\end{definition}
In words a function $g$ is $q$-th derivative dominated by $f$ on
$\mathcal{D}$ if all its directional derivatives of order $q$ are
bounded on all of $\mathcal{X}$ by the corresponding directional
derivatives of $f$ on $\mathcal{D}$. Based on this definition we now
consider functions $f\in C^q(\mathcal{X})$ that satisfy
\begin{equation}
  \label{eq:amae}
  f\triangleleft_{\mathcal{D}}^q f.
\end{equation}
This formalizes the previously mentioned intuitive notion of behaving
at most as 'extreme' on $\mathcal{X}$ as on $\mathcal{D}$. Whenever a
function satisfies \eqref{eq:amae}, we can use
\eqref{eq:taylor_equation} to provide bounds on its behavior on
$\mathcal{X}$ that only depend on the values it attains on
$\mathcal{D}$.

\begin{definition}[Extrapolation bounds]
  \label{def:extrapolation_bounds}
  For all $f\in C^q(\mathcal{X})$ and all non-empty closed
  $\mathcal{D}\subseteq\mathcal{X}$ define the \emph{extrapolation
    bounds}
  $B_{q, f, \mathcal{D}}^{\operatorname{lo}}, B_{q, f,
    \mathcal{D}}^{\operatorname{up}}:\mathcal{X}\rightarrow[-\infty,\infty]$
  given for all $x\in\mathcal{X}$ by
  \begin{equation*}
    B_{q, f, \mathcal{D}}^{\operatorname{lo}}(x)\coloneqq \sup_{x_0\in\mathcal{D}}\left(\sum_{\ell=0}^{q-1}D_{\overline{v}(x_0, x)}^{\ell}f(x_0)\frac{\|x-x_0\|_2^{\ell}}{\ell!} + \inf_{z\in\mathcal{D}}D^q_{\overline{v}(x_0,x)}f(z) \frac{\|x-x_0\|_2^{q}}{q!}\right)
  \end{equation*}
  and
  \begin{equation*}
    B_{q, f,\mathcal{D}}^{\operatorname{up}}(x)\coloneqq \inf_{x_0\in\mathcal{D}}\left(\sum_{\ell=0}^{q-1}D_{\overline{v}(x_0, x)}^{\ell}f(x_0)\frac{\|x-x_0\|_2^{\ell}}{\ell!} + \sup_{z\in\mathcal{D}}D^q_{\overline{v}(x_0,x)}f(z) \frac{\|x-x_0\|_2^{q}}{q!}\right).
  \end{equation*}
\end{definition}
The extrapolation bounds are constructed using
\eqref{eq:taylor_equation} and then replacing the highest order
derivative with the worst possible directional derivative $f$ attains
in $\mathcal{D}$. Since the resulting bound is valid for any anchor
point $x_0\in\mathcal{D}$, we select the one that results in the
tightest bound. From this construction it can be shown that for
all $x\in\mathcal{D}$ the bounds satisfy
\begin{equation*}
  B_{q, f, \mathcal{D}}^{\operatorname{lo}}(x)=f(x)=B_{q, f, \mathcal{D}}^{\operatorname{up}}(x),
\end{equation*}
as long as the $q$-th directional derivatives of $f$ in all directions
are bounded on $\mathcal{D}$. The bounds are visualized for three one
dimensional functions in
Figure~\ref{fig:extrapolation_bounds_visualization}. Different orders
$q$ capture different aspects of the function, for example, for $q=1$
monotone behavior is captured in
Figure~\ref{fig:extrapolation_bounds_visualization}
(middle). Moreover, as seen in
Figure~\ref{fig:extrapolation_bounds_visualization} (left), the
extrapolation bounds only bound the true function if it indeed
satisfies $f\triangleleft_{\mathcal{D}}^qf$.
\begin{figure}[t]
  \centering
  \includegraphics{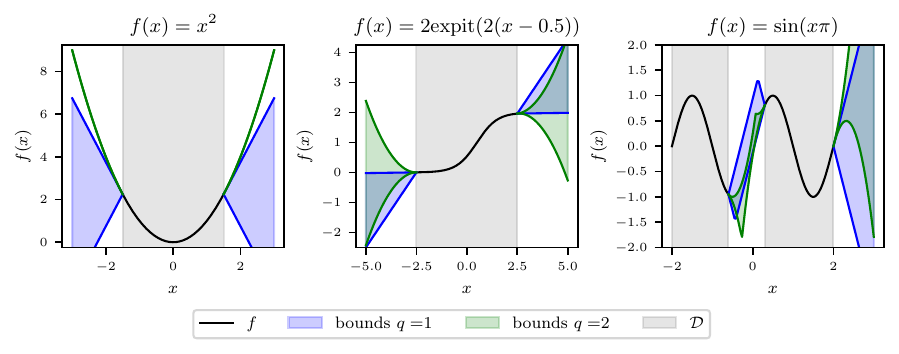}
  \caption{Visualization of the extrapolation bounds given in
    Definition~\ref{def:extrapolation_bounds} for three different
    functions $f$ and domains $\mathcal{D}$. The shaded gray area
    corresponds to $\mathcal{D}$ on which the function we would like
    to extrapolate is given in black. Blue corresponds to the first
    order upper and lower extrapolation bounds
    $B_{1,f,\mathcal{D}}^{\operatorname{lo}}$ and
    $B_{1,f,\mathcal{D}}^{\operatorname{up}}$ and green to the second
    order extrapolation bounds
    $B_{2,f,\mathcal{D}}^{\operatorname{lo}}$ and
    $B_{2,f,\mathcal{D}}^{\operatorname{up}}$. For this visualization,
    we approximate the bounds by sampling points uniformly in
    $\mathcal{D}$, which is consistent by
    Theorem~\ref{thm:consistency} below.
  }
  \label{fig:extrapolation_bounds_visualization}
\end{figure}
The following theorem provides a formal connection between the
extrapolation bounds for a function $f\in C^q(\mathcal{X})$ and
functions that are $q$-th derivative dominated by $f$.
\begin{theorem}[Properties of extrapolation bounds]
  \label{thm:taylor_extrapolation_bounds}
  Let $f\in C^q(\mathcal{X})$ be a function and
  $\mathcal{D}\subseteq\mathcal{X}$ be a non-empty closed set. Then,
  the following two statements hold:
  \begin{itemize}
  \item[(i)] For all $g\in C^{q}(\mathcal{X})$ satisfying for all
    $x\in\mathcal{D}$ that $g(x)=f(x)$ and
    $g\triangleleft_{\mathcal{D}}^q f$ it holds that
    \begin{equation*}
      \forall x\in\mathcal{X}:\quad B_{q, f, \mathcal{D}}^{\operatorname{lo}}(x)\leq g(x)\leq
      B_{q, f, \mathcal{D}}^{\operatorname{up}}(x).
    \end{equation*}
  \item[(ii)] If $\mathcal{X}$ is compact, there exists, for all
    $\star\in\{\operatorname{lo}, \operatorname{up}\}$, a sequence
    $(g_n^{\star})_{n\in\mathbb{N}}\subseteq C^q(\mathcal{X})$
    satisfying
    \begin{equation*}
      \lim_{n\rightarrow\infty}\sup_{x\in\mathcal{X}}\left|B_{q, f,
          \mathcal{D}}^{\star}(x)-g_n^{\star}(x)\right|=0
      \quad\text{and}\quad
      \forall n\in\mathbb{N}:\, g_n^{\star}\triangleleft_{\mathcal{D}}^q f.
    \end{equation*}
  \end{itemize}
\end{theorem}
A proof is given in Supplementary
material~\ref{proof:taylor_extrapolation_bounds}.  Part (i) in
particular implies that for all $f\in C^q(\mathcal{X})$ satisfying
$f\triangleleft_{\mathcal{D}}^qf$ the extrapolation bounds bound the
values of $f$ on all of $\mathcal{X}$, i.e.,
\begin{equation*}
  \forall x\in\mathcal{X}:\quad B_{q, f, \mathcal{D}}^{\operatorname{lo}}(x)\leq f(x)\leq
  B_{q, f, \mathcal{D}}^{\operatorname{up}}(x).
\end{equation*}
Part (ii) provides a partial converse of this statement. Specifically,
it states that the extrapolation bounds can be uniformly approximated
by a sequence of functions that are $q$-th derivative dominated by
$f$. As a side result, part~(ii) further implies that the
extrapolation bounds are uniformly continuous functions.

\section{Extrapolation in statistical applications}\label{sec:statistical_extrapolation}

In the previous section we used Taylor's theorem to construct
extrapolation bounds on functions under the assumption that the
directional derivatives of order $q$ are at most as extreme on
$\mathcal{X}$ as they are inside the observed domain $\mathcal{D}$. We
now use this as an extrapolation assumption in statistical
applications where we link the data-generating distributions to target
conditional functions. To ease notation we assume $q$ is fixed and
drop it from the notation.  We begin by formalizing extrapolation
within a nonparametric statistical model.

Let $P_0$ be a (data-generating) distribution over
$(X, Y)\in\mathcal{X}\times\mathbb{R}$ and define
$\Din\coloneqq\supp{X}$ and $\Dout\coloneqq\mathcal{X}\setminus\Din$.
Furthermore, let
$Q_0:\mathcal{X}\times\mathfrak{B}(\mathbb{R})\rightarrow[0,1]$ be a
Markov kernel that for all distributions over $X$ generates a
corresponding distribution over $Y$ and satisfies for all
$B\in\mathfrak{B}(\mathbb{R})$ that
\begin{equation}
\label{eq:probability_kernel_prop}
    P_0^Y(B) =\int_{\mathcal{X}}Q_0(x, B)P_0^X(dx),
\end{equation}
where $P_0^X$ and $P_0^Y$ denote the marginal distributions of $X$ and
$Y$ under $P_0$, respectively. The crucial problem is that the
conditional distribution of $Y\vert X=x$ is only defined on
$\Din$. However, the Markov kernel $Q_0$ provides a well-defined
notion of conditioning on all of $\mathcal{X}$, which will be
important for extrapolation to make sense. As we discuss in
Section~\ref{sec:causality} it is also possible to define
extrapolation using a causal model, as the resulting interventional
conditionals, unlike the observational conditionals, immediately
extend to the entire domain $\mathcal{X}$. We nevertheless choose to
avoid the overhead of a causal model and instead use the Markov kernel
$Q_0$ which is a well-defined albeit not necessarily fully identified
function. Our goal is now to perform inference on a conditional
function $\Phi_0:\mathcal{X}\rightarrow\mathbb{R}$ describing some
aspect of the Markov kernel $Q_0$. Here, we focus on the conditional
expectation and quantile functions, which cover a range of interesting
applications, but the theory extends to any conditional function.

\begin{definition}[Conditional expectations and quantiles]
  We call the function $\Psi_0:\mathcal{X}\rightarrow\mathbb{R}$ defined for all $x\in\mathcal{X}$ by
  \begin{equation}
    \label{eq:conditional_expectation_def}
    \Psi_0(x)\coloneqq\int_{\mathbb{R}} y Q_0(x, dy)
  \end{equation}
  the \emph{conditional
    expectation function}. Similarly, for all $\alpha\in(0,1)$ the function
  $\mathcal{T}_0^{\alpha}:\mathcal{X}\rightarrow\mathbb{R}$ defined for all $x\in\mathcal{X}$ by
  \begin{equation}
    \label{eq:conditional_quantile_def}
    \mathcal{T}_0^{\alpha}(x)\coloneqq\inf\{t\in\mathbb{R}\mid \textstyle\int_{\mathbb{R}} \mathds{1}(y\leq t) Q_0(x, dy)\geq \alpha\}.
  \end{equation}
  is called
  \emph{conditional $\alpha$-quantile function}.
\end{definition}
Both the conditional expectation and the conditional quantile are
defined on all elements of $\mathcal{X}$. However, a priori, they
might not be fully identified by $P_0$. To see this observe that,
using \eqref{eq:probability_kernel_prop}, it holds $P_0$-a.s. that
$\Psi_0(X)=\mathbb{E}[Y\vert X]$ and
$\mathcal{T}_0^{\alpha}(X)=\inf\{t\in\mathbb{R}\mid \mathbb{P}(Y\leq
t\vert X)\geq \alpha\}$. Hence -- without additional assumptions --
they are not identified at values $x\in\Dout$.  For a conditional
function $\Phi_0$ of interest, we therefore distinguish two types of
statistical inference: (i) \emph{interpolation}, which is inference on
$\Phi_0(x)$ for $x\in\Din$ and (ii) \emph{extrapolation}, which is
inference on $\Phi_0(x)$ for $x\in\Dout$. The distinguishing feature
between interpolation and extrapolation is that interpolation is
feasible with conventional nonparametric assumptions (e.g., smoothness
or shape constraints), while extrapolation is impossible without
additional extrapolation assumptions. To render extrapolation
feasible, we first specify extrapolation assumptions that are
reasonable in practice. The following example illustrates two types of
extrapolation assumptions -- parametric and periodic -- that are
widely used.

\begin{example}[Parametric and periodic extrapolation]
  \label{ex:simple_extra}
  The most common approach to making extrapolation meaningful is to
  assume that the conditional function of interest is parametric. For
  example, one could assume a linear structural equation model (SEM)
  given by
  \begin{equation*}
    Y = \theta^{\top} X + \varepsilon
    \quad\text{and}\quad \varepsilon\independent X,
  \end{equation*}
  where $\theta\in\mathbb{R}^d$ and $\varepsilon\sim\mu$ for some
  distribution $\mu$. This model implies a Markov kernel $Q_0$ that
  satisfies for all $x\in\mathcal{X}$ and all
  $B\in\mathcal{B}(\mathbb{R})$ that
  $Q_0(x, B)=\mu(B-\theta^{\top}x)$. The conditional expectation is
  then simply $\Psi_0:x\mapsto\theta^{\top}x$ and hence fully
  identified by the parameter $\theta$ which is identified as long as
  $P_0$ is a non-degenerate distribution.

  A further approach is to assume a periodic extrapolation model. For
  example, again using a SEM model, one could assume a model of the
  form
  \begin{equation*}
    Y = f(X) + \varepsilon
    \quad\text{and}\quad \varepsilon \independent X,
  \end{equation*}
  where $f$ is a measurable function satisfying for all $x\in[0,1)$
  and all $k\in\mathbb{Z}$ that $f(x)=f(x+k)$ and $\varepsilon\sim\mu$
  for some distribution $\mu$. The Markov kernel $Q_0$ implied by this
  model is given for all $x\in\mathcal{X}$ and
  $B\in\mathcal{B}(\mathbb{R})$ by $Q_0(x, B)=\mu(B-f(x))$. The
  conditional $\alpha$-quantile $\mathcal{T}^{\alpha}_0$ is then given
  for all $x\in\mathcal{X}$ by
  $\mathcal{T}^{\alpha}_0(x)=f(x) + \inf\{t\in\mathbb{R}\mid
  \mathbb{P}(\varepsilon\leq t)\geq \alpha\}$ which is identifiable as
  long as $f$ is identifiable from $P_0$. Since $f$ is periodic,
  identifiablility of $f$ again does not require $X$ to have full
  support on $\mathcal{X}$.
\end{example}
Both in the parametric and the periodic example the extrapolation
assumptions significantly constrain the data-generating distribution
$P_0$, so that conventional nonparametric estimates no longer
apply. We instead propose to place a smoothness based extrapolation
assumption directly on the conditional function of interest. This can
then be combined with existing nonparametic approaches. More
specifically, we assume that the conditional function of interest
$\Phi_{0}$ is $q$-th derivative dominated by itself.
\begin{definition}[$q$-th derivative extrapolating]
\label{def:extrapolating_model}
A conditional function $\Phi_0:\mathcal{X}\rightarrow\mathbb{R}$ is
called \emph{$q$-th derivative extrapolating} if
$\Phi_0\in C^q(\mathcal{X})$ and
  \begin{equation*} 
    \Phi_0\triangleleft_{\Din}^q\Phi_0.
  \end{equation*}
\end{definition}
This condition depends both on the conditional distribution of $Y$
given $X$ via the conditional function and on the marginal
distribution of $X$ via the support $\Din$. Importantly, whenever
$\Din=\mathcal{X}$, it only requires that the conditional function is
$q$-times continuously differentiable -- a common nonparametric
assumption. Hence, whenever our method performs an interpolation task,
the usual nonparametric smoothness assumption suffices. We note that
our methodology does not require knowing whether a task is inter- or
extrapolating and the method works in an automatic fashion. Moreover,
if we are extrapolating, the condition constrains the behavior of the
conditional function on $\Dout$ sufficiently much to provide
meaningful inference (albeit with added extrapolation uncertainty).
One can also make a high-level connection to the (empirical) Bayes
perspective. To see this, consider a prior over the conditional
functions with support within the class of functions with bounded
$q$-th derivatives everywhere with upper and lower bounds
$C_{\operatorname{upp}}$ and $C_{\operatorname{low}}$,
respectively. The posterior given data from the model would then have
the same support. If the bounds are a priori unknown, one could use
empirical Bayes and estimate the prior from data. In our context, this
means that we would estimate the bounds of the $q$-th derivatives from
data $C_{\operatorname{upp}}$ and $C_{\operatorname{low}}$ and any
posterior update would then preserve such upper and lower bounds, as
required in Definition \ref{def:extrapolating_model}.  Essentially,
assuming $\Phi_0$ is $q$-th derivative extrapolating, therefore
corresponds to assuming that nothing unexpected happens with the
$q$-th derivatives outside of the observed data range.

Under the assumption that $\Phi_{0}$ is $q$-th derivative
extrapolating, it is possible to bound $\Phi_{0}$ on $\Dout$ by only
knowing its values on $\Din$. More specifically, assuming that
$\Phi_0$ is $q$-th derivative extrapolating
Theorem~\ref{thm:taylor_extrapolation_bounds} ensures that $\Phi_0$
lies in the set of \emph{feasible conditional functions} defined by
\begin{equation}
    \label{eq:feasible_functions}
    \mathcal{F}_{\Phi_0}\coloneqq\left\{\phi\in C^0(\mathcal{X})\,\big\vert\,
    \forall x\in\mathcal{X}:\, B_{\Phi_{0},\Din}^{\operatorname{lo}}(x)\leq \phi(x)\leq
  B_{\Phi_{0},\Din}^{\operatorname{up}}(x)\right\},
\end{equation}
which is identifiable from $P_0$. This insight has useful implications
for a variety of statistical inference tasks. In the following
sections, we discuss three specific applications: (i) Out-of-support
prediction, (ii) extrapolation-aware uncertainty quantification, and
(iii) quantifying extrapolation. The main idea is to perform the
inference on the extrapolation bounds instead of on the conditional
function directly, leading to extrapolation-aware inference.

\subsection{Out-of-support prediction}\label{sec:out-of-support-prediction}

Consider a setting in which $n$ i.i.d.\ observations
$(X_1,Y_1),\ldots,(X_n,Y_n)$ from $P_0$ are observed and we want to
learn a prediction function $\hat{f}$ to predict the response $Y$ for
$X=x$ for some $x\in\mathcal{X}$. A standard nonparametric approach is
to estimate the conditional expectation $\Psi_{0}$ and use it as the
prediction function $\hat{f}$. The idea behind this is that under mild
regularity conditions $\Phi_0$ minimizes the mean squared prediction
error under $P_0$ and hence performs optimal (in the mean squared
sense). However, since $\Phi_0$ is only identified on $\supp{X}$
(without extrapolation assumptions), this guarantee is no longer valid
when considering predictions at points $x\in\Dout$. To avoid such
problems, we need to make explicit extrapolation assumptions.

Using the terminology of the previous section, assume that
$\Psi_{0}$ is $q$-th derivative extrapolating. Then, by
Theorem~\ref{thm:taylor_extrapolation_bounds}, it holds for all
$x\in\mathcal{X}$ that
\begin{equation}
  \label{eq:psi_bound_oos_pred}
  B_{\Psi_{0},\Din}^{\operatorname{lo}}(x)\leq \Psi_{0}(x)\leq
  B_{\Psi_{0},\Din}^{\operatorname{up}}(x).
\end{equation}
Since the extrapolation bounds are identified from $P_0$, this allows
us to bound the conditional expectation on all of $\mathcal{X}$. We
can use this for prediction by constructing a point estimate based on
the extrapolation bounds. As the conditional expectation is not
necessarily identified everywhere on $\mathcal{X}$, the performance of
the point estimate may depend on the underlying Markov kernel
$Q_0$. We therefore attempt to find a prediction function that is
worst-case optimal in the sense that it minimizes the worst-case
mean-squared prediction error among all Markov kernels $Q$ that are
equal to $Q_0$ on $\Din$ and for which
$x\mapsto \int_{\mathbb{R}}yQ(x, dy)$ is $q$-th derivative
extrapolating. Formally, we define the set of feasible Markov kernels
by
\begin{align}
  \label{eq:potential_kernels}
  \mathcal{Q}_0\coloneqq\big\{Q:\mathcal{X}\times\mathfrak{B}(\mathbb{R})\rightarrow
  [0,1] \text{ Markov kernel}\,\big\vert\, &\forall x\in\Din:\, Q(x,
                                             \cdot)= Q_0(x,\cdot) \\
                                           &\text{ and
                                             $x\mapsto\textstyle\int_{\mathbb{R}}yQ(x,
                                             dy)$ is $q$-th
                                             deriv. extr.}\big\}.\nonumber
\end{align}
For all $Q\in\mathcal{Q}_0$ and all $x\in\mathcal{X}$ we denote by
$Y_x$ a random variable with distribution $Q(x,\cdot)$. Using this
notation, our goal is to find a prediction function $\hat{f}$, such
that for all $x\in\mathcal{X}$, it minimizes the worst-case mean
squared prediction error
\begin{equation*} \sup_{Q\in\mathcal{Q}_0}\mathbb{E}_{Q(x,
\cdot)}[(Y_x-\hat{f}(x))^2].
\end{equation*} This guarantees that the prediction function $\hat{f}$
also performs well in cases where the true underlying Markov kernel
$Q_0$ is adversarial.

\begin{proposition}[Worst-case optimal prediction under extrapolation]
  \label{thm:worst_case_prediction}
  Let $\mathcal{X}$ be compact, assume the Markov kernel $Q_0$ is such
  that the conditional expectation $\Psi_{0}$ is $q$-th derivative
  extrapolating. Then, $f^*:\mathcal{X}\rightarrow\mathbb{R}$ defined
  for all $x\in\mathcal{X}$ by
  \begin{equation}
    \label{eq:point_estimate_oos}
    f^*(x)\coloneqq\frac{1}{2}\left(B^{\operatorname{lo}}_{\Psi_{0},\Din}(x)+B^{\operatorname{up}}_{\Psi_{0}, \Din}(x)\right)
  \end{equation}
  satisfies
  \begin{equation}
    \label{eq:worst_case_optimal_result}
    \inf_{f\in C^0(\mathcal{X})}\sup_{Q\in\mathcal{Q}_0}\mathbb{E}_{Q(x,
      \cdot)}[(Y_x-f(x))^2]=\sup_{Q\in\mathcal{Q}_0}\mathbb{E}_{Q(x,\cdot)}[(Y_x-f^*(x))^2].
  \end{equation}
\end{proposition}

A proof is given in Supplementary
material~\ref{proof:worst_case_prediction}. The type of guarantee in
Proposition~\ref{thm:worst_case_prediction} is common in the field of
distribution generalization. In distribution generalization one
assumes the observed (training) data was generated under a
distribution $P_0$ but wants to predict $Y$ under a new (potentially
different) test distribution $\Ptest$. Without further assumptions on
$\Ptest$ this is clearly impossible. A well-established assumption is
the covariate-shift assumption \citep{sugiyama2007covariate}, which
assumes that the conditional expectation under $P_0$ and $\Ptest$
remains fixed. For this assumption to be sufficient for generalization
without additional extrapolation assumptions one however requires that
\begin{equation*}
    \supp{\Ptest^X}\subseteq \supp{P^X_0}=\Din,
\end{equation*}
otherwise predictions can be arbitrarily wrong outside of
$\Din$. Proposition~\ref{thm:worst_case_prediction} extends this by
allowing the test distribution to have arbitrary support as long as
the conditional $Y|X$ under $\Ptest$ is generated by a Markov kernel
$Q\in\mathcal{Q}_0$.

\subsection{Extrapolation-aware uncertainty
  quantification}\label{sec:CIsandPIs}

Quantifying uncertainty is important in real-world applications. In
this section we consider two important approaches for uncertainty
quantification when predicting a real-valued response $Y$ from
predictors $X$; confidence intervals, which quantify the uncertainty
in the estimation of the regression function and prediction intervals
which quantify the uncertainty in the prediction itself. Existing
methods for nonparametric regression, only apply within $\Din$ and
therefore cannot provide coverage guarantees when extrapolating. This
is particularly troubling in applications where it is difficult or
even impossible to ensure whether and to what degree extrapolation
occurs. The extrapolation assumptions discussed above can provide a
solution. We now show that the extrapolation bounds derived above can
be combined with existing nonparametric approaches to construct both
extrapolation-aware confidence and prediction intervals.

We begin with confidence intervals for the conditional expectation
function $\Psi_0$ -- the same can be done for other conditional
functions $\Phi_0$. Consider a setting in which we want to estimate
the conditional expectation $\Psi_{0}$ from i.i.d.\ observations
$(X_1,Y_1),\ldots,(X_n,Y_n)\sim P_0$. Assuming that $\Psi_{0}$ is
$q$-th derivative extrapolating, it follows from
Theorem~\ref{thm:worst_case_prediction} that
\begin{equation}
  \label{eq:psi_bounds}
  \Psi_{0}(x)\in\left[B^{\operatorname{lo}}_{\Psi_{0},
      \Din}(x),\, B^{\operatorname{up}}_{\Psi_{0}, \Din}(x)\right].
\end{equation}
Since the extrapolation bounds are identifiable from $P_0$, we can use
any procedure that produces asymptotically valid confidence intervals
for both the lower and upper extrapolation bound and combine them to
get extrapolation bounds that are valid on all of $\mathcal{X}$.
\begin{proposition}[Extrapolation-aware confidence interval coverage]
  \label{thm:confidence_intervals}
  Fix $\alpha\in(0,1)$ and assume the conditional expectation function
  $\Psi_{0}$ is $q$-th derivative extrapolating. For both
  $\star\in\{\operatorname{lo}, \operatorname{up}\}$, all $x\in\mathcal{X}$ and all
  $\gamma\in(0,1)$ let $\widehat{G}^{\star}_{n}(\gamma, x)$ be an
  estimation procedure based on $n$ i.i.d.\ observations from $P_0$
  satisfying
  $\lim_{n\rightarrow\infty}\mathbb{P}\left(B^{\star}_{\Psi_0,
      \Din}(x)\leq\widehat{G}^{\star}_{n}(\gamma, x)\right)=\gamma$.
  Define for all $x\in\mathcal{X}$ the confidence intervals
  \begin{equation*}
    \widehat{\operatorname{C}}_{n;\alpha}^{\operatorname{conf}}(x)\coloneqq \left[\widehat{G}_{n}^{\operatorname{lo}}(\tfrac{\alpha}{2},x),\, \widehat{G}_{n}^{\operatorname{up}}(1-\tfrac{\alpha}{2},x)\right].
  \end{equation*}
  Then, it holds for all $x\in\mathcal{X}$ that
  \begin{equation*}
    \liminf_{n\rightarrow\infty}\mathbb{P}(\Psi_0(x)\in\widehat{\operatorname{C}}_{n;
      \alpha}^{\operatorname{conf}}(x))\geq 1-\alpha.
  \end{equation*}
\end{proposition}
The proof follows from a direct application of \eqref{eq:psi_bounds}
and is provided in Supplementary
material~\ref{proof:confidence_intervals} for completeness. The
estimates $\widehat{G}_n^{\star}(\gamma, x)$ in
Proposition~\ref{thm:confidence_intervals} can be constructed by a
combination of an estimate of $B^{\star}_{\Psi_0,\Din}(x)$ (discussed
in Section~\ref{sec:estimation}) and a method to estimate the quantile
of the estimator distribution. This could for example be a
nonparametric bootstrap procedure, e.g., the percentile bootstrap
\citep{efron1981nonparametric}.  Since the lower and upper
extrapolation bounds are equal to $\Psi_{0}$ on $\Din$, the confidence
intervals for conditional expectation resulting from such a procedure
are extrapolation-aware in the sense that they become larger on
$\Dout$ while remaining tight on $\Din$, see
Example~\ref{ex:bootstrap}.

\begin{example}[Extrapolation-aware bootstrap confidence intervals]
  \label{ex:bootstrap}
  In Figure~\ref{fig:linear_nonlinear_xtrapolation} (left), the data
  come from a one-dimensional linear model. In this case, assuming a
  linear model is sufficient for extrapolation and a linear regression
  can be used to predict values outside of the support. However, in
  Figure~\ref{fig:linear_nonlinear_xtrapolation} (right) the data come
  from a model with a nonlinear conditional expectation. Even though a
  linear regression leads to a good fit, it does not extrapolate
  well. If one instead assumes that the model is first order
  $\Psi$-extrapolating, which is satisfied in both examples,
  extrapolation-aware confidence intervals based on the extrapolation
  bounds and a nonparametric percentile bootstrap (blue, solid) using
  a random forest estimate of $\Psi_{P_0}$ (red, solid) are able to
  preserve coverage also while extrapolating. In particular, the
  confidence intervals are extrapolation-aware: they detect the
  linearity in the linear setting resulting in tight bounds and
  capture the indeterminacy in the behavior of $\Psi_{P_0}$ outside of
  the support in the nonlinear setting.
\end{example}

\begin{figure}
  \centering
  \includegraphics{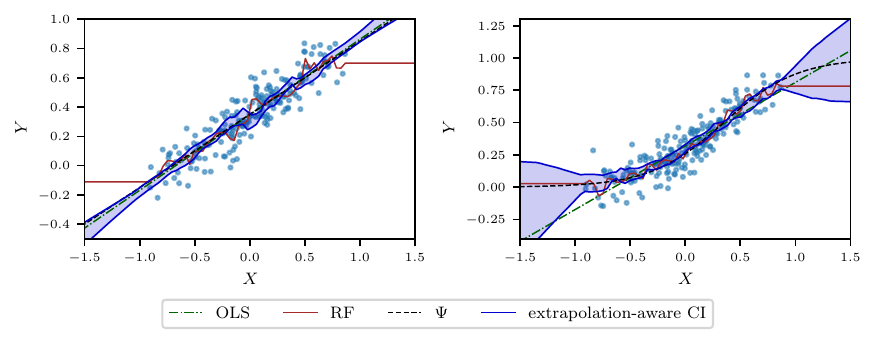}
  \caption{Linear (left) and nonlinear (right) conditional expectation
    models. In both cases, OLS regression (green, dotdashed) fits the
    data well ($\operatorname{RMSE}=0.31$ and
    $\operatorname{RMSE}=0.30$ respectively). However, the OLS only
    extrapolates well for the linear setting, while it is misleading
    in the nonlinear setting. Extrapolation-aware confidence intervals
    constructed using a percentile bootstrap and estimates of the
    extrapolation bounds based on random forest (RF) are able to
    capture the extrapolating behavior in both cases. In the linear
    setting it closely matches the OLS prediction with tight bounds,
    while in the nonlinear setting it captures the uncertainty outside
    of the data support.}
  \label{fig:linear_nonlinear_xtrapolation}
\end{figure}
Next, we consider prediction intervals. Now instead of estimating the
conditional expectation function, we want to estimate an interval that
contains $Y\vert X=x$ with a pre-specified probability $1-\alpha$. One
existing nonparametric approach is to estimate quantiles of the
conditional distribution of $Y$ and use them to construct prediction
intervals, that is, estimate the prediction interval
$[\mathcal{T}^{\alpha/2}_{0}(x), \mathcal{T}^{1-\alpha/2}_{0}(x)]$. We
can modify this approach to be extrapolation-aware as follows.
\begin{proposition}[Extrapolation-aware prediction interval coverage]
  \label{thm:prediction_intervals}
  Fix $\alpha\in(0,1)$ and assume the Markov kernel $Q_0$ is such that
  the conditional quantiles $\mathcal{T}^{\alpha/2}_{0}$ and
  $\mathcal{T}^{1-\alpha/2}_{0}$ are both $q$-th derivative
  extrapolating. Define for all $x\in\mathcal{X}$ the prediction
  intervals
  \begin{equation*}
    \operatorname{C}_{\alpha}^{\operatorname{pred}}(x)\coloneqq \left[B^{\operatorname{lo}}_{\mathcal{T}_{0}^{\alpha/2},\Din}(x),\, B^{\operatorname{up}}_{\mathcal{T}_{0}^{1-\alpha/2},\Din}(x)\right].
  \end{equation*}
  Then, it holds for all $x\in\mathcal{X}$ that
  \begin{equation*}
    \mathbb{P}_{Q_0(x, \cdot)}(Y_x\in\operatorname{C}_{\alpha}^{\operatorname{pred}}(x))\geq 1-\alpha.
  \end{equation*}
\end{proposition}
A proof is given in Supplementary
material~\ref{proof:prediction_intervals}.  In practice one needs to
estimate the extrapolation bounds resulting in prediction intervals
being random. Consistency of the extrapolation bound estimates is
discussed in Section \ref{sec:estimation}, which then leads to
(pointwise) asymptotically valid prediction intervals (see
Corollary~\ref{thm:est_prediction_intervals}). The extrapolation-aware
confidence and prediction intervals can both be conservative for
points in $\Dout$. We see this as a positive feature, since it
protects against potentially misleading conclusions due to accidental
extrapolation.

\subsection{Quantifying extrapolation}\label{sec:extrapolation_score}

It is often difficult to quantify the level of extrapolation in a
meaningful way. Simply measuring the minimal distance to observed
samples does not always provide a good indicator of whether
extrapolation is problematic for a given point $x$, in particular, if
$X$ is multi-dimensional. We can use extrapolation bounds to quantify
the level of extrapolation in a way that takes both the target of
inference and the extrapolation assumptions into account. To this end,
we use that (by Theorem~\ref{thm:taylor_extrapolation_bounds}) if the
conditional function of interest $\Phi_0$ is $q$-th derivative
extrapolating, it holds 
that $\Phi_0(x)$ is
identifiable for all $x\in\mathcal{X}$ from $P_0$ if
\begin{equation*}
    B_{\Phi_{0},\Din}^{\operatorname{up}}(x)-B_{\Phi_{0},\Din}^{\operatorname{lo}}(x)=0.
\end{equation*}
Therefore, we can use the difference between the upper and lower
extrapolation bounds as a quantification of the level of
extrapolation. Quantifying extrapolation in this way, explicitly takes
the target of inference $\Phi_0$ into account. For example, if
$\Phi_0$ only changes in the first coordinate but is constant in the
remaining coordinates, the extrapolation bounds do not changes as long
as the first coordinate is kept fixed. Hence, the bounds can overlap
even if $x$ is far away from any observed samples in Euclidean
distance.

Depending on the conditional function under consideration, we propose
to use different variations of this score. For example, when
estimating a conditional expectation $\Psi_0$, we suggest using the
extrapolation score $S:\mathcal{X}\rightarrow [0,\infty)$ defined for
all $x\in\mathcal{X}$ by
\begin{equation*}
  S(x)\coloneqq \frac{B_{\Psi_{0},\Din}^{\operatorname{up}}(x)-B_{\Psi_{0},\Din}^{\operatorname{lo}}(x)}{\sqrt{\mathbb{E}[(Y-\Psi_{0}(X))^2]}}.
\end{equation*}
At a point $x\in\mathcal{X}$ with $S(x)=0$ there is no extrapolation,
while for $S(x)>0$ there is extrapolation. The normalization by the
residual standard deviation allows us to interpret the score values
more directly: if $S(x)=1$ the extrapolation uncertainty in the
conditional expectation function $\Psi_{0}$ is equal to the standard
deviation of the residual noise. In particular for scores greater than
one, we should be cautious when interpreting any point estimates as
the error due to unidentifiablity of $\Phi_0(x)$ may be larger than
the noise level.

\subsection{A causal perspective on
  extrapolation}\label{sec:causality}

When extrapolating a conditional function $\Phi_0$, we are implicitly
assuming a two step generative model that first selects $X$ and then
samples $Y$ according to the distribution $Q_0(X, \cdot)$ specified by
the Markov kernel. This two step generative model, does not need to
correspond to an underlying mechanism and hence needs to be
interpreted carefully. Instead of plain conditioning, one therefore
might be interested in quantities with a causal interpretation, where
one not only conditions but actively sets $X$ (or a part of it) to
specific values. Causal models provide a rigorous mathematical
framework to define such quantities using the notion of
interventions. Once a causal quantity of interest has been defined the
nonparametric approach to causal inference consists in formalizing
precise causal assumptions under which the causal quantity can be
expressed as a function of $P_0$ and hence estimated using
observational statistical methods. Since the proposed extrapolation
framework extends the conventional nonparametric model, it also
immediately extends this causal inference approach and allows us to
reason about causal quantities outside the observed support.

\begin{remark}[Causal extrapolation]
  \label{rmk:causal_extrapolation}
  In causal inference, the term extrapolation is sometimes used to
  refer to the task of generalizing from the observational to a
  previously unseen interventional distribution. To avoid confusion,
  we call such inference tasks \emph{causal extrapolation}. Causal
  extrapolation is in general different from the notion of
  extrapolation considered here because it does not necessarily (and
  in fact mostly does not) correspond to evaluating a conditional
  function outside of $\supp{X}$.
\end{remark}

We now formally discuss how to combine the proposed extrapolation
framework with causal models based on a treatment-response example. We
use structural causal models (SCMs) \citep{pearl2009causality} for
this as it allows us to specify the function classes more naturally,
but it is easy to transfer the ideas to other causal models e.g., the
potential outcome model \citep{rubin2005causal}. Let $M_0$ be an SCM
over $(X, Y)\in\mathcal{X}\times\mathbb{R}$, assume that the
covariates can be divided into $X=(T, W)$, where $T\in\mathcal{T}$ are
treatment variables and $W\in\mathcal{W}$ pre-treatment covariates,
and let $M_0$ be given by

\begin{minipage}{0.5\textwidth}
  \begin{equation*}
    M_0:\,
    \begin{cases}
      &W \leftarrow\varepsilon_W\\
      &T \leftarrow h(W, \varepsilon_T)\\
      &Y \leftarrow g(W, T) +  \varepsilon_Y.
    \end{cases}
  \end{equation*}
\end{minipage}%
\begin{minipage}{0.5\textwidth}
  \vspace{1em}
  \begin{tikzpicture}[scale=1.3]
    \tikzset{round/.style={circle, draw, fill=white!90!black,
        minimum width=2em}}
    \node[round] at (1,0) (Y){$Y$};
    \node[round] at (-1,0) (T){$T$};
    \node[round] at (0,1) (W){$W$};
    \path[-Latex,draw,thick] (T) edge (Y);
    \path[-Latex,draw,thick] (W) edge (Y);
    \path[-Latex,draw,thick] (W) edge (T);
  \end{tikzpicture}
  \vspace{1em}
\end{minipage}
The SCM $M_0$ induces an observed distribution $P_{0}$ over
$(T, W, Y)$ and for all $t\in\mathcal{T}$ the interventional
distributions $P_{\operatorname{do}(T=t)}$ over $(T, W, Y)$
corresponding to the intervention $\operatorname{do}(T=t)$ that
assigns treatment $t$. Given the SCM $M_0$ we can therefore define the
causal conditional functions we are interested in. For example, if $T$
is continuous, we could consider the dose-response curve
$\Phi_{0}^{\operatorname{DRC}}:\mathcal{T}\rightarrow\mathbb{R}$, also
called the average treatment function, defined for all
$t\in\mathcal{T}$ by
\begin{equation*}
  \Phi_{0}^{\operatorname{DRC}}(t)\coloneqq\mathbb{E}_{\operatorname{do}(T=t)}[Y],
\end{equation*}
where the subscript $\operatorname{do}(T=t)$ denotes that the
expectation is taken with respect the distribution
$P_{\operatorname{do}(T=t)}$.  Similarly, if $T$ is
binary, we could consider the conditional average treatment effect
$\Phi_{0}^{\operatorname{CATE}}:\mathcal{W}\rightarrow\mathbb{R}$
defined for all $w\in\mathcal{W}$ by
\begin{equation*}
  \Phi_{0}^{\operatorname{CATE}}(w)\coloneqq\mathbb{E}_{\operatorname{do}(T=1)}[Y\vert
  W=w]
  -\mathbb{E}_{\operatorname{do}(T=0)}[Y\vert W=w].
\end{equation*}
The conditional average treatment effect corresponds to the example
mentioned in the introduction (with $W$ taken to be age). Using
well-established identification procedures from causal inference
(e.g., the g-computation-formula \citep{robins1986new}) it directly
follows from the causal assumptions encoded in $M_0$ that both causal
target quantities can be expressed as functions of $P_0$. More
specifically, it holds for all $t\in\mathcal{T}$ that
\begin{equation*}
  \Phi_{0}^{\operatorname{DRC}}(t)=\mathbb{E}[\mathbb{E}[Y\vert T=t, W]]
\end{equation*}
and for all $w\in\mathcal{W}$ that
\begin{equation*}
  \Phi_{0}^{\operatorname{CATE}}(w)=\mathbb{E}[Y\vert T=1, W=w] - \mathbb{E}[Y\vert T=0, W=w].
\end{equation*}
In both cases, we thus reduced the causal quantities to conditional
functions based only on the observed distribution $P_0$. This
reduction from the interventional to the observed distribution is what
is sometimes called causal extrapolation (see
Remark~\ref{rmk:causal_extrapolation}). At least in the nonparametric
approach to causal inference, the functions
$\Phi_0^{\operatorname{DRC}}$ and $\Phi_{0}^{\operatorname{CATE}}$
are, however, only identified on $\supp{T}$ and on $\supp{W}$,
respectively. Therefore, if one is interested in genuine extrapolation,
additional assumptions, as discussed in
Section~\ref{sec:statistical_extrapolation}, are required. We propose
to use our extrapolation framework to perform inference on the causal
quantities outside of their respective observed supports.  Since the
framework applies to arbitrary conditional functions (see
Section~\ref{sec:categorical_covariates} for how to account for
categorical variables), it also applies here if we are willing to
assume that the causal quantity of interest (e.g.,
$\Phi_0^{\operatorname{DRC}}$ or $\Phi_{0}^{\operatorname{CATE}}$) is
$q$-th derivative extrapolating.

\section{Estimating extrapolation bounds}\label{sec:estimation}

Assume $n$ i.i.d.\ observations $(X_1,Y_1),\ldots,(X_n, Y_n)$ from the
distribution $P_0$ are observed and we want to estimate the
extrapolation bounds $B_{\Phi_{0}, \Din}^{\operatorname{lo}}$ and
$B_{\Phi_{0}, \Din}^{\operatorname{up}}$. By definition the
extrapolation bounds are completely identified by $P_0$ since they can
be exactly computed from the conditional function $\Phi_{0}$ on the
set $\Din$. A natural plugin estimator is therefore given by first
estimating $\Phi_{0}$ with a $q$-times differentiable function
$\widehat{\Phi}$ and then directly evaluating the bounds. In practice
this can be difficult for two reasons: (i) Directly computing the
extrapolation bounds involves two optimizations over the potentially
unknown set $\Din$ and (ii) existing nonparametric estimation
procedures for estimating $\Phi_{0}$ might not result in $q$-times
differentiable functions and even if they do the derivatives may be
ill-behaved.

We consider these two problems separately. First, in
Section~\ref{sec:consistency}, we start by assuming access to a
$q$-times differentiable estimate $\widehat{\Phi}_n$ that approximates
$\Phi_{0}$ and its derivatives sufficiently well. In that case, we
show that the extrapolation bounds can be estimated consistently, by
performing the two optimizations over the sample points
$X_1,\ldots,X_n$ only. Second, in Section~\ref{sec:forest_locpol}, we
propose a procedure based on random forests and local polynomials that
uses only
$(X_1,\widehat{\Phi}_n(X_1)),\ldots,(X_n,\widehat{\Phi}_n(X_n))$ to
estimate directional derivatives
$D_v^k\widehat{\Phi}_n(X_1),\ldots D_v^k\widehat{\Phi}_n(X_1)$ for
arbitrary directions $v$ and orders $k$. Finally, in
Section~\ref{sec:xtrapolation_estimator}, we combine the procedures
from Sections~\ref{sec:consistency} and~\ref{sec:forest_locpol} in a
computationally efficient way.

\subsection{Extrapolation bounds from differentiable
  estimates}\label{sec:consistency}

Let $\widehat{\Phi}_n$ be a $q$-times differentiable estimate of
$\Phi_{0}$ based on the data $(X_1, Y_1),\ldots,(X_n,Y_n)$. Then, at
an arbitrary point $x\in\mathcal{X}$, we propose to estimate the
extrapolation bounds by plugging the differentiable estimate
$\widehat{\Phi}_n$ into the definition and optimizing only over the
observed samples $X_1,\ldots,X_n$. More formally, define
\begin{equation}
  \label{eq:plugin_estimates}
  \widehat{B}^{\operatorname{lo}}_n(x)\coloneqq B^{\operatorname{lo}}_{\widehat{\Phi}_n,\{X_1,\ldots,X_n\}}(x)
  \quad\text{and}\quad
  \widehat{B}^{\operatorname{up}}_n(x)\coloneqq B^{\operatorname{up}}_{\widehat{\Phi}_n,\{X_1,\ldots,X_n\}}(x).
\end{equation}
Using multi-index notation, the lower extrapolation bound estimate can
be expressed as
\begin{equation} \label{eq:vanilla_lower_bound_estimator}
  \widehat{B}^{\operatorname{lo}}_n(x)\coloneqq
  \max_{i\in\{1,\ldots,n\}}\left(\sum_{\ell=0}^{q-1}\sum_{|\balpha|=\ell}\partial^{\balpha}\widehat{\Phi}_n(X_i)\tfrac{(x-X_i)^{\balpha}}{\balpha!}+\min_{k\in\{1,\dots,n\}}\sum_{|\balpha|=q}\partial^{\balpha}\widehat{\Phi}_n(X_k)\tfrac{(x-X_i)^{\balpha}}{\balpha!}\right).
\end{equation}
Similarly, the upper extrapolation bound estimate can be expressed as
\begin{equation}
  \label{eq:vanilla_upper_bound_estimator}
  \widehat{B}^{\operatorname{up}}_n(x)\coloneqq
  \min_{i\in\{1,\ldots,n\}}\left(\sum_{\ell=0}^{q-1}\sum_{|\balpha|=\ell}\partial^{\balpha}\widehat{\Phi}_n(X_i)\tfrac{(x-X_i)^{\balpha}}{\balpha!}+\max_{k\in\{1,\dots,n\}}\sum_{|\balpha|=q}\partial^{\balpha}\widehat{\Phi}_n(X_k)\tfrac{(x-X_i)^{\balpha}}{\balpha!}\right).
\end{equation}
From these expressions, it can be seen that, instead of evaluating all
possible directional derivatives, the estimates can be computed by
only evaluating $\partial^{\balpha}\widehat{\Phi}_n(X_i)$ once at
every observation $X_i$ and every partial derivative
$\partial^{\balpha}$. If $q=1$, for example, this means it is
sufficient to evaluate all $\partial^j\widehat{\Phi}_n(X_i)$
corresponding to $n d$ evaluations instead of $n^2$ evaluations for
$D_{\overline{v}(x,X_j)}\widehat{\Phi}(X_i)$. This becomes
particularly beneficial if the bounds are evaluated at many target
points $x\in\mathcal{X}$. As shown in the following theorem, the
estimates are consistent if $\widehat{\Phi}_n$ and all partial
derivatives $\partial^{\balpha}\widehat{\Phi}_n$ up to order $q$ are
uniformly consistent on $\Din$.
\begin{theorem}[Consistency of extrapolation bound estimates]
  \label{thm:consistency}
  Assume $\mathcal{X}$ is compact and let
  $\Phi_0:\mathcal{X}\rightarrow\mathbb{R}$ be a conditional function
  satisfying $\Phi_{0}\in C^{q+1}(\mathcal{X})$. For all
  $n\in\mathbb{N}$, let $\widehat{\Phi}_n$ be a $q$-times
  differentiable estimate of $\Phi_{0}$ based on $n$ i.i.d.\
  observations $(X_1, Y_1),\ldots,(X_n,Y_n)\sim P_0$ satisfying for
  all $\balpha\in\mathbb{N}^d$ with $|\balpha|\leq q$ that
  \begin{equation*}
    \sup_{x\in\Din}\left|\widehat{\Phi}_n(x)-\Phi_{0}(x)\right|\overset{P_0}{\longrightarrow}0
    \quad\text{and}\quad
    \sup_{x\in\Din}\left|\partial^{\balpha}\widehat{\Phi}_n(x)-\partial^{\balpha}\Phi_{0}(x)\right|\overset{P_0}{\longrightarrow}0
    \quad\text{as }n\rightarrow\infty.
  \end{equation*}
  Additionally, assume
  $\Lambda_n\coloneqq\sup_{z\in\Din}\min_{k\in\{1,\ldots,n\}}\|X_k-z\|_2\overset{P_0}{\longrightarrow}0$
  as $n$ goes to infinity and for all $n\in\mathbb{N}$ and all
  $x\in\mathcal{X}$, let $\widehat{B}^{\operatorname{lo}}_n(x)$ and
  $\widehat{B}^{\operatorname{up}}_n(x)$ be defined
  in~\eqref{eq:plugin_estimates}. Then, for all $x\in\mathcal{X}$, it
  holds that
  \begin{equation*}
    \left|B^{\operatorname{lo}}_{\Phi_{0},\Din}(x)-\widehat{B}^{\operatorname{lo}}_n(x)\right|\overset{P_0}{\longrightarrow}0
    \quad\text{and}\quad
    \left|B^{\operatorname{up}}_{\Phi_{0},\Din}(x)-\widehat{B}^{\operatorname{up}}_n(x)\right|\overset{P_0}{\longrightarrow}0
    \quad\text{as }n\rightarrow\infty.
  \end{equation*}
\end{theorem}
A proof is given in Supplementary material~\ref{proof:consistency}. An
immediate implication of the consistency in Theorem
\ref{thm:consistency} is asymptotic validity of the prediction
intervals discussed in Section~\ref{sec:CIsandPIs}.
\begin{corollary}
  \label{thm:est_prediction_intervals}
  Assume $\mathcal{X}$ is compact, fix $\alpha\in(0,1)$ and assume the
  Markov kernel $Q_0$ is such that the conditional quantiles
  $\mathcal{T}^{\alpha/2}_{0}$ and $\mathcal{T}^{1-\alpha/2}_{0}$ are
  both $q$-th derivative extrapolating.  Denote by
  $\widehat{B}^{\operatorname{lo}}_{n;\alpha}$ and
  $\widehat{B}^{\operatorname{up}}_{n;\alpha}$ estimates of the lower
  extrapolation bound of $\mathcal{T}_0^{\alpha/2}$ and the upper
  extrapolation bound of $\mathcal{T}_0^{1-\alpha/2}$,
  respectively. Furthermore, assume that the estimates satisfies the
  consistency in Theorem~\ref{thm:consistency} and define for all
  $x\in\mathcal{X}$ the intervals
  \begin{equation*}
    \widehat{C}_{n;\alpha}^{\mathrm{pred}}(x) \coloneqq
    \left[\widehat{B}^{\operatorname{lo}}_{n;\alpha}(x), \widehat{B}^{\operatorname{up}}_{n;\alpha}(x)\right].
  \end{equation*}
  Then, it holds for all $x\in\mathcal{X}$ that
  \begin{equation*}
    \liminf_{n \to \infty} \mathbb{P}_{Q_0(x,\cdot)}\left(Y_x\in\widehat{C}_{n;\alpha}^{\mathrm{pred}}(x)\right) \geq 1 - \alpha.
  \end{equation*}
\end{corollary}
  
Whenever a $q$-times differentiable estimator $\widehat{\Phi}_n$ is
available the above estimates for the lower and upper bounds can be
used. For conditional expectations, multiple methods have been
proposed that either provide differentiable estimates
\citep[e.g.,][]{hardle1989, mack1989, wahba1990spline} or that
estimate the corresponding derivatives separately
\citep[e.g.,][]{wang2015derivative, dai2016optimal}. These methods
are, however, generally constructed only for conditional expectations
and do not apply to other conditional functions $\Phi_0$. Furthermore,
they are targeted towards specific estimation procedures and often
only work for univariate $X$, making them inapplicable to modern
applications, where state-of-the-art performance is achieved with
nonparametric machine learning procedures. Unfortunately, those
machine learning estimates, in general, cannot be used directly as
they are either not smooth (e.g., random forests or boosted trees) or
the derivatives of the resulting estimates are ill-behaved without
additional regularization (e.g., neural networks or support vector
machines) \citep[e.g.,][]{de2013derivative}.  One possible solution is
to use procedures that start from a potentially non-differentiable
pilot estimate $\widehat{\Phi}_n$ and smooth the estimate such that
the smoothed estimate has well-behaved derivatives. Such a procedure
based on kernel-smoothing has been proposed by
\citet{klyne2023average}. We propose a related approach but instead of
directly smoothing $\widehat{\Phi}_n$ we only use the predictions
$\widehat{\Phi}_n(X_1),\ldots,\widehat{\Phi}_n(X_n)$ to estimate the
required derivatives. A related procedure was also used by
\citet{lundborg2023perturbation}.

\subsection{Estimating directional derivatives}\label{sec:forest_locpol}

We now introduce a procedure for estimating directional derivatives
based only on
\begin{equation}
  \label{eq:predicted_data}
  (X_1, \widehat{\Phi}_n(X_1)),\ldots, (X_n, \widehat{\Phi}_n(X_n)),
\end{equation}
where $\widehat{\Phi}_n$ is a potentially non-differentiable estimate
of the conditional function $\Phi_{0}$. We omit a detailed statistical
analysis of the procedure as this goes beyond the scope of this
article and instead argue heuristically and empirically (see
Section~\ref{sec:numerical_experiments}) that the proposed procedure
has several properties making it amenable to our application.

Throughout this section we fix a direction\footnote{For our proposed
  implementation of the first-order derivative estimation, we only
  need the directions $v\in\{e_1,\ldots,e_d\}$, where $e_j$ denotes
  the $j$-th unit vector.} $v\in\mathcal{B}$ and an order
$k\in\{1,\ldots,q\}$ and aim to estimate the directional derivatives
$D_v^k\Phi_{0}(X_1), \ldots, D_v^k\Phi_{0}(X_n)$ from
\eqref{eq:predicted_data}. Derivative estimation is well-known to be
statistically challenging, particularly in high-dimensions. Our
proposal aims to overcome these challenges by combining random forests
with local polynomials. Local polynomials are among the most prominent
methods used for derivative estimation. The idea is to estimate the
derivative $D_v^k\Phi_{0}(X_i)$ using a polynomial $p$ of order $q+1$
defined for all $x\in\mathcal{X}$ by
\begin{equation*}
  p(x)=\sum_{j=0}^{q+1}\beta_j((x-X_i)^{\top}v)^j.
\end{equation*}
Then, if for all $\ell\in\{1,\ldots,q+1\}$ it holds that
$\beta_{\ell} = \frac{D_{v}^{\ell}\Phi_{0}(X_i)}{k!}$, Taylor's
theorem implies that $p$ is a good approximation of $\Phi_{0}$ around
$X_i$ in the direction $v$ or, more formally, for all $h\in\mathbb{R}$
close to zero, $\Phi_{0}(X_i+hv) = p(X_i+hv)+\mathcal{O}(h^{q+2})$.
Based on this observation, we estimate coefficients
$\beta_0,\ldots,\beta_{q+1}$ such that the polynomial $p$ is a good
local approximation of $\Phi_{0}$ around $X_i$ and then use them to
estimate the directional derivative. More concretely, for a given
weight matrix $W\in\mathbb{R}^{n\times n}$, we minimize the weighted
mean squared loss
\begin{equation}
  \label{eq:local_polynomial_fit}
  \hat{\beta}(i)\coloneqq\argmin_{\beta\in\mathbb{R}^{q+1}}\sum_{\ell=1}^{n}\left(\widehat{\Phi}_n(X_\ell)-\sum_{j=0}^{q+1}\beta_j((X_{\ell}-X_i)^{\top}
    v)^{j}\right)^2 W_{i, \ell}.
\end{equation}
Then, using the estimated coefficients
$\hat{\beta}(i)$, we can estimate $D_v^k\Phi_{0}(X_i)$ as
\begin{equation}
  \label{eq:locpol_deriv_est}
  \widehat{D_v^k\Phi}_n(X_i)\coloneqq k!\, \hat{\beta}_k(i).
\end{equation}
Using local polynomials to estimate derivatives of conditional
expectations has been analyzed extensively in the literature
\citep[e.g.,][]{masry1997, de2013derivative}, however using
$Y_{\ell}$'s instead of $\widehat{\Phi}_n(X_{\ell})$'s in
\eqref{eq:local_polynomial_fit}. Most existing approaches use kernel
weights, i.e., $W_{i,\ell}=k((X_i-X_{\ell})/\sigma)$ for a kernel
function $k$ and a bandwidth $\sigma>0$. For our purposes, kernel
weights are not ideal for two reasons. Firstly, kernel weights can
perform poorly in higher dimensions and secondly, require careful
tuning of the bandwidth parameter.

To avoid these issues, we instead suggest to use weights constructed
by a random forest with a modified splitting rule
\citep{lin2006random, meinshausen2006quantile, athey2019}.
Intuitively, the weights at an observation $i$, that is
$W_{i, 1},\ldots,W_{i,n}$, should up weight a large set of
observations $\ell_1,\ldots,\ell_m$ for which
$(v^{\top}X_{\ell_1}, \widehat{\Phi}_n(X_{\ell_1})),\ldots,
(v^{\top}X_{\ell_m}, \widehat{\Phi}_n(X_{\ell_m}))$ and
$(v^{\top}X_{i}, \widehat{\Phi}_n(X_{i}))$ all can be (approximately)
described by the same polynomial of order $q+1$. Such weights can be
constructed in a greedy fashion by using a random forest with the
following modified splitting rule: For each proposed split, fit a
polynomial of order $q+1$ with $v^{\top}X$ as argument on each child
node and use the residual sum of squares across both child nodes as
impurity measure. We denote this type of random forest with polynomial
splitting in direction $v$ by \texttt{rfpoly}-$v$. The fitted random
forest regression function $\widehat{\mu}$ can be expressed as
\begin{equation*}
    \widehat{\mu}(x)=\sum_{i=1}^{n}\widehat{w}_i(x)\widehat{\Phi}_n(X_i),
\end{equation*}
where $\widehat{w}_i:\mathcal{X}\rightarrow [0,1]$ are weight
functions that are given by
\begin{equation*}
  \widehat{w}_i(x)=\frac{1}{M}\sum_{k=1}^{M}\frac{\mathds{1}(i\in \widehat{\mathcal{L}}_k(x))}{|\widehat{\mathcal{L}}_k(x)|},
\end{equation*}
with $M$ the number of trees in the random forest and
$\widehat{\mathcal{L}}_k(x)$ the sample indices specified by the
$k$-th tree's terminal node in which $x$ lies. Based on these weights,
we then define for all $i,\ell\in\{1,\ldots,n\}$ the weights
$W_{i,\ell}\coloneqq\widehat{w}_i(X_{\ell})$ and use them in the local
polynomial derivative estimation. The full procedure is detailed in
Algorithm~\ref{alg:forest_locpol} which includes an additional
regularization step discussed in the following section.

\begin{algorithm}
  \caption{\texttt{RFLocPol}}\label{alg:forest_locpol}
  \SetKwInOut{Input}{Input}
  \SetKwInOut{Parameters}{\normalfont \textit{Tuning}}
  \SetKwInOut{Output}{Output}

  \Input{Data
    $(X_1,\widehat{\Phi}_n(X_1)),\ldots,(X_n,\widehat{\Phi}_n(X_n))$,
    order $k$, direction $v$}
  \Parameters{Penalty $\lambda$, \texttt{rf} parameters $\Gamma$}
  \Output{Directional derivative estimates
    $\widehat{D_v^k\Phi}_n(X_1),\ldots,\widehat{D_v^k\Phi}_n(X_n)$}

  \BlankLine

  $\widehat{\mu}\gets$ \texttt{rfpoly}-$v$ on
  $(X_1,\widehat{\Phi}_n(X_1)),\ldots,(X_n,\widehat{\Phi}_n(X_n))$ with
  parameters $\Gamma$
  
  Extract weight matrix $W=(w_{i}(X_{\ell}))_{i, \ell}$ from $\widehat{\mu}$
  
  \For{$i\in\{1,\ldots,n\}$}{
    $\hat{\beta}\gets $ coefficients of
      order $q+1$ local polynomial fit in \eqref{eq:local_polynomial_penalized}
      with penalty $\lambda$
    
    $\widehat{D_v^k\Phi}_n(X_i)\gets k!\,\hat{\beta}_{i,k}$
  }
\end{algorithm}

\subsubsection{Additional regularization and tuning of
  hyperparameters}\label{sec:parameter_tuning}

Since the function $\Phi_0$ is assumed to be continuously
differentiable up to order $q$, it can be beneficial to regularize the
local polynomial estimate in \eqref{eq:local_polynomial_fit} to ensure
the derivatives become smoother. We propose to do this by estimating
all coefficients $\hat{\beta}=(\hat{\beta}(1),\ldots,\hat{\beta}(n))$
simultaneously by minimizing the penalized weighted mean squared loss
\begin{equation}
  \label{eq:local_polynomial_penalized}
  \hat{\beta}\coloneqq\argmin_{\beta\in\mathbb{R}^{n\times (q+1)}}\sum_{i,\ell=1}^{n}\left(\widehat{\Phi}_n(X_\ell)-\sum_{j=0}^{q+1}\beta_{i,j}((X_{\ell}-X_i)^{\top}
    v)^{j}\right)^2 W_{i, \ell} + \lambda P_W(\beta),
\end{equation}
where the penalty term $P_W$ is defined for all
$\beta\in\mathbb{R}^{n\times (q+1)}$ by
\begin{equation*}
    P_W(\beta)\coloneqq\sum_{i=1}^n\sum_{j=1}^{q+1}\left(\sum_{\ell=1}^n(j!\,\beta_{i,
      j}-j!\,\beta_{\ell, j})W_{i, \ell}\right)^2.
\end{equation*}
By \eqref{eq:locpol_deriv_est} the term $j!\,\beta_{i,j}$ parametrizes
the derivative $D_v^{j}\Phi_0(X_i)$, which implies that the penalty
term $P_W$ penalizes large differences between the derivatives at each
point $i$ and the locally averaged derivatives close to $i$. This
penalty therefore enforces smoothness of the derivatives.

Including this penalization the full \texttt{RFLocPol} procedure
depends on two types of hyperparameters; the penalty parameter
$\lambda$ from the penalized local polynomial and the random forest
parameters $\Gamma$ used to fit the random forests
\texttt{rfpoly}-$v$. Both parameters substantially affect the
performance of the overall procedure and need to be selected
carefully. We suggest a heuristic tuning procedure that selects
optimal parameters $(\Gamma^*, \lambda^*)$ from a $K$-tuple
$(\Gamma_1,\ldots,\Gamma_K)$ of random forest parameters and a
$L$-tuple $(\lambda_1,\ldots,\lambda_L)$ of penalty parameters. For
this we assume that the tuples are both ordered with decreasing
regularization stength, i.e., $\Gamma_i$ regularizes more than
$\Gamma_j$ and $\lambda_i>\lambda_j$ for all $i<j$. For the random
forest parameters, we could for example use an increasing sequence of
maximal depths or decreasing minimal node sizes. We then apply
\texttt{RFLocPol} for all different parameter settings and select the
parameters for which the local polynomial estimates
$\hat{\beta}_0(1),\ldots,\hat{\beta}_0(n)$ are not significantly worse
than $\widehat{\Phi}_n(X_1),\ldots,\widehat{\Phi}_n(X_n)$, measured by
a given loss function. Full details on this tuning procedure are
provided in Algorithm~\ref{alg:parameter_tuning} in Supplementary
material~\ref{sec:additional_algorithms}.

\subsection{Xtrapolation}\label{sec:xtrapolation_estimator}

We now adapt the plug-in estimates for the extrapolation bounds from
Section~\ref{sec:consistency} to use the forest-weighted local
polynomial derivative estimates from Section~\ref{sec:forest_locpol}
in a computationally efficient way. This leads to a procedure, which
we call \methodname/, that can estimate the extrapolation bounds from
arbitrary and potentially non-differentiable pilot estimates
$\widehat{\Phi}_n$.

Since the \texttt{RFLocPol} procedure estimates directional
derivatives it does not directly apply to the plug-in estimates in
\eqref{eq:vanilla_lower_bound_estimator} and
\eqref{eq:vanilla_upper_bound_estimator} which are expressed in terms
of partial derivatives. A workaround is to consider plug-in estimates
based on directional derivatives instead, however this relies on
computing directional derivatives in $n$ different directions which
would involve $n$ random forest fits. As this is computationally
infeasible in practice, we only focus on two special cases: The
order-one case (i.e., $q=1$ and arbitrary $d$) and the one-dimensional
case (i.e., $d=1$ and arbitrary $q$). In both cases the partial
derivatives correspond to directional derivatives and hence the
plug-in estimates can be combined with \texttt{RFLocPol}. More
specifically, for the order-one case the plug-in estimates only
involve first order partial derivatives
$\partial_1\widehat{\Phi}_n,\ldots,\partial_d\widehat{\Phi}_n$ which
are equal to the directional derivatives in the directions
$v\in\{e_1,\ldots,e_d\}$. Similarly, for the one-dimensional case all
involved partial derivatives correspond to the directional derivatives
in the direction $v=1$.

When using \texttt{RFLocPol} to estimate derivatives it can happen
that the derivative estimates are not equal to the derivatives of the
original estimate (assuming they even exist). As a consequence, it is
no longer guaranteed that the lower extrapolation bound estimate is
smaller than the upper bound estimate. To enforce this constraint, we
propose to check at a specific target point whether the lower estimate
is indeed smaller than the upper estimate and if not to set both
estimates to the average of the lower and upper extrapolation bound
estimate. The full \methodname/ procedure for the order-one case is
detailed in Algorithm~\ref{alg:xtrapolation_orderone}. The version for
the one-dimensional case is very similar and provided in
Algorithm~\ref{alg:xtrapolation_onedim} in Supplementary
material~\ref{sec:additional_algorithms}.

\begin{algorithm}
\caption{\methodname/ (order-one version)}\label{alg:xtrapolation_orderone}
\SetKwInOut{Input}{Input}
\SetKwInOut{Parameters}{\normalfont \textit{Tuning}}
\SetKwInOut{Output}{Output}

\Input{Estimates $\widehat{\Phi}_n(X_1),\ldots,\widehat{\Phi}_n(X_n)$,
  data $X_1,\ldots,X_n$, target points $\bar{x}_1,\ldots,\bar{x}_m$}

\Parameters{Penalty $\lambda$, \texttt{rf} parameters $\Gamma$}

\Output{Extrapolation bound estimates
  $\widehat{B}^{\operatorname{lo}}(\bar{x}_1),\ldots,
  \widehat{B}^{\operatorname{lo}}(\bar{x}_m)$,
  $\widehat{B}^{\operatorname{up}}(\bar{x}_1),\ldots,
  \widehat{B}^{\operatorname{up}}(\bar{x}_m)$}

\BlankLine

$\mathcal{D}\gets (X_1,\widehat{\Phi}_n(X_1)),\ldots,
(X_1,\widehat{\Phi}_n(X_n))$

\For{$j\in\{1,\ldots,d\}$}{
  $\widehat{\partial_j\Phi}_n(X_1),\ldots,\widehat{\partial_j\Phi}_n(X_n)\gets
  \texttt{RFLocPol}(\mathcal{D}, k=1, v=e_j, \lambda=\lambda,
  \Gamma=\Gamma)$
}

\For{$\ell\in\{1,\ldots,m\}$}{
  \For{$i\in\{1,\ldots,n\}$}{
    $S\gets \left\{\sum_{j=1}^d\widehat{\partial_j\Phi}_n(X_1)^{\top}(\bar{x}_{\ell}^j-X_i^j),\ldots,
      \sum_{j=1}^d\widehat{\partial_j\Phi}_n(X_n)^{\top}(\bar{x}_{\ell}^j-X_i^j)\right\}$
    
    $B_i^{\operatorname{lo}}\gets \widehat{\Phi}_n(X_i) + \min(S)
    \quad\text{and}\quad
    B_i^{\operatorname{up}}\gets \widehat{\Phi}_n(X_i) + \max(S)$
  }
  $\widehat{B}^{\operatorname{lo}}(\bar{x}_{\ell})\gets\max_{i}B_i^{\operatorname{lo}}
  \quad\text{and}\quad
  \widehat{B}^{\operatorname{up}}(\bar{x}_{\ell})\gets\min_{i}B_i^{\operatorname{up}}$
  
  \If{$\widehat{B}^{\operatorname{lo}}(\bar{x}_{\ell})> \widehat{B}^{\operatorname{up}}(\bar{x}_{\ell})$}{
    $\widehat{B}^{\operatorname{lo}}(\bar{x}_{\ell})\gets\big(\widehat{B}^{\operatorname{lo}}(\bar{x}_{\ell})
    + \widehat{B}^{\operatorname{up}}(\bar{x}_{\ell}))/2
    \quad\text{and}\quad
    \widehat{B}^{\operatorname{up}}(\bar{x}_{\ell})\gets\big(\widehat{B}^{\operatorname{lo}}(\bar{x}_{\ell})
    + \widehat{B}^{\operatorname{up}}(\bar{x}_{\ell}))/2$
  }

}
\end{algorithm}

\subsubsection{Computational speed up}\label{sec:computational_speedup}

We now consider a modification of the default \methodname/ procedure
that speeds up the computation in settings with large sample sizes $n$
and where one needs to estimate the extrapolation bounds at many
target points $m$. In those cases the default \methodname/ procedure
can be computationally expensive due to the $\mathcal{O}(nm)$
complexity. This can be reduce by only considering a subset of all
possible anchor points (i.e., the loop in line 6 of
Algorithm~\ref{alg:xtrapolation_orderone}). A naive approach could be
to simply subsample random anchor points and use those, however as is
clear in the one-dimensional case the bounds are generally tighter for
anchor points close the target point. Therefore, it can be beneficial
to subselect the anchor points by considering a notion of closeness to
the target points. While in one-dimensional settings using the
Euclidean distance is an obvious choice, it becomes more subtle in
multi-dimensional settings. This is because points that are far away
in Euclidean distance may have tight bounds if they are only far away
in directions in which the variance of the observed directional
derivatives is small. Therefore, to capture a more meaningful notion
of closeness (in the $q=1$ case), we propose to use a first-order
derivative scaled Euclidean distance. More specifically, denote by
$\widehat{\nabla\Phi}_n(X)\in\mathbb{R}^{n\times d}$ the matrix where
the $(i,j)$-th entry is $\widehat{\partial_j\Phi}_n(X_i)$ and let
$V\Sigma V^{\top}$ be the eigenvalue decomposition of the estimated
covariance
$$\widehat{\nabla\Phi}_n(X)^{\top}\widehat{\nabla\Phi}_n(X) -
\left(\frac{1}{n}\sum_{i=1}^n\widehat{\nabla\Phi}_n(X_i)\right)^{\top}\left(\frac{1}{n}\sum_{i=1}^n\widehat{\nabla\Phi}_n(X_i)\right).$$
To measure closeness of a sample point $X_i$ to a target point
$\bar{x}_k$, we propose to use the distance
$\|V\Sigma^{1/2}(X_i - \bar{x}_k)\|_2$. Intuitively, this distance is
larger in directions in which the derivatives change a lot and small
in directions in which the derivatives remain fixed.  As an
alternative, one can also use a distance measure induced by a random
forest, as proposed in the following section.

\subsubsection{Allowing for categorical covariates}\label{sec:categorical_covariates}

In some applications, not all covariates $X$ are continuous, which
means that the framework does not apply directly. However, it can be
adapted in settings where $X=(Z, W)$ with
$Z\in\mathcal{Z}\subseteq\mathbb{R}^{d_Z}$ continuous and
$W\in\mathcal{W}\subseteq\mathbb{R}^{d_W}$ categorical, as long as no
extrapolation occurs in the categorical predictors. The idea is to
apply the framework conditional on $W$. More specifically, we can
assume that for all $w\in\mathcal{W}$ the conditional function
$\Phi_0$ satisfies that $z\mapsto\Phi_0((z, w))$ is $q$-derivative
extrapolating and derive the same extrapolation bounds but conditional
on $W=w$. A straightforward modification of
Algorithm~\ref{alg:xtrapolation_orderone} is to select the anchor
points used in the loop of line 6 based on random forest weights. If
the random forest is grown sufficiently deep, one can expect that
samples with different $W$ values for which $\Phi_0$ is sufficiently
different will have small weights and hence not be used as anchor
points. This approach is used in the real data example in
Section~\ref{sec:exp_real} below.

\section{Numerical experiments}\label{sec:numerical_experiments}

We now present numerical experiments in which we investigate the
performance of the proposed estimation procedure
(Section~\ref{sec:exp_simulations}) and demonstrate possible
applications of the extrapolation bounds on real data
(Section~\ref{sec:exp_real}). All experiments can be reproduced using
the publicly available code
at
\url{https://github.com/NiklasPfister/ExtrapolationAware-Inference},
which includes an easy-to-use function to apply \texttt{Xtrapolation}
in other settings, too. For the regressions and cross-validation we
used the Python packages \texttt{scikit-learn} \citep{scikit-learn}
and \texttt{quantile-forest} \citep{Johnson2024}. 

\subsection{Simulation experiments}\label{sec:exp_simulations}

We begin by empirically analyzing the proposed \texttt{Xtrapolation}
procedure on simulated data. To this end, we consider random data
generating models for different sample sizes $n$ and dimensions
$d$. For each simulation, we randomly choose $\Din\subseteq[-2, 2]^d$
and $f:[-2, 2]^d\rightarrow\mathbb{R}$ and then generate $n$ i.i.d.\
copies $(X_1,Y_1),\ldots(X_{n}, Y_{n})$ of
$(X, Y)\in[-2, 2]^{d}\times\mathbb{R}$ defined via
\begin{equation*}
  X\sim \operatorname{Unif}(\Din)
  \quad\text{and}\quad
  Y = f(X) + \tfrac{1}{10}\varepsilon,
\end{equation*}
where $\varepsilon\sim\mathcal{N}(0, 1)$ independent of $X$. The way
we select the set $\Din$ and the function $f$ is such that the
extrapolation assumptions are satisfied and it is easy to interpret
the results. More specifically, we select $\Din$ and $f$ sequentially
as follows.
\begin{itemize}
\item[(1)] \emph{Selection of $\Din$:} Define the intervals
  $I_1\coloneqq[-2, -1)$, $I_2\coloneqq[-1, 0)$, $I_3\coloneqq[0, 1)$
  and $I_4\coloneqq[1, 2]$. Sample $d$ sets $C_1,\ldots,C_d$ uniformly
  from $\{[-2, 2]\setminus I_1,\ldots, [-2, 2]\setminus I_4\}$ (with
  replacement) and define $\Din\coloneqq \bigtimes_{j=1}^dC_j$ and
  $\Dout\coloneqq [-2,2]^d\setminus\Din$.
\item[(2)] \emph{Selection of $f$:} Let $f$ be a piecewise linear
  function in the first coordinate such that for all $x\in[-2,2]^d$ it
  holds
  \begin{equation*}
    f(x)=\sum_{j=1}^4(s_jx^1 + c_j)\mathds{1}_{I_j}(x^1),
  \end{equation*}
  where $s_1,\ldots,s_4$ are drawn randomly (see
  Supplementary material~\ref{sec:details_slopes}) and
  $c_1,\ldots,c_4$ are selected such that $f$ is continuous and
  $f(-2)=0$. Importantly, the slopes $s_1,\ldots,s_4$ are drawn in
  such a way that $f\triangleleft_{\Din}^1f$ (almost every) and
  $\operatorname{var}(f(U))=1$ for $U$ uniform on $[-2, 2]^d$.
\end{itemize}
This sampling procedure leads to models for which the conditional
expectation $\Psi_0$ of $Y$ give $X$ is first derivative extrapolating
almost everywhere. Furthermore, there are two types of extrapolation
scenarios that can happen depending on $\Din$ and $f$. Either the
lower and upper extrapolation bounds are equal on $\Dout$ implying
that the extrapolation assumption identifies $\Psi_0$ on $\Dout$ or
they do not coincide on $\Dout$ in which case $\Psi_0$ is not
identifiable on $\Dout$. The first case occurs if one of the inner
intervals (i.e., $I_2$ or $I_3$) is missing in the first coordinate
and $f$ has the smallest or largest slope on that interval. Examples
are shown in Figure~\ref{fig:details_slopes} in Supplementary
material~\ref{sec:details_slopes}.

For the following experiments, we generate $50$ datasets for all
combinations of
$n\in\{100,\allowbreak 200,\allowbreak 400,\allowbreak 600,\allowbreak
800,\allowbreak 1600\}$ and $d\in\{2, 8\}$, where each dataset is
sampled with randomly selected $\Din$ and $f$ as described above. We
then consider four regression procedures: Random forest regression
(\texttt{rf}), support vector regression (\texttt{svr}), neural
network regression (\texttt{mlp}) and ordinary least square regression
(\texttt{ols}). For all procedures -- except \texttt{ols} -- we tune
hyperparameters using a $5$-fold cross-validation and additionally
screen for variables using random forest based Gini impurity (see Supplementary material~\ref{sec:details_reg_tuning}). On top of each
regression fit, we then apply \texttt{Xtrapolation} with order $q=1$
and the parameter tuning discussed in
Section~\ref{sec:parameter_tuning} (see
Supplementary material~\ref{sec:details_xtra}) to
predict the lower and upper extrapolation bounds
$\widehat{B}^{\operatorname{lo}}_{n}(x)$ and
$\widehat{B}^{\operatorname{up}}_{n}(x)$ for all
$x\in\widehat{\mathcal{D}}_{\operatorname{in}}\cup\widehat{\mathcal{D}}_{\operatorname{out}}$,
where $\widehat{\mathcal{D}}_{\operatorname{in}}$ consists of $200$ uniformly
sampled points on $\Din$ and $\widehat{\mathcal{D}}_{\operatorname{out}}$ consists
of $200$ uniformly sampled points on $\Dout$.

\begin{figure}[t]
  \centering
  \includegraphics{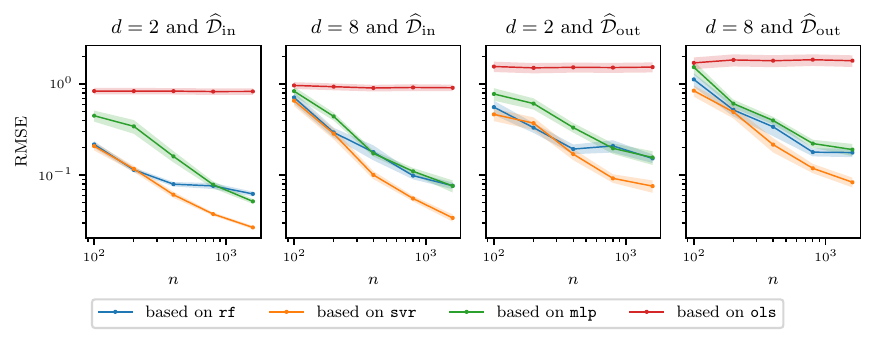}
  \caption{Accuracy of estimated extrapolation bounds measured using
    RMSE given in \eqref{eq:rmse_consistency}. For all three
    nonparametric regression procedures the RMSE decays with
    increasing $n$. As expected the extrapolation bounds based on
    \texttt{ols} do not decay as a linear regression cannot
    approximate the piecewise linear conditional expectations on
    $\Din$.}
  \label{fig:consistency_extrapolation_bound}
\end{figure}

\paragraph{Estimation accuracy of extrapolation bounds}

We first assess how accurately the extrapolation bounds for different
regression procedures are estimated and empirically validate
Theorem~\ref{thm:consistency}. We evaluate the accuracy of the
estimated bounds by comparing them with oracle extrapolation bounds
(i.e., using the true function $f$ but only optimizing over the anchor
points $X_1,\ldots,X_n$) evaluated at the same new
observations. Formally, we consider the root mean squared error (RMSE)
given by
\begin{equation}
  \label{eq:rmse_consistency}
  \sqrt{\frac{1}{200}\sum_{x\in\mathcal{D}}
    \left(\widehat{B}^{\operatorname{lo}}_{n}(x)
      -
      B^{\operatorname{lo}}_{f,\{X_1,\ldots,X_n\}}(x)\right)^2}
  + \sqrt{\frac{1}{200}\sum_{x\in\mathcal{D}}\left(\widehat{B}^{\operatorname{up}}_{n}(x)
      - B^{\operatorname{up}}_{f,\{X_1,\ldots,X_n\}}(x)\right)^2},
\end{equation}
where $\mathcal{D}=\widehat{\mathcal{D}}_{\operatorname{in}}$ or
$\mathcal{D}=\widehat{\mathcal{D}}_{\operatorname{out}}$. The results
are shown in Figure~\ref{fig:consistency_extrapolation_bound}. As
expected the extrapolation bounds based on \texttt{ols} do not
converge. In contrast, for all three nonparametric estimators
\texttt{rf}, \texttt{svr} and \texttt{mlp} the extrapolation bounds
converge both on $\Din$ and $\Dout$. Moreover, while the increased
dimension leads to slightly worse accuracy the RMSE decays at a
similar rate. This is particularly interesting as the
\texttt{Xtrapolation} procedure does not explicitly take sparsity into
account, but appears to automatically adapt to the sparsity in the
regression estimates (due to the variable screening). We see this as
promising empirical evidence that the random forest weights ensure
that the derivatives are estimated well even in multiple dimensions.

\paragraph{Out-of-support prediction}
Next we show how the extrapolation bounds can be used to construct
regression-agnostic predictions on $\Dout$ that are worst-case optimal
as discussed in Section~\ref{sec:out-of-support-prediction}. We use
the same simulations with $n=1600$ and $d=2$ but now additionally
estimate \eqref{eq:point_estimate_oos} in
Section~\ref{sec:out-of-support-prediction} by
\begin{equation}
  \label{eq:predicted_values}
  \widehat{f}_{\operatorname{xtra}}(x)=\frac{\widehat{B}_n^{\operatorname{lo}}(x)+\widehat{B}_n^{\operatorname{up}}(x)}{2}.
\end{equation}
We then compare this estimate with regression estimates
$\widehat{f}_{\operatorname{reg}}$ resulting from the plain
regressions. We evaluate the performance using the worst-case RMSE
given by
\begin{equation}
  \label{eq:worst_case_rmse}
  \frac{1}{|\mathcal{D}|}\sum_{x\in\mathcal{D}}\sup_{Q\in\mathcal{Q}_0}\sqrt{\mathbb{E}_{Q(x,\cdot)}[(Y_x-\widehat{f}(x))^2]},
\end{equation}
where
$\widehat{f}\in\{\widehat{f}_{\operatorname{xtra}},
\widehat{f}_{\operatorname{reg}}\}$ and
$\mathcal{D}\in\{\widehat{D}_{\operatorname{in}},
\widehat{D}_{\operatorname{out}}\}$. As its not directly possible to
evaluate this loss, we use that the worst-case $Q\in\mathcal{Q}_0$ is
attained at the true extrapolation bounds (see proof of
Proposition~\ref{thm:worst_case_prediction}) and then approximate the
loss using the oracle extrapolation bounds. The results are shown in
Figure~\ref{fig:out_of_support_prediction}, We additionally
distinguish between identifiable and non-identifiable extrapolation
settings (i.e., where $\Psi_0$ is identified on $\Dout$ and where not)
by splitting the $50$ simulations into either identifiable ($22$
simulations) and unidentifiable settings ($28$ simulations) depending
on whether the oracle bounds are approximately equal on the evaluated
points or not. We observe that on $\Din$ both plain regression and
\texttt{Xtrapolation} perform similarly. The largest difference occurs
for \texttt{rf}, which makes sense as \texttt{Xtrapolation} smooths
the estimates which has almost no effect on the already smooth
\texttt{svr} and \texttt{mlp} estimates but slightly improves the
non-smooth \texttt{rf} estimates. Furthermore, while the regression
estimates extrapolate differently on $\Dout$, the differences
disappear after applying \texttt{Xtrapolation} which is expected as
the extrapolation estimates the same target quantities in all cases.

\begin{figure}[t]
  \centering
  \includegraphics[width=\textwidth]{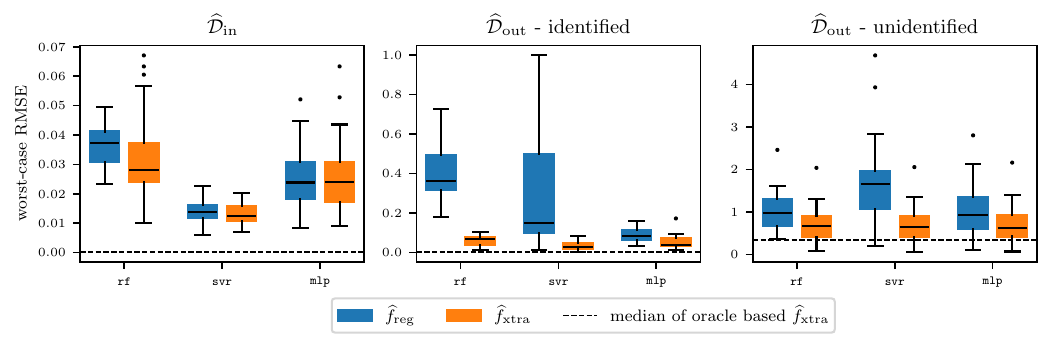}
  \caption{Comparison of worst-case RMSE (see \eqref{eq:worst_case_rmse}) for
    $\widehat{f}_{\operatorname{xtra}}$ and $\widehat{f}_{\operatorname{reg}}$ on
    $\Din$ (left) and $\Dout$ (middle and right). The out-of-support
    predictions are separated into simulations for which the oracle
    lower and upper bound agree, i.e., the conditional expectation is
    identified (middle) and those that are not (right). While the
    out-of-support predictions of plain regression
    $\widehat{f}_{\operatorname{reg}}$ depend on the used regression
    procedure this is not the case for the estimates
    $\widehat{f}_{\operatorname{xtra}}$ based on
    $\texttt{Xtrapolation}$.}
  \label{fig:out_of_support_prediction}
\end{figure}

\paragraph{Quantifying extrapolation}
Finally, we consider the proposed extrapolation score from
Section~\ref{sec:extrapolation_score}. To this end, we consider the
simulations with $n=1600$ and $d\in\{2, 8\}$. We estimate the
extrapolation scores for all
$x\in\widehat{\mathcal{D}}_{\operatorname{in}}\cup\widehat{\mathcal{D}}_{\operatorname{out}}$
by
\begin{equation*}
  \widehat{S}(x)\coloneqq \frac{\widehat{B}^{\operatorname{up}}_{n}(x)-\widehat{B}^{\operatorname{lo}}_{n}(x)}{\widehat{\sigma}_{\operatorname{CV}}},
\end{equation*}
where $\widehat{\sigma}_{\operatorname{CV}}$ is the square root of the
cross-validation generalization error computed for the regression
method used to estimate the extrapolation bounds. As a benchmark, we
additionally compute a minimal Euclidean distance defined for all
$x\in\widehat{\mathcal{D}}_{\operatorname{in}}\cup\widehat{\mathcal{D}}_{\operatorname{out}}$
by
\begin{equation*}
  \widehat{E}(x)\coloneqq \min_{i\in\{1,\ldots,n\}}\|x-X_i\|_2.
\end{equation*}
To compare how well $\widehat{S}$ and $\widehat{E}$ capture
extrapolation, we compute for all thresholds $\lambda\in[0,\infty)$
(i) the fraction of observations in
$\widehat{\mathcal{D}}_{\operatorname{in}}\cup\widehat{\mathcal{D}}_{\operatorname{out}}$
which have an extrapolation score below $\lambda$ and (ii) the
cumulative RMSE of the predictions $\widehat{f}_{\operatorname{xtra}}$
defined in \eqref{eq:predicted_values} at all points with a score
below $\lambda$, i.e.,
\begin{equation*}
  \sqrt{|\{x\in\widehat{\mathcal{D}}_{\operatorname{in}}\cup\widehat{\mathcal{D}}_{\operatorname{out}}\mid
    \widehat{\operatorname{score}}(x)\leq\lambda\}|^{-1}\sum_{x\in\widehat{\mathcal{D}}_{\operatorname{in}}\cup\widehat{\mathcal{D}}_{\operatorname{out}}}\left(\widehat{f}_{\operatorname{xtra}}(x)-f(x)\right)^2\mathds{1}(\widehat{\operatorname{score}}(x)\leq\lambda)},
\end{equation*}
where $\widehat{\operatorname{score}}=\widehat{S}$ or
$\widehat{\operatorname{score}}=\widehat{E}$. We use
$\hat{f}_{\operatorname{xtra}}$, here as the plain regression
estimates may behave arbitrary outside of $\Dout$. The results are
shown in Figure~\ref{fig:quantify_extrapolation} (left and
middle). For all regression methods the cumulative RMSE increases
sharply after $0.5$ when sorted according to $\widehat{E}$. This makes
sense as $\widehat{E}$ only separates $\Din$ from $\Dout$ but does not
take into account whether the function might also be accurate on
$\Dout$. In contrast, the extrapolation score $\widehat{S}$ also
separates points on $\Dout$ for which the predictions are expected to
be good.  We further separate points with $\widehat{S}\leq 1$ for
which the extrapolation error is of smaller order than residual noise
level and points with $\widehat{S}>1$ for which extrapolation error is
of larger order. The aggregated RMSEs for these splits are shown in
Figure~\ref{fig:quantify_extrapolation} (right). As expected the RMSE
is small (and on the order of the residual noise level) when
$\widehat{S}\leq 1$ and becomes large if $\widehat{S}> 1$.

\begin{figure}[t]
  \centering
  \includegraphics[width=\textwidth]{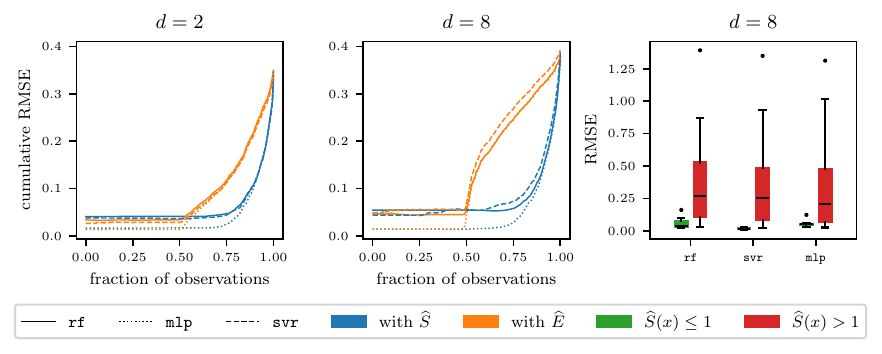}
  \caption{(left and middle) Comparison of extrapolation score
    $\widehat{S}$ and a Euclidean based benchmark score
    $\widehat{E}$. While $\widehat{E}$ is only able to separate $50\%$
    of the samples (which corresponds to the points in $\Din$) before
    the cumulative RMSE increases, the extrapolation score
    $\widehat{S}$ is able to separate approximately $75\%$ (hence also
    points in $\Dout$) before the cumulative RMSE increases. (right)
    Comparison of the RMSE for points with $\widehat{S}\leq 1$ versus
    points with $\widehat{S}>1$ (the value 1 corresponds to points
    where the extrapolation error is equal to the residual noise
    level).}
  \label{fig:quantify_extrapolation}
\end{figure}

\subsection{Extrapolation-aware prediction intervals on real
  data}\label{sec:exp_real}

In this section we illustrate how explicitly taking into account
extrapolation can improve uncertainty quantification and help detect
when nonparametric prediction procedures are extrapolating. We
consider two datasets: (i) The \texttt{biomass} dataset due to
\citet{hiernaux2023}, where the task is predicting the foilage dry
mass of a tree from its crown area and (ii) the well-known
\texttt{abalone} dataset from the UCI ML repository
\citep{misc_abalone_1}, where the task is to predict the age of
abalone shells from several phenotype measurements (sex, length,
diameter, height, whole weight, shucked weight, viscera weight and
shell weight).

Throughout this section we fix $\alpha=0.2$. We consider four
different standard nonparametric methods to construct predictions
intervals: (i) Quantile regression forests
\citep{meinshausen2006quantile}, denoted by \texttt{qrf}, (ii)
quantile neural networks \citep{taylor2000quantile}, denoted by
\texttt{qnn}, (iii) conformalized quantile regression forests, denoted
by \texttt{cpqrf}, and (iv) conformalized quantile neural networks,
denoted by \texttt{cpqnn}.  The two conformalized methods are based on
\citet{romano2019conformalized} and calibrate the prediction intervals
from the corresponding quantile regression to have a finite sample
exact unconditional coverage guarantee similar to conventional
conformal prediction \citep{balasubramanian2014conformal}. We then
compare each of these methods with its extrapolation-aware
counterpart, denoted by \texttt{xtra-qrf}, \texttt{xtra-qnn},
\texttt{xtra-cpqrf} and \texttt{xtra-cpqnn} respectively, which is
constructed by applying \texttt{Xtrapolation} to the conditional
quantiles. More specifically, for all
$\star\in\{\texttt{qrf},\texttt{qnn},\texttt{cpqrf},\texttt{cpqnn}\}$
we compare
\begin{equation*}
  \widehat{C}_{\star}^{\mathrm{pred}}(x) \coloneqq
  \left[\widehat{\mathcal{T}}_{\star}^{\alpha/2}(x),
    \widehat{\mathcal{T}}_{\star}^{1-\alpha/2}(x)\right]
  \quad\text{with}\quad
  \widehat{C}_{\texttt{xtra-}\star}^{\mathrm{pred}}(x) \coloneqq
  \left[\widehat{B}^{\operatorname{lo}}_{\widehat{\mathcal{T}}^{\alpha/2}_{\star}}(x), \widehat{B}^{\operatorname{up}}_{\widehat{\mathcal{T}}^{1-\alpha/2}_{\star}}(x)\right],
\end{equation*}
where $\widehat{\mathcal{T}}_{\star}^{\alpha/2}$ and
$\widehat{\mathcal{T}}_{\star}^{1-\alpha/2}$ are estimates based on
quantile regression forest,
$\widehat{B}^{\operatorname{lo}}_{\widehat{\mathcal{T}}^{\alpha/2}_{\star}}(x)$
is the estimate of the lower extrapolation bound for
$\widehat{\mathcal{T}}_{\star}^{\alpha/2}$ and
$\widehat{B}^{\operatorname{up}}_{\widehat{\mathcal{T}}^{1-\alpha/2}_{\star}}(x)$
is the estimate of the upper extrapolation bound for
$\widehat{\mathcal{T}}_{\star}^{1-\alpha/2}$. As the data contains
point masses (the age variable in the \texttt{abalone} dataset is
discrete), we use averaged randomized prediction intervals to
calibrate the coverage to the precise level $\alpha$ (see
Supplementary material~\ref{sec:details_randomized_pi} for details).

For the comparison we generate two types of train and test splits for
both datasets: (i) \textit{Random splits} that randomly split the data
into $8$ approximately equally sized sets and (ii)
\textit{extrapolation splits} that split the data into $8$
approximately equally sized sets according to a predictor variable
(crown area for \texttt{biomass} and length for \texttt{abalone}). The
resulting coverage on each split is given in
Figure~\ref{fig:biomass_inter_vs_extra} (top for \texttt{biomass}
bottom for \texttt{abalone}).
\begin{figure}[t]
  \centering
  \includegraphics[width=\textwidth]{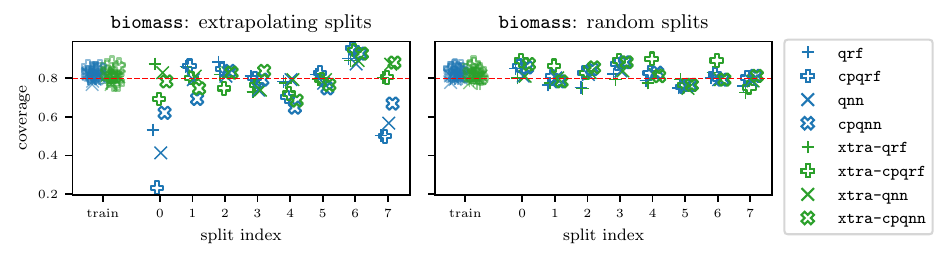}
  \includegraphics[width=\textwidth]{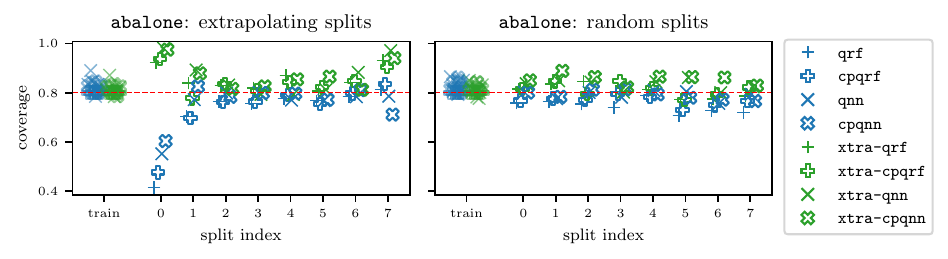}
  \caption{Coverage of prediction intervals on \texttt{biomass} and
    \texttt{abalone} datasets. Extrapolating splits are train test
    splits that leave out certain ranges of tree crown areas for
    \texttt{biomass} and length for \texttt{abalone}, while random
    splits are randomly drawn splits. While the standard prediction
    intervals (blue markers) under cover for some extrapolating
    splits, the extrapolation-aware counterparts (blue markers) guard
    against such under coverage.}
  \label{fig:biomass_inter_vs_extra}
\end{figure}
While the standard prediction interval estimates (blue markers) have
good coverage for random splits, they under cover for some of the
extrapolation splits. The reason is that they are not intended to work
outside of the support and will behave differently depending on the
underlying regression procedure (e.g., tree-based models extrapolate
constant and neural networks linearly). For \texttt{qrf} applied to
\texttt{biomass} this can be seen in
Figure~\ref{fig:quantile_scatterplot}, where we plot the estimated
quantiles for both \texttt{qrf} and \texttt{xtra-qrf} (similar plots
for the other methods are provided in Supplementary
material~\ref{sec:additional_results}). The extrapolation-aware
prediction intervals (green markers) perform similar to standard
prediction intervals on random splits but avoid under coverage on the
extrapolation splits. The slightly conservative behavior of the
extrapolation-aware prediction intervals is expected since the
extrapolation bounds account for the uncertainty due to the
extrapolation. The fact that the coverage is preserved provides
empirical evidence that the proposed extrapolation assumption is
indeed satisfied on both datasets.

\begin{figure}[t]
  \centering
  \includegraphics{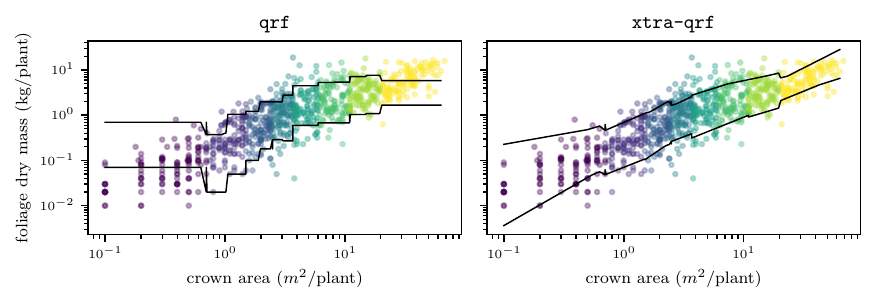}
  \caption{Estimated prediction intervals for \texttt{qrf} (left) and
    \texttt{xtra-qrf} (right) for the extrapolation split. Colors
    correspond to different extrapolation splits. While \texttt{qrf}
    always extrapolates constantly, the extrapolation-aware versions
    adapt to the changes observed on the training data.}
  \label{fig:quantile_scatterplot}
\end{figure}

We now show how the difference between lower and upper extrapolation
bounds can be used as an extrapolation score. To this end, we further
compute, in an $8$-fold cross-validation style (using the
extrapolation splits), for all
$\star\in\{\texttt{qrf},\texttt{qnn},\texttt{cpqrf},\texttt{cpqnn}\}$
and for each sample point $X_i$ the extrapolation scores
\begin{equation}
  \label{eq:extra-score-predint}
  \left(\widehat{B}_{
    \widehat{\mathcal{T}}^{\alpha/2}_{\star}}^{\operatorname{up}}(X_i)-\widehat{B}_{
    \widehat{\mathcal{T}}^{\alpha/2}_{\star}}^{\operatorname{lo}}(X_i)\right) + \left(\widehat{B}_{\widehat{\mathcal{T}}^{1-\alpha/2}_{\star}}^{\operatorname{up}}(X_i)-\widehat{B}_{\widehat{\mathcal{T}}^{1-\alpha/2}_{\star}}^{\operatorname{lo}}(X_i)\right),
\end{equation}
where the estimates in this expression are computed on all splits not
containing $X_i$. We then sort the samples according to this score
from small to large and estimate the coverage using a rolling window
(size $100$ for \texttt{biomass} and size $400$ for
\texttt{abalone}). The result is shown in
Figure~\ref{fig:coverage_extrapolation_score}.
\begin{figure}[h!]
  \centering
  \includegraphics{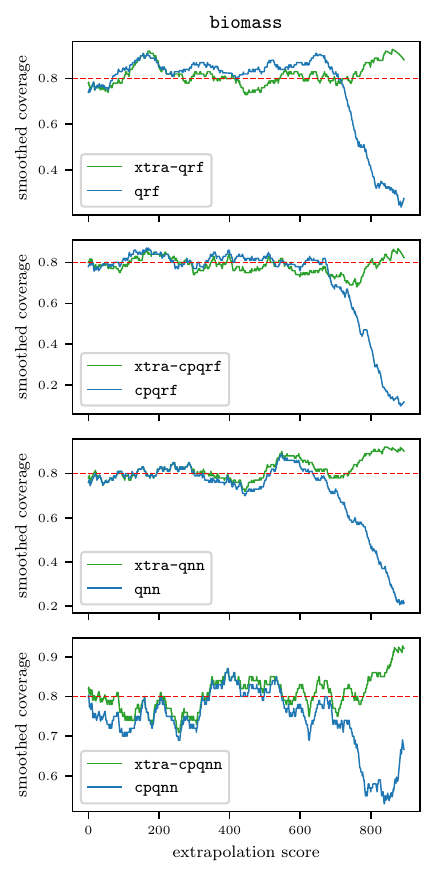}%
  \includegraphics{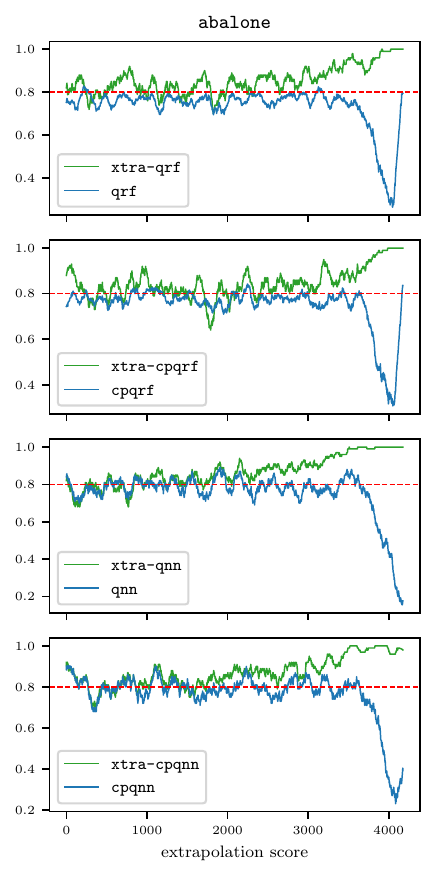}%
  \caption{Coverage computed based on a rolling window across
    observations sorted by the extrapolation score in
    \eqref{eq:extra-score-predint}. For observations with a small
    extrapolation score the coverage of standard prediction intervals
    is close to their extrapolation-aware counterparts. However, for
    large extrapolation scores, the coverages diverge. The
    extrapolation-aware prediction intervals are expected to remain
    valid or become conservative (i.e., above $0.8$) when the
    extrapolation score becomes large, while we have no guarantees on
    the extrapolation behavior of standard prediction intervals. }
  \label{fig:coverage_extrapolation_score}
\end{figure}
For both datasets and all methods the coverage remains similar for
small extrapolation scores but starts diverging for larger ones.
Importantly, while the standard prediction intervals start to under
cover, their extrapolation-aware counterparts only become
conservative. Figure~\ref{fig:biomass_inter_vs_extra} (extrapolating
split 7) and Figure~\ref{fig:coverage_extrapolation_score} (for high
extrapolation scores) further show that the extrapolation behavior on
\texttt{abalone} is substantially different between the \texttt{qrf}
and \texttt{qnn} based intervals (likely due to the difference in the
underlying function classes). In contrast and as expected from the
underlying theory, the extrapolation-aware counterparts appear to
extrapolate similarly regardless of the underlying method.

As an overall summary, we find that none of the standard methods or
methodologies for constructing prediction intervals (e.g., quantile
regression or conformalized versions of them) provides guarantees for
extrapolation and may fail severely when extrapolating. Our
extrapolation-aware framework, however, works well but is somewhat
conservative. While mathematical guarantees for extrapolation must
rely on uncheckable conditions, namely in our context that derivatives
are extrapolating, it is encouraging to see that the validation of
conditional-tailored coverage on real datasets supports our
theoretical results assuming an uncheckable extrapolation condition.

\section{Discussion}

We defined extrapolation as the process of performing inference on a
conditional function outside of the support of the conditioning
variable $X$. This type of extrapolation is, however, not directly
feasible in a conventional nonparametric sense as the data generating
distribution $P_0$ does not specify the conditional distribution
outside of the support of $X$. We therefore assumed the existence of a
Markov kernel that fully specifies the conditional -- also outside of
the support of $X$ -- but which might not be identified by $P_0$
alone. Then, by assuming that the conditional function behaves at most
as extreme on the entire domain as it does on the support of $X$, we
were able to construct extrapolation bounds on the conditional
function that are identified by $P_0$. We proposed to perform
inference on these extrapolation bounds instead of on the conditional
function directly. We emphasize that an extrapolation assumption is
needed: ours, assuming "at most as extreme derivatives" as in the
observed domain, seems natural and we gave some additional
interpretation after its Definition \ref{def:extrapolating_model}.

A key feature of our framework is that performing inference on the
extrapolation bounds instead of the conditional function (at least on
a population level) only affects the nonparametric analysis if
extrapolation occurs, since the lower and upper bound are both equal
to the conditional function on the support of $X$. This ensures
that the analysis is extrapolation-aware, which guards against
potential errors and can provide additional insights when one is
indeed extrapolating. Here, we considered three specific applications,
out-of-support prediction, uncertainty quantification and quantifying
extrapolation, but there are likely many more that could benefit
from this type of extrapolation-aware analysis. Even though the
extrapolation bounds are identified, they may be difficult to estimate
in practice. We propose a method that is able to estimate the bounds
in an estimator-agnostic way, which performs well in empirical
experiments. Importantly, this estimation procedure only takes the
estimated conditional function at the sample points as input (i.e.,
$(X_1,\widehat{\Phi}(X_1)),\ldots,(X_n,\widehat{\Phi}(X_n))$),
allowing practitioners to use their preferred nonparametric estimates
$\widehat{\Phi}$ to estimate the extrapolation bounds.

We hope the proposed framework will inspire new developments related
to extrapolation, which -- given its importance -- has received too
little attention in the wider statistics community so far. An
important direction of future work is to consider other types of
extrapolation assumptions. For example, one could consider shape
constraints on the conditional function, which can likely be
incorporated similarly to the derivative assumptions considered here.
Additionally, it would be interesting to see the proposed framework
applied across several applied domains in the style of the real data
analysis reported in Figure
\ref{fig:coverage_extrapolation_score}. This would provide valuable
insights into how realistic the extrapolation assumptions are and
whether there are variations that are more amenable in certain
applications.

\section*{Acknowledgements}

We thank Anton Rask Lundborg for helpful discussions as well as
Christian Igel and Pierre Hiernaux for providing the biomass data.  NP
was supported by a research grant (0069071) from Novo Nordisk
Fonden. The research was partially conducted during NP's research stay
at the Institute for Mathematical Research at ETH Z\"urich (FIM). PB
has received funding from the European Research Council (ERC) under
the European Union’s Horizon 2020 research and innovation programme
(grant agreement No 786461).

\bibliographystyle{abbrvnat}
\bibliography{refs}

\begin{thebibliography}{38}
\providecommand{\natexlab}[1]{#1}
\providecommand{\url}[1]{\texttt{#1}}
\expandafter\ifx\csname urlstyle\endcsname\relax
  \providecommand{\doi}[1]{doi: #1}\else
  \providecommand{\doi}{doi: \begingroup \urlstyle{rm}\Url}\fi

\bibitem[Athey et~al.(2019)Athey, Tibshirani, and Wager]{athey2019}
S.~Athey, J.~Tibshirani, and S.~Wager.
\newblock {Generalized random forests}.
\newblock \emph{The Annals of Statistics}, 47\penalty0 (2):\penalty0 1148 --
  1178, 2019.

\bibitem[Balasubramanian et~al.(2014)Balasubramanian, Ho, and
  Vovk]{balasubramanian2014conformal}
V.~Balasubramanian, S.-S. Ho, and V.~Vovk.
\newblock \emph{Conformal prediction for reliable machine learning: theory,
  adaptations and applications}.
\newblock Newnes, 2014.

\bibitem[Christiansen et~al.(2021)Christiansen, Pfister, Jakobsen, Gnecco, and
  Peters]{christiansen2021causal}
R.~Christiansen, N.~Pfister, M.~E. Jakobsen, N.~Gnecco, and J.~Peters.
\newblock A causal framework for distribution generalization.
\newblock \emph{IEEE Transactions on Pattern Analysis and Machine
  Intelligence}, 44\penalty0 (10):\penalty0 6614--6630, 2021.

\bibitem[Dai et~al.(2016)Dai, Tong, and Genton]{dai2016optimal}
W.~Dai, T.~Tong, and M.~G. Genton.
\newblock Optimal estimation of derivatives in nonparametric regression.
\newblock \emph{Journal of Machine Learning Research}, 17\penalty0
  (1):\penalty0 5700--5724, 2016.

\bibitem[De~Brabanter et~al.(2013)De~Brabanter, De~Brabanter, Gijbels, and
  De~Moor]{de2013derivative}
K.~De~Brabanter, J.~De~Brabanter, I.~Gijbels, and B.~De~Moor.
\newblock Derivative estimation with local polynomial fitting.
\newblock \emph{Journal of Machine Learning Research}, 14\penalty0
  (1):\penalty0 281--301, 2013.

\bibitem[Dong and Ma(2022)]{dong2022first}
K.~Dong and T.~Ma.
\newblock First steps toward understanding the extrapolation of nonlinear
  models to unseen domains.
\newblock In \emph{The Eleventh International Conference on Learning
  Representations}, 2022.

\bibitem[Efron(1981)]{efron1981nonparametric}
B.~Efron.
\newblock Nonparametric standard errors and confidence intervals.
\newblock \emph{Canadian Journal of Statistics}, 9\penalty0 (2):\penalty0
  139--158, 1981.

\bibitem[Folland(1999)]{folland1999real}
G.~B. Folland.
\newblock \emph{Real analysis: modern techniques and their applications},
  volume~40.
\newblock John Wiley \& Sons, 1999.

\bibitem[H\"ardle and Stoker(1989)]{hardle1989}
W.~H\"ardle and T.~M. Stoker.
\newblock Investigating smooth multiple regression by the method of average
  derivatives.
\newblock \emph{Journal of the American Statistical Association}, 84\penalty0
  (408):\penalty0 986--995, 1989.

\bibitem[Hiernaux et~al.(2023)Hiernaux, Issoufou, Igel, Kariryaa, Kourouma,
  Chave, Mougin, and Savadogo]{hiernaux2023}
P.~Hiernaux, H.~B.-A. Issoufou, C.~Igel, A.~Kariryaa, M.~Kourouma, J.~Chave,
  E.~Mougin, and P.~Savadogo.
\newblock Allometric equations to estimate the dry mass of sahel woody plants
  mapped with very-high resolution satellite imagery.
\newblock \emph{Forest Ecology and Management}, 529:\penalty0 120653, 2023.

\bibitem[H{\"o}rmander(2015)]{hormander2015analysis}
L.~H{\"o}rmander.
\newblock \emph{The analysis of linear partial differential operators {I}:
  {D}istribution theory and {F}ourier analysis}.
\newblock Springer, 2015.

\bibitem[Johnson(2024)]{Johnson2024}
R.~A. Johnson.
\newblock quantile-forest: A python package for quantile regression forests.
\newblock \emph{Journal of Open Source Software}, 9\penalty0 (93):\penalty0
  5976, 2024.
\newblock \doi{10.21105/joss.05976}.
\newblock URL \url{https://doi.org/10.21105/joss.05976}.

\bibitem[Klyne and Shah(2023)]{klyne2023average}
H.~Klyne and R.~D. Shah.
\newblock Average partial effect estimation using double machine learning.
\newblock \emph{arXiv preprint arXiv:2308.09207}, 2023.

\bibitem[Li and Heckman(2003)]{li2003local}
X.~Li and N.~E. Heckman.
\newblock Local linear extrapolation.
\newblock \emph{Journal of Nonparametric Statistics}, 15\penalty0
  (4-5):\penalty0 565--578, 2003.

\bibitem[Lin and Jeon(2006)]{lin2006random}
Y.~Lin and Y.~Jeon.
\newblock Random forests and adaptive nearest neighbors.
\newblock \emph{Journal of the American Statistical Association}, 101\penalty0
  (474):\penalty0 578--590, 2006.

\bibitem[Lundborg and Pfister(2023)]{lundborg2023perturbation}
A.~R. Lundborg and N.~Pfister.
\newblock Perturbation-based analysis of compositional data.
\newblock \emph{arXiv preprint arXiv:2311.18501}, 2023.

\bibitem[Mack and M{\"u}ller(1989)]{mack1989}
Y.~Mack and H.-G. M{\"u}ller.
\newblock Derivative estimation in nonparametric regression with random
  predictor variable.
\newblock \emph{Sankhy{\=a}: The Indian Journal of Statistics, Series A},
  51\penalty0 (1):\penalty0 59--72, 1989.

\bibitem[Masry and Fan(1997)]{masry1997}
E.~Masry and J.~Fan.
\newblock Local polynomial estimation of regression functions for mixing
  processes.
\newblock \emph{Scandinavian Journal of Statistics}, 24\penalty0 (2):\penalty0
  165--179, 1997.

\bibitem[Meinshausen(2006)]{meinshausen2006quantile}
N.~Meinshausen.
\newblock Quantile regression forests.
\newblock \emph{Journal of Machine Learning Research}, 7\penalty0
  (35):\penalty0 983--999, 2006.

\bibitem[Nash et~al.(1995)Nash, Sellers, Talbot, Cawthorn, and
  Ford]{misc_abalone_1}
W.~Nash, T.~Sellers, S.~Talbot, A.~Cawthorn, and W.~Ford.
\newblock {Abalone}.
\newblock UCI Machine Learning Repository, 1995.
\newblock {DOI}: https://doi.org/10.24432/C55C7W.

\bibitem[Pan and Yang(2009)]{pan2009survey}
S.~J. Pan and Q.~Yang.
\newblock A survey on transfer learning.
\newblock \emph{IEEE Transactions on knowledge and data engineering},
  22\penalty0 (10):\penalty0 1345--1359, 2009.

\bibitem[Pearl(2009)]{pearl2009causality}
J.~Pearl.
\newblock \emph{Causality}.
\newblock Cambridge University Press, 2009.

\bibitem[Pedregosa et~al.(2011)Pedregosa, Varoquaux, Gramfort, Michel, Thirion,
  Grisel, Blondel, Prettenhofer, Weiss, Dubourg, Vanderplas, Passos,
  Cournapeau, Brucher, Perrot, and Duchesnay]{scikit-learn}
F.~Pedregosa, G.~Varoquaux, A.~Gramfort, V.~Michel, B.~Thirion, O.~Grisel,
  M.~Blondel, P.~Prettenhofer, R.~Weiss, V.~Dubourg, J.~Vanderplas, A.~Passos,
  D.~Cournapeau, M.~Brucher, M.~Perrot, and E.~Duchesnay.
\newblock Scikit-learn: Machine learning in {P}ython.
\newblock \emph{Journal of Machine Learning Research}, 12\penalty0
  (85):\penalty0 2825--2830, 2011.

\bibitem[Pfister et~al.(2019)Pfister, Bauer, and
  Peters]{pfister2019causalkinetix}
N.~Pfister, S.~Bauer, and J.~Peters.
\newblock Learning stable and predictive structures in kinetic systems.
\newblock \emph{Proceedings of the National Academy of Sciences}, 116\penalty0
  (51):\penalty0 25405--25411, 2019.

\bibitem[Robins(1986)]{robins1986new}
J.~Robins.
\newblock A new approach to causal inference in mortality studies with a
  sustained exposure period—application to control of the healthy worker
  survivor effect.
\newblock \emph{Mathematical modelling}, 7\penalty0 (9-12):\penalty0
  1393--1512, 1986.

\bibitem[Romano et~al.(2019)Romano, Patterson, and
  Candes]{romano2019conformalized}
Y.~Romano, E.~Patterson, and E.~Candes.
\newblock Conformalized quantile regression.
\newblock \emph{Advances in neural information processing systems}, 32, 2019.

\bibitem[Rubin(2005)]{rubin2005causal}
D.~B. Rubin.
\newblock Causal inference using potential outcomes: Design, modeling,
  decisions.
\newblock \emph{Journal of the American Statistical Association}, 100\penalty0
  (469):\penalty0 322--331, 2005.

\bibitem[Saengkyongam et~al.(2023)Saengkyongam, Rosenfeld, Ravikumar, Pfister,
  and Peters]{saengkyongam2023identifying}
S.~Saengkyongam, E.~Rosenfeld, P.~Ravikumar, N.~Pfister, and J.~Peters.
\newblock Identifying representations for intervention extrapolation.
\newblock \emph{arXiv preprint arXiv:2310.04295}, 2023.

\bibitem[Schreiber et~al.(2023)Schreiber, Boix, wook Lee, Li, Guan, Chang,
  Chang, Hawkins-Hooker, Sch{\"o}lkopf, Schweikert, Carulla, Canakoglu, Guzzo,
  Nanni, Masseroli, Carman, Pinoli, Hong, Yip, Spence, Batra, Song, Mahony,
  Zhang, Tan, Shen, Sun, Shi, Adrian, Sandstrom, Farrell, Halow, Lee, Jiang,
  Yang, Epstein, Strattan, Bernstein, Snyder, Kellis, Stafford, Kundaje, and
  Participants]{schreiber2023encode}
J.~Schreiber, C.~Boix, J.~wook Lee, H.~Li, Y.~Guan, C.-C. Chang, J.-C. Chang,
  A.~Hawkins-Hooker, B.~Sch{\"o}lkopf, G.~Schweikert, M.~R. Carulla,
  A.~Canakoglu, F.~Guzzo, L.~Nanni, M.~Masseroli, M.~J. Carman, P.~Pinoli,
  C.~Hong, K.~Y. Yip, J.~P. Spence, S.~S. Batra, Y.~S. Song, S.~Mahony,
  Z.~Zhang, W.~Tan, Y.~Shen, Y.~Sun, M.~Shi, J.~Adrian, R.~Sandstrom,
  N.~Farrell, J.~Halow, K.~Lee, L.~Jiang, X.~Yang, C.~Epstein, J.~S. Strattan,
  B.~Bernstein, M.~Snyder, M.~Kellis, W.~Stafford, A.~Kundaje, and E.~I.~C.
  Participants.
\newblock The encode imputation challenge: a critical assessment of methods for
  cross-cell type imputation of epigenomic profiles.
\newblock \emph{Genome biology}, 24\penalty0 (1):\penalty0 79, 2023.

\bibitem[Shen and Meinshausen(2023)]{shen2023engression}
X.~Shen and N.~Meinshausen.
\newblock Engression: Extrapolation for nonlinear regression?
\newblock \emph{arXiv preprint arXiv:2307.00835}, 2023.

\bibitem[Sinha et~al.(2018)Sinha, Namkoong, and Duchi]{sinha2018certifying}
A.~Sinha, H.~Namkoong, and J.~Duchi.
\newblock Certifying some distributional robustness with principled adversarial
  training.
\newblock In \emph{International Conference on Learning Representations}, 2018.

\bibitem[Sugiyama et~al.(2007)Sugiyama, Krauledat, and
  M{\"u}ller]{sugiyama2007covariate}
M.~Sugiyama, M.~Krauledat, and K.-R. M{\"u}ller.
\newblock Covariate shift adaptation by importance weighted cross validation.
\newblock \emph{Journal of Machine Learning Research}, 8\penalty0
  (35):\penalty0 985--1005, 2007.

\bibitem[Taylor(1715)]{taylor1715}
B.~Taylor.
\newblock Methodus incrementorum directa et inversa [direct and reverse methods
  of incrementation].
\newblock \emph{Pearsonianis prostant apud Gul. Innys}, 1715.

\bibitem[Taylor(2000)]{taylor2000quantile}
J.~W. Taylor.
\newblock A quantile regression neural network approach to estimating the
  conditional density of multiperiod returns.
\newblock \emph{Journal of forecasting}, 19\penalty0 (4):\penalty0 299--311,
  2000.

\bibitem[Wahba(1990)]{wahba1990spline}
G.~Wahba.
\newblock \emph{Spline models for observational data}.
\newblock SIAM, 1990.

\bibitem[Wang et~al.(2022)Wang, He, and Hahn]{wang2022local}
M.~Wang, J.~He, and P.~R. Hahn.
\newblock Local {G}aussian process extrapolation for bart models with
  applications to causal inference.
\newblock \emph{arXiv preprint arXiv:2204.10963}, 2022.

\bibitem[Wang and Lin(2015)]{wang2015derivative}
W.~W. Wang and L.~Lin.
\newblock Derivative estimation based on difference sequence via locally
  weighted least squares regression.
\newblock \emph{Journal of Machine Learning Research}, 16\penalty0
  (1):\penalty0 2617--2641, 2015.

\bibitem[Wilson and Adams(2013)]{wilson2013gaussian}
A.~Wilson and R.~Adams.
\newblock Gaussian process kernels for pattern discovery and extrapolation.
\newblock In \emph{International conference on machine learning}, pages
  1067--1075. PMLR, 2013.

\end{thebibliography}

\newpage
\appendix
\renewcommand{\appendixpagename}{Supplementary material}

\appendixpage

\begin{itemize}
  \item Section~\ref{sec:additional_algorithms}: Additional algorithms
  \item Section~\ref{sec:additional_results}: Additional results from numerical experiments
  \item Section~\ref{sec:details_numerical_experiments}: Details on numerical experiments
  \item Section~\ref{sec:proofs}: Proofs
  \item Section~\ref{sec:auxiliary_results}: Auxiliary results
\end{itemize}

\section{Additional algorithms}\label{sec:additional_algorithms}

\begin{algorithm}
  \caption{\texttt{ParameterTuning}}\label{alg:parameter_tuning}
  \SetKwInOut{Input}{Input}
  \SetKwInOut{Parameters}{\normalfont \textit{Tuning}}
  \SetKwInOut{Output}{Output}

  \Input{Data
    $(X_1,\widehat{\Phi}_n(X_1)),\ldots,(X_n,\widehat{\Phi}_n(X_n))$,
    parameter lists $(\lambda_1,\ldots,\lambda_L)$ and
    $(\Gamma_1,\ldots,\Gamma_K)$} \Parameters{Tolerance \texttt{tol},
    number of folds $K$}
  \Output{Optimal parameters $\lambda^*$ and $\Gamma^*$}

  \BlankLine

  Split indices $\{1,\ldots,n\}$ into $M$ disjoint sets
  $I_1,\ldots, I_M$ of (roughly) equal size

  \For{$k\in\{1,\ldots,K\}$}{
    $\widehat{\mu}\gets$ \texttt{rfpoly}-$v$ on
    $(X_1,\widehat{\Phi}_n(X_1)),\ldots,(X_n,\widehat{\Phi}_n(X_n))$ with
    parameters $\Gamma_k$
  
    Extract weight matrix $W=(w_{i}(X_{\ell}))_{i, \ell}$ from
    $\widehat{\mu}$
    
    \For{$\ell\in\{1,\ldots,L\}$}{
      \For{$m\in\{1,\ldots,M\}$}{
        
        $\hat{f}_m\gets $ local polynomial fit with
        \eqref{eq:local_polynomial_penalized} and penalty
        $\lambda=\lambda_{\ell}$ using
        $(X_i,\widehat{\Phi}_i(X_i))_{i\in I_{-m}}$

        $\hat{Y}_{i}\gets \hat{f}_m(X_i)$ for all $i\in I_m$
      }
      
      $E_{k, \ell, i}\gets L(\hat{Y}_{i},\widehat{\Phi}_n(X_i))$ for
      all $i\in\{1,\ldots,n\}$

      $\overline{E}_{k, \ell}\gets \frac{1}{n}\sum_{i=1}^nE_{k,\ell, i}$
    }
  }
  $(\bar{k}, \bar{\ell})\gets\argmin\{\overline{E}_{k,\ell}\mid (k, \ell)\in\{1,\ldots,K\}\times\{1,\ldots,L\}\}$

  \For{$(k,\ell)\in\{1,\ldots,K\}\times\{1,\ldots,L\}$}{
      $S_{k, \ell}\gets \left(\frac{1}{n}\sum_{i=1}^n(E_{\bar{k},
          \bar{\ell}, i} - E_{k, \ell,i})^2\right)^{\frac{1}{2}}/\sqrt{n}$
  }

  $k^*\gets \min\{k\in\{1,\ldots,K\}\mid \exists
  \ell\in\{1,\ldots,L\}:\, \overline{E}_{k,\ell} \leq \overline{E}_{\bar{k},\bar{\ell}} +
  \texttt{tol}\cdot S_{k, \ell}\}$

  $\ell^*\gets \min\{\ell\in\{1,\ldots,L\}\mid \overline{E}_{k^*,\ell}
  \leq \overline{E}_{\bar{k},\bar{\ell}} + \texttt{tol}\cdot S_{k^*,
    \ell}\}$

  $\Gamma^*\gets \Gamma_{k^*}$ and $\lambda^*\gets \lambda_{\ell^*}$
\end{algorithm}

\begin{algorithm}
\caption{\methodname/ (one-dimensional version)}\label{alg:xtrapolation_onedim}
\SetKwInOut{Input}{Input}
\SetKwInOut{Output}{Output}

\SetKwInOut{Input}{Input}
\SetKwInOut{Parameters}{\normalfont \textit{Tuning}}
\SetKwInOut{Output}{Output}

\Input{Estimates $\widehat{\Phi}_n(X_1),\ldots,\widehat{\Phi}_n(X_n)$,
  data $X_1,\ldots,X_n$, target points $\bar{x}_1,\ldots,\bar{x}_m$, order $k$}

\Parameters{Penalty $\lambda$, \texttt{rf} parameters $\Gamma$}

\Output{Extrapolation bound estimates
  $\widehat{B}^{\operatorname{lo}}(\bar{x}_1),\ldots,
  \widehat{B}^{\operatorname{lo}}(\bar{x}_m)$,
  $\widehat{B}^{\operatorname{up}}(\bar{x}_1),\ldots,
  \widehat{B}^{\operatorname{up}}(\bar{x}_m)$}

\BlankLine

$\mathcal{D}\gets (X_1,\widehat{\Phi}_n(X_1)),\ldots,
(X_1,\widehat{\Phi}_n(X_n))$

\For{$k\in\{1,\ldots,q\}$}{
  $\widehat{\partial^k\Phi}_n(X_1),\ldots,\widehat{\partial^k\Phi}_n(X_n)\gets
  \texttt{RFLocPol}(\mathcal{D}, k=k, v=1, \lambda=\lambda, \Gamma=\Gamma)$
}

\For{$\ell\in\{1,\ldots,m\}$}{
  \For{$i\in\{1,\ldots,n\}$}{
    $S\gets \left\{\widehat{\partial^q\Phi}_n(X_1)\frac{(\bar{x}_{\ell}-X_i)^q}{q!},\ldots,
    \widehat{\partial^q\Phi}_n(X_n)\frac{(\bar{x}_{\ell}-X_i)^q}{q!}\right\}$

    $B_i^{\operatorname{lo}}\gets
    \sum_{k=1}^{q-1}\widehat{\partial^k\Phi}_n(X_i)\frac{(\bar{x}_{\ell}-X_i)^k}{k!}+\min(S)$
    
    $B_i^{\operatorname{up}}\gets
    \sum_{k=1}^{q-1}\widehat{\partial^k\Phi}_n(X_i)\frac{(\bar{x}_{\ell}-X_i)^k}{k!}+\max(S)$
  }
  $\widehat{B}^{\operatorname{lo}}(\bar{x}_{\ell})\gets\max_{i}B_i^{\operatorname{lo}}
  \quad\text{and}\quad
  \widehat{B}^{\operatorname{up}}(\bar{x}_{\ell})\gets\min_{i}B_i^{\operatorname{up}}$
  
  \If{$\widehat{B}^{\operatorname{lo}}(\bar{x}_{\ell})> \widehat{B}^{\operatorname{up}}(\bar{x}_{\ell})$}{
    $\widehat{B}^{\operatorname{lo}}(\bar{x}_{\ell})\gets\big(\widehat{B}^{\operatorname{lo}}(\bar{x}_{\ell})
    + \widehat{B}^{\operatorname{up}}(\bar{x}_{\ell}))/2
    \quad\text{and}\quad
    \widehat{B}^{\operatorname{up}}(\bar{x}_{\ell})\gets\big(\widehat{B}^{\operatorname{lo}}(\bar{x}_{\ell})
    + \widehat{B}^{\operatorname{up}}(\bar{x}_{\ell}))/2$
  }
}

\end{algorithm}

\FloatBarrier
\section{Additional results from numerical
  experiments}\label{sec:additional_results}

  \begin{figure}[h!]
  \centering
  \includegraphics{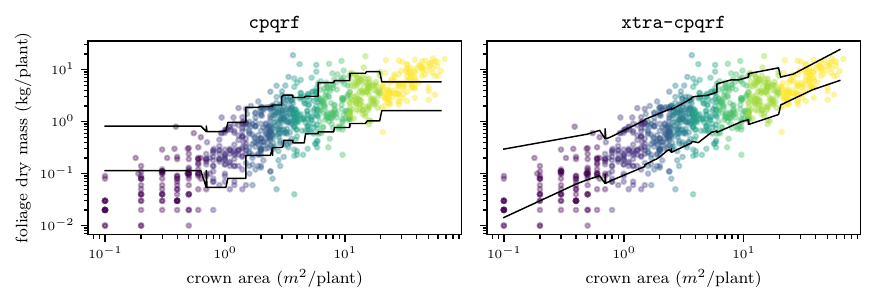}
  \caption{Estimated prediction intervals for \texttt{cpqrf} (left) and
    \texttt{xtra-cpqrf} (right) for the extrapolation split. Colors
    correspond to different extrapolation splits. }
  \label{fig:quantile_scatterplot}
\end{figure}

\begin{figure}[h!]
  \centering
  \includegraphics{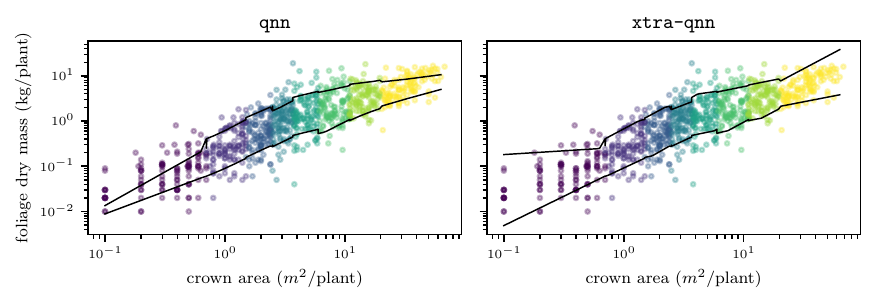}
  \caption{Estimated prediction intervals for \texttt{qnn} (left) and
    \texttt{xtra-qnn} (right) for the extrapolation split. Colors
    correspond to different extrapolation splits.}
  \label{fig:quantile_scatterplot}
\end{figure}

\begin{figure}[h!]
  \centering
  \includegraphics{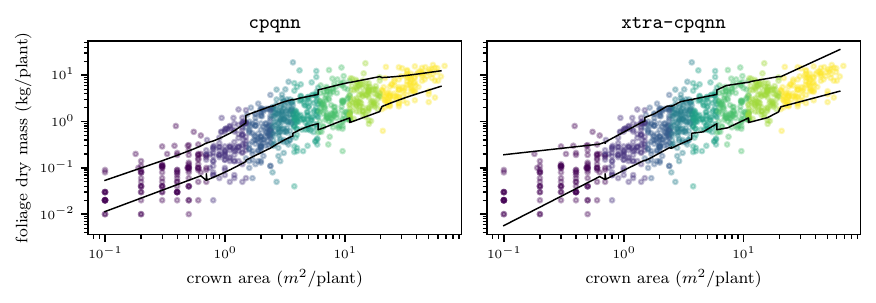}
  \caption{Estimated prediction intervals for \texttt{cpqnn} (left) and
    \texttt{xtra-cpqnn} (right) for the extrapolation split. Colors
    correspond to different extrapolation splits. }
  \label{fig:quantile_scatterplot}
\end{figure}

\FloatBarrier

\section{Details on numerical
  experiments}\label{sec:details_numerical_experiments}

\subsection{Sampling of slopes in simulation}\label{sec:details_slopes}

Let $j\in\{1,\ldots,4\}$ be the index of the interval $I_j$ that was
left out from $C_1$. The slopes $s_1,s_2,s_3,s_4$ are then sampled as
follows. First, for all $\ell\in\{1,\ldots,4\}\setminus\{j\}$
independently of each other sample
$s_{\ell}\sim\operatorname{Unif}([-10, 10])$. Then, randomly draw
$j^*\sim\operatorname{Unif}(\{1,\ldots,4\}\setminus\{j\})$ and set
$s_j\coloneqq s_{j^*}$. This ensures that the slope corresponding to
the left out region of the first coordinate $I_j$ are all observed in
$\Din$, which further guarantees that $f\triangleleft_{\Din}^1f$
(almost every).

\begin{figure}[h!]
  \centering
  \includegraphics{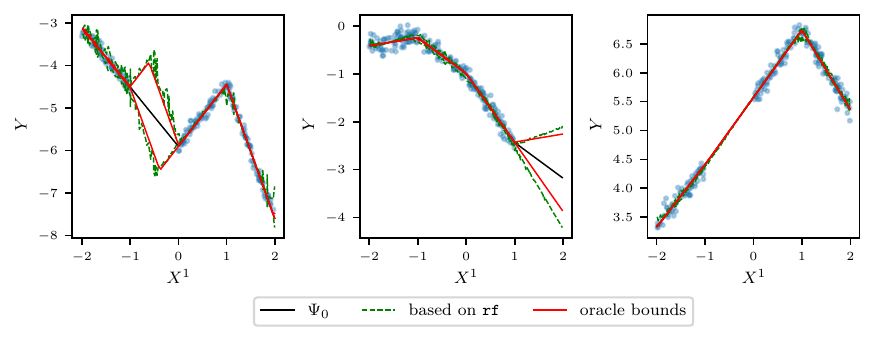}
  \caption{Three simulations generated according to the model
    introduced in Section~\ref{sec:exp_simulations}. Since $\Psi_0$ is
    first derivative extrapolating, it is identified also on $\Dout$
    in the example on the right.}\label{fig:details_slopes}
\end{figure}

\subsection{Hyperparameter selection for regression procedures}\label{sec:details_reg_tuning}

Firstly, for the simulation experiments in
Section~\ref{sec:exp_simulations}, we combine each regression
procedure with a variable screening step and tune the hyperparameters
with cross-validation, to ensure that the regressions procedures
perform well across a large range of settings. More specifically, for
the variable selection we fit a random forest and only keep features
with Gini importance that is larger than the mean of the Gini
importance across all coordinates. The variable screening captures the
sparsity in the simulation setup and improves the predictive
performance if $X$ is multi-dimensional. This variable screening step
is always performed on the same training data on which the subsequent
regression procedure is fitted. For each regressions procedure, we
then construct grids of potential hyperparameters and select the
optimal one based on a $5$-fold cross-validation using the mean
squared error as a score:
\begin{itemize}
\item \texttt{rf}: We used $500$ trees and choose an optimal
  \texttt{max\_depth} from $\{1, 2, 4, 8, 16, \infty\}.$
\item \texttt{svr}: We used a radial basis function kernel and choose
  an optimal bandwidth \texttt{gamma} in
  $\{\frac{10^{-3}}{d}, \frac{10^{-2}}{d}, \frac{10^{-1}}{d},
  \frac{1}{d}, \frac{10^{1}}{d}, \frac{10^{2}}{d}, \frac{10^{3}}{d}\}$
  and an optimal penalty \texttt{C} in
  $\{10^{-3}, 10^{-2}, 10^{-1}, 1, 10, 10^2, 10^3\}$. Furthermore, we
  scaled the data before applying the support vector regression.
\item \texttt{mlp}: We fit a neueral network with 'relu' activations
  using the 'adam' solver and with early stopping. We chose the
  optimal architecture \texttt{hidden\_layer\_sizes} in
  $\{(100,), (20, 20, 20,)\}$ and the optimal penalty \texttt{alpha}
  in $\{10^{-5},10^{-4}, 10^{-3},10^{-2},10^{-1}\}$.
\end{itemize}

Secondly, in the real data applications in Section~\ref{sec:exp_real},
we tuned the hyperparameters of the regression models using a $5$-fold
cross-validation with the average pinball loss for the two target
quantiles as a score. We used the following hyperparameter grids for
the different methods:
\begin{itemize}
\item \texttt{qrf} and \texttt{cpqrf}: We used the grid
  $\{1, 5, 10, 20, 40, 80, 160\}$ for \texttt{min\_samples\_leaf}.
\item \texttt{qnn} and \texttt{cpqnn}: We used the grids
  $\{10^{-4}, 10^{-5}, 10^{-6}, 10^{-7}, 10^{-8}\}$ and
  $\{64, 128, 256\}$ for \texttt{weight\_decay} (specifying the
  penalty) and \texttt{hidden\_size} (specifying the number of hidden
  variables), respectively.
\end{itemize}
In all cases we selected the most regularized parameter that was at
most one standard deviation worse in score than the best model
(similar in spirit to Algorithm~\ref{alg:parameter_tuning} with
$\texttt{tol}=1$). We used the quantile neural network and the
conformalized regression procedures included in the code of
\citet{romano2019conformalized}. The remaining parameters were left at
their default values. In all cases we included a scaling step that
scaled the training data to have mean zero and variance one.

\subsection{Hyperparameter selection for \texttt{Xtrapolation}}\label{sec:details_xtra}

When applying extrapolation (for $q=1$) as described in
Algorithm~\ref{alg:xtrapolation_orderone}, we need to select the
penalty parameter $\lambda$ and the \texttt{rf} parameters
$\Gamma$. To select the optimal parameter we used the parameter tuning
described in Algorithm~\ref{alg:parameter_tuning} with
$\texttt{tol}=1$ and $K=5$ and parameter grids that depend on the
experiment:
\begin{itemize}
\item Simulation experiment (Section~\ref{sec:exp_simulations}): We
  selected $\lambda$ among $\{10, 1, 0.1, 0.01, 0.001, 0\}$ and
  $\Gamma$ among default trees but with \texttt{impurity\_tol} among
  $\{100, 10, 1, 0.1, 0.01\}$.
\item Section~\ref{sec:exp_real} \texttt{biomass} data: We set
  $\lambda=0$ and selected $\Gamma$ among default trees but with
  \texttt{min\_samples\_leaf} among $\{40, 30, 20, 10\}$.
\item Section~\ref{sec:exp_real} \texttt{abalone} data: We selected
  $\lambda$ among $\{0.1, 0.01, 0.001\}$ and $\Gamma$ among default
  trees but with \texttt{min\_samples\_leaf} among $\{10, 5, 1\}$.
\end{itemize}
To speed up the computation, we further only went over a subselection
of all possible anchor points (see
Section~\ref{sec:computational_speedup}). More specifically, for the
experiments in Section~\ref{sec:exp_simulations} we used the $n/2$
closest points in Euclidean distance. For the \texttt{biomass}
application we used all possible anchor points and for the
\texttt{abalone} application we used the $100$ closest (but non-zero
weighted) points based on a random forest closeness measure (required
here as the covariate 'sex' is binary and hence not included in the
extrapolation).

\subsection{Randomized prediction intervals}\label{sec:details_randomized_pi}

Since in both the \texttt{biomass} and \texttt{abalone} data there are
samples with the exact same $Y$ values (relatively common in
\texttt{abalone} since years are counts), we randomized the prediction
intervals in the experiments. This ensures that its always possible to
reach the exact coverage on average. To formalize randomized
prediction intervals, we introduce a random variable
$U\sim\operatorname{Ber}(p)$. Then for any interval
$C=[C_{\operatorname{lo}}, C_{\operatorname{up}}]$ we can define a
corresponding randomized interval $C^{\operatorname{rand}}$ by
\begin{equation*}
    C^{\operatorname{rand}}(U)\coloneqq
    \begin{cases}
        [C_{\operatorname{lo}}, C_{\operatorname{up}}]\quad &\text{if }U=1\\
        (C_{\operatorname{lo}}, C_{\operatorname{up}})\quad &\text{if }U=0.
    \end{cases}
\end{equation*}
For a given prediction interval $\widehat{C}$ which does not have the
correct level due to atoms on the boundaries of the interval, we can
then use the corresponding randomized version
$\widehat{C}^{\operatorname{rand}}(U)$, where we choose the
probability $p$ to calibrate the prediction interval on the training
data to have the correct coverage. For all results, we then report the expected coverage of such randomized prediction intervals.

\section{Proofs}\label{sec:proofs}

\subsection{Proof of Theorem~\ref{thm:taylor_extrapolation_bounds}}\label{proof:taylor_extrapolation_bounds}

\begin{proof}
  We prove the two parts separately.
  
  \textit{Part (i):} First, fix arbitrary $x\in\mathcal{X}$ and let
  $g\in C^{q}(\mathcal{X})$ satisfy for all $x\in\mathcal{D}$ that
  $g(x)=f(x)$ and for all $v\in\mathcal{B}$ that
  \begin{equation}
    \label{eq:dominating_equation_proof}
    \inf_{x\in\mathcal{X}}D_v^qg(x)\geq \inf_{x\in\mathcal{D}}D_v^qf(x)
    \quad\text{and}\quad
    \sup_{x\in\mathcal{X}}D_v^qg(x)\leq \sup_{x\in\mathcal{D}}D_v^qf(x).
  \end{equation}
  Then, using \eqref{eq:taylor_equation}, it holds for all
  $x_0\in\mathcal{D}$ that there exists $\xi_{x_0}\in\mathcal{X}$ such
  that
  \begin{align}
    g(x)
    &=\sum_{\ell=0}^{q-1}D_{\overline{v}(x_0,
      x)}^{\ell}g(x_0)\frac{\|x-x_0\|_2^{\ell}}{\ell!} +
      D_{\overline{v}(x_0, x)}^{q}g(\xi_{x_0})
      \frac{\|x-x_0\|_2^{q}}{q!}\nonumber\\
    &=\sum_{\ell=0}^{q-1}D_{\overline{v}(x_0,
      x)}^{\ell}f(x_0)\frac{\|x-x_0\|_2^{\ell}}{\ell!} +
      D_{\overline{v}(x_0, x)}^{q}g_n(\xi_{x_0})
      \frac{\|x-x_0\|_2^{q}}{q!},\label{eq:upper_bound_part0}
  \end{align}
  where we used that $g=f$ on $\mathcal{D}$. Moreover, using
  \eqref{eq:dominating_equation_proof}, we get for all
  $x_0\in\mathcal{D}$ that
  \begin{equation}
    \label{eq:upper_bound_part2}
    D_{\overline{v}(x_0, x)}^{q}g(\xi_{x_0})
    \leq\sup_{z\in\mathcal{X}}D_{\overline{v}(x_0, x)}^{q}g(z)
    \leq\sup_{z\in\mathcal{D}}D_{\overline{v}(x_0, x)}^{q}f(z).
  \end{equation}
  Hence, together with \eqref{eq:upper_bound_part0} we have for all
  $x_0\in\mathcal{D}$ that
  \begin{equation}
    g(x)\leq \sum_{\ell=0}^{q-1}D_{\overline{v}(x_0,
        x)}^{\ell}f(x_0)\frac{\|x-x_0\|_2^{\ell}}{\ell!} +
      \sup_{z\in\mathcal{D}}D_{\overline{v}(x_0, x)}^{q}f(z)
      \frac{\|x-x_0\|_2^{q}}{q!}.\label{eq:upper_bound_part3}
  \end{equation}
  Finally, taking the infimum over $x_0\in\mathcal{D}$ on both sides
  results in
  \begin{equation*}
    g(x)\leq B_{q, f, \mathcal{D}}^{\operatorname{up}}(x).
  \end{equation*}
  The same argument also applies to the lower extrapolation bound,
  which completes the proof of part~(i).

  \textit{Part (ii):} We only prove the result for the upper bound, the
  same arguments apply to the lower bound. Since we assume that
  $\mathcal{X}$ is compact it holds that $\mathcal{D}$ is also
  compact. We can therefore define the function
  $\optx:\mathcal{X}\rightarrow\mathcal{D}$ such that for all
  $x\in\mathcal{X}$ it holds that $\optx(x)\in\mathcal{D}$ satisfies
  \begin{equation*}
    B_{q, f,
      \mathcal{D}}^{\operatorname{up}}(x)=\sum_{\ell=0}^{q-1}D_{\overline{v}(\optx(x),
      x)}^{\ell}f(\optx(x))\frac{\|x-\optx(x)\|_2^{\ell}}{\ell!} +
    \sup_{z\in\mathcal{D}}D^q_{\overline{v}(\optx(x),
      x)}f(z)
    \frac{\|x-\optx(x)\|_2^{q}}{q!}.
  \end{equation*}
  Moreover, define function
  $\optz:\mathcal{X}\rightarrow\mathcal{D}$ such that
  for all $x\in\mathcal{X}$ it holds that
  $\optz(x)\in\mathcal{D}$ satisfies
  \begin{equation*}
    \sup_{z\in\mathcal{D}}D^q_{\overline{v}(\optx(x),
      x)}f(z)
    =D^q_{\overline{v}(\optx(x),
      x)}f(\optz(x)).
  \end{equation*}
  In the following, we construct a q-times differentiable sequence
  $(g^{\operatorname{up}}_n)_{n\in\mathbb{N}}$ which approximates the
  extrapolation bound $B^{\operatorname{up}}_{q,f,\mathcal{D}}$, is
  $q$-times continuously differentiable and satisfies
  $g^{\operatorname{up}}_n\triangleleft_{\mathcal{D}}^q
  f$. Constructing such a sequence directly is difficult due to the
  two optimizations $\optx$ and $\optz$. We therefore construct the
  approximation in two steps.

  \textit{First approximation step:} Fix arbitrary
  $\varepsilon>0$. For all $z\in\mathcal{X}$, using multi-index
  notation, define the functions
  $B_{z}:\mathbb{R}^d\rightarrow\mathbb{R}$ for all $x\in\mathbb{R}^d$
  by
  \begin{equation}
    \label{eq:B_z_delta}
    B_{z}(x)\coloneqq \sum_{|\balpha|<
      q}\partial^{\balpha}f(\optx(z))\cdot\frac{(x-\optx(z))^{\balpha}}{\balpha!}
    +\sum_{|\balpha|=q}\partial^{\balpha}f(\optz(z))\cdot\frac{(x-\optx(z))^{\balpha}}{\balpha!}.
  \end{equation}
  Since for fixed $z\in\mathcal{X}$ the points $\optx(z)$ and
  $\optz(z)$ are fixed it is easier to analyze $B_z$, which is
  multivariate polynomial of degree at most $q$. Moreover, for all
  $v\in\mathcal{B}$, $\ell\in\{0,\ldots,d\}$ and $x\in\mathcal{X}$ the
  directional derivative satisfies
  $D_{v}^{\ell}f(x)=\sum_{|\balpha|=\ell}\frac{\ell!}{\balpha!}\partial^{\balpha}f(x)v^{\balpha}$.
  This implies for all $z\in\mathcal{B}$ that
  \begin{equation}
    \label{eq:Bz_equals_Bup}
    B_z(z)=B^{\operatorname{up}}_{q,f,\mathcal{D}}(z).
  \end{equation}
  More properties of
  $B_z$ are listed in
  Lemma~\ref{thm:properties_first_approximation}.

  Consider now the collection of open sets
  $\{\mathcal{U}_{\varepsilon}(x)\mid x\in\mathcal{X}\}$, where
  $\mathcal{U}_{\varepsilon}(x)\coloneqq\{y\in\mathbb{R}^d\mid
  \|x-y\|_2<\varepsilon\}$. Since, $\mathcal{X}$ is compact, there
  exists a finite set of points
  $x_1^{\varepsilon},\ldots,x_{M_{\varepsilon}}^{\varepsilon}$ such
  that
  $\{\mathcal{U}_{\varepsilon}(x_{\ell}^{\varepsilon})\mid
  \ell\in\{1,\ldots,M_{\varepsilon}\}\}$ covers $\mathcal{X}$. We can
  therefore define the function
  $g_{\varepsilon}:\mathbb{R}^d\rightarrow\mathbb{R}$ for all
  $x\in\mathbb{R}^d$ by
  \begin{equation*}
    g_{\varepsilon}(x)\coloneqq \frac{1}{w_{\varepsilon}(x)}\sum_{\ell=1}^{M_{\varepsilon}}B_{x_{\ell}^{\varepsilon}}(x)\mathds{1}_{\mathcal{U}_{\varepsilon}(x_{\ell}^{\varepsilon})}(x),
  \end{equation*}
  where $w_{\varepsilon}(x)\coloneqq|\{\ell\in\{1,\ldots,M_{\varepsilon}\}\mid x\in
  \mathcal{U}_{\varepsilon}(x_{\ell}^{\varepsilon})\}|$. We now prove that there exists
  $K_1>0$ such that
  \begin{equation}
    \label{eq:step1_approximation}
    \sup_{x\in\mathcal{X}}\left|B_{q,f,\mathcal{D}}^{\operatorname{up}}(x)-g_{\varepsilon}(x)\right|\leq
    K_1 \varepsilon.
  \end{equation}
  To see this, fix $z\in\mathcal{X}$, then using the triangle
  inequality and \eqref{eq:Bz_equals_Bup} it holds that
  \begin{align}
    \sup_{x\in\mathcal{U}_{\varepsilon}(z)}\big|B_{q,f,\mathcal{D}}^{\operatorname{up}}(x)-B_{z}(x)\big|
    &\leq\sup_{x\in\mathcal{U}_{\varepsilon}(z)}\big|B_{q,f,\mathcal{D}}^{\operatorname{up}}(x)-B_{q,f,\mathcal{D}}^{\operatorname{up}}(z)\big| +\sup_{x\in\mathcal{U}_{\varepsilon}(z)}\big|B_{z}(z) -
      B_{z}(x)\big|\nonumber\\
    &\leq (c_1+c_2)\varepsilon,\label{eq:Bx_uniform_bound}
  \end{align}
  where $c_1>0$ and $c_2>0$ exists by
  Lemma~\ref{thm:uniform_continuity_of_bounds} and
  Lemma~\ref{thm:properties_first_approximation}~(ii). Using
  \eqref{eq:Bx_uniform_bound} we further get
  \begin{align*}
    \sup_{x\in\mathcal{X}}\left|B_{q,f,\mathcal{D}}^{\operatorname{up}}(x)-g_{\varepsilon}(x)\right|
    &= \sup_{x\in\mathcal{X}}\left|B_{q,f,\mathcal{D}}^{\operatorname{up}}(x)-\frac{1}{w_{\varepsilon}(x)}\sum_{\ell=1}^{M_{\varepsilon}}B_{x_{\ell}^{\varepsilon}}(x)\mathds{1}_{\mathcal{U}_{\varepsilon}(x_{\ell}^{\varepsilon})}(x)\right|\\
    & \leq \sup_{x\in\mathcal{X}} \left(\frac{1}{w_{\varepsilon}(x)}\sum_{\ell=1}^{M_{\varepsilon}}\left|B_{q,f,\mathcal{D}}^{\operatorname{up}}(x)-B_{x_{\ell}^{\varepsilon}}(x)\right|\mathds{1}_{\mathcal{U}_{\varepsilon}(x_{\ell}^{\varepsilon})}(x)\right)\\
    & \leq \sup_{x\in\mathcal{X}} \left(\frac{1}{w_{\varepsilon}(x)}\sum_{\ell=1}^{M_{\varepsilon}}\mathds{1}_{\mathcal{U}_{\varepsilon}(x_{\ell}^{\varepsilon})}(x)\right)
     (c_1+c_2)\varepsilon\\
    & = \underbrace{(c_1+c_2)}_{\eqqcolon K_1}\varepsilon,
  \end{align*}
  which proves \eqref{eq:step1_approximation}. Furthermore, notice
  that the set
  $\mathcal{W}\coloneqq\mathbb{R}^d\setminus(\partial\mathcal{U}_{\varepsilon}(x_{1}^{\varepsilon})\cup\cdots\cup\partial\mathcal{U}_{\varepsilon}(x_{M_{\varepsilon}}^{\varepsilon}))$
  is open and hence for all $z\in\mathcal{W}$ there exists
  $\kappa>0$ such that
  $\mathcal{U}_{\kappa}(z)\subseteq\mathcal{W}$. Moreover, for
  all $x\in \mathcal{U}_{\kappa}(z)$ it holds that
  \begin{equation}
    \label{eq:diff_of_geps}
    g_{\varepsilon}(x)=\frac{1}{w_{\varepsilon}(x)}\sum_{\ell=1}^{M_{\varepsilon}}B_{x_{\ell}^{\varepsilon}}(x)\mathds{1}_{\mathcal{U}_{\varepsilon}(x_{\ell}^{\varepsilon})}(x)=\frac{1}{w_{\varepsilon}(z)}\sum_{\ell=1}^{M_{\varepsilon}}B_{x_{\ell}^{\varepsilon}}(x)\mathds{1}_{\mathcal{U}_{\varepsilon}(x_{\ell}^{\varepsilon})}(z).
  \end{equation}
  Since $B_{x_{\ell}^{\varepsilon}}$ is $q$-times continuously
  differentiable and since $\mathbb{R}^d\setminus\mathcal{W}$ is a
  Lebesgue null-set this implies that $g_{\varepsilon}$ is Lebesgue
  almost everywhere $q$-times continuously differentiable. In
  particular, using \eqref{eq:diff_of_geps} and
  Lemma~\ref{thm:properties_first_approximation}~(i), it holds for every
  multi-index $\balpha\in\mathbb{N}^d$ with $|\balpha|=q$ and
  every $x\in\mathcal{W}$ that
  \begin{equation}
    \label{eq:derivative_g_delta_general}
    \partial^{\balpha}g_{\varepsilon}(x)
    =\frac{1}{w_{\varepsilon}(x)}\sum_{\ell=1}^{M_{\varepsilon}}\partial^{\balpha}f(\optz(x_{\ell}^{\varepsilon}))\mathds{1}_{\mathcal{U}_{\varepsilon}(x_{\ell}^{\varepsilon})}(x).
  \end{equation}

  \textit{Second approximation step:} Again let $\varepsilon>0$ be
  arbitrary. Denote by
  $\eta_{\varepsilon}:\mathbb{R}^d\rightarrow\mathbb{R}$ the standard
  mollifier defined for all $x\in\mathbb{R}^d$ by
  \begin{equation*}
    \eta_{\varepsilon}(x)\coloneqq
    \begin{cases}
      C\exp\left(1/\left(\|x\|_2^2-\varepsilon^2\right)\right) &\quad\text{if
                                      }x\in\mathcal{U}_{\varepsilon}(0)\\
      0 &\quad\text{otherwise},
    \end{cases}
  \end{equation*}
  where $C>0$ is chosen such that
  $\int_{\mathbb{R}^d}\eta_{\epsilon}(x)\text{d}x=1$. For any
  integrable function $h:\mathbb{R}^d\rightarrow\mathbb{R}$, we define
  the mollified function $h * \eta_{\varepsilon}$ for all
  $x\in\mathbb{R}^d$ by the convolution
  \begin{equation*}
    (h * \eta_{\varepsilon})(x)\coloneqq \int_{\mathcal{U}_{\varepsilon}(0)}h(x-y)\eta_{\varepsilon}(y)\text{d}y.
  \end{equation*}
  We now define the approximation function
  $\bar{g}_{\varepsilon}:\mathbb{R}^d\rightarrow\mathbb{R}$ for all
  $x\in\mathbb{R}^d$ by
  \begin{equation*}
    \bar{g}_{\varepsilon}(x)\coloneqq (g_{\varepsilon}*\eta_{\varepsilon})(x).
  \end{equation*}
  Using that the convolution with a mollifier is smooth
  \citep[Theorem~1.3.1]{hormander2015analysis}, we get that
  $\bar{g}_{\varepsilon}$ is $q$-times continuously
  differentiable. Furthermore, using Jensen's inequality and the
  triangle inequality, we can bound the approximation error of the
  mollification as follows,
  \begin{align}
    &\sup_{x\in\mathcal{X}}\left|\bar{g}_{\varepsilon}(x)-g_{\varepsilon}(x)\right|\nonumber\\
    &\quad=\sup_{x\in\mathcal{X}}\left|(g_{\varepsilon}*\eta_{\varepsilon})(x)-g_{\varepsilon}(x)\right|\nonumber\\
    &\quad\leq\sup_{x\in\mathcal{X}}\int_{\mathcal{U}_{\varepsilon}(0)}\left|g_{\varepsilon}(x-y)-g_{\varepsilon}(x)\right|\eta_{\varepsilon}(y)\text{d}y\nonumber\\
    &\quad\leq\sup_{\substack{x,y\in\mathcal{X}:\\ \|x-y\|_2\leq\varepsilon}}\left|g_{\varepsilon}(y)-g_{\varepsilon}(x)\right|\nonumber\\
    &\quad=\sup_{\substack{x,y\in\mathcal{X}:\\ \|x-y\|_2\leq\varepsilon}}\left|g_{\varepsilon}(y)-B_{q,f,\mathcal{D}}^{\operatorname{up}}(y)
    - g_{\varepsilon}(x)+B_{q,f,\mathcal{D}}^{\operatorname{up}}(x) +
    B_{q,f,\mathcal{D}}^{\operatorname{up}}(y) - B_{q,f,\mathcal{D}}^{\operatorname{up}}(x)\right|\nonumber\\
    &\quad\leq
      2\sup_{x\in\mathcal{X}}\left|g_{\varepsilon}(x)-B_{q,f,\mathcal{D}}^{\operatorname{up}}(x)\right|
      + \sup_{\substack{x,y\in\mathcal{X}:\\
    \|x-y\|_2\leq\varepsilon}}\left|B_{q,f,\mathcal{D}}^{\operatorname{up}}(x)-B_{q,f,\mathcal{D}}^{\operatorname{up}}(y)\right|\nonumber\\
    &\quad\leq \underbrace{(2K_1+c_1)}_{\eqqcolon K_2}\varepsilon,\label{eq:step2_approximation}
  \end{align}
  where we used \eqref{eq:step1_approximation} and
  Lemma~\ref{thm:uniform_continuity_of_bounds} in the
  last step.

  Let $\balpha\in\mathbb{N}^d$ with $|\balpha|=q$, then using that all
  partial derivatives of $g_{\varepsilon}*\eta_{\varepsilon}$ can be
  bounded on $\overline{\mathcal{U}_{\varepsilon}(x)}$, we can
  pull derivatives under the integral
  \citep[e.g.,][Theorem~2.27]{folland1999real}. Together with \eqref{eq:derivative_g_delta_general}, we hence get
  \begin{align*}
    \partial^{\balpha}\bar{g}_{\varepsilon}(x)
    &=\partial^{\balpha}(g_{\varepsilon}*\eta_{\varepsilon})(x)\\
    &=\partial^{\balpha}\int_{\mathcal{U}_{\varepsilon}(0)}g_{\varepsilon}(x-y)\eta_{\varepsilon}(y)\text{d}y\\
    &=\int_{\mathcal{U}_{\varepsilon}(0)}(\partial^{\balpha}g_{\varepsilon})(x-y)\eta_{\varepsilon}(y)\text{d}y\\
    &=\int_{\mathcal{U}_{\varepsilon}(0)}\frac{1}{w_{\varepsilon}(x-y)}\sum_{\ell=1}^{M_{\varepsilon}}\partial^{\balpha}f(\optz(x_{\ell}^{\varepsilon}))\mathds{1}_{\mathcal{U}_{\varepsilon}(x_{\ell}^{\varepsilon})}(x-y)\eta_{\varepsilon}(y)\text{d}y\\
    &=\sum_{\ell=1}^{M_{\varepsilon}}\partial^{\balpha}f(\optz(x_{\ell}^{\varepsilon}))\underbrace{\int_{\mathcal{U}_{\varepsilon}(0)}\frac{1}{w_{\varepsilon}(x-y)}\mathds{1}_{\mathcal{U}_{\varepsilon}(x_{\ell}^{\varepsilon})}(x-y)\eta_{\varepsilon}(y)\text{d}y}_{\eqqcolon
      \bar{w}_{\ell}(x)}.
  \end{align*}
  Hence expressing the directional derivative in terms of partial
  derivatives we get for all $v\in\mathcal{B}$ and all
  $x\in\mathcal{X}$ that
  \begin{align*}
    D_{v}^q\bar{g}_{\varepsilon}(x)
    &=\sum_{|\balpha|=q}\frac{q!}{\balpha!}\partial^{\balpha}\bar{g}_{\varepsilon}(x)
      v^{\balpha}\\
    &=\sum_{\ell=1}^{M_{\varepsilon}}\sum_{|\balpha|=q}\frac{q!}{\balpha!}\partial^{\balpha}f(\optz(x_{\ell}^{\varepsilon}))v^{\balpha}\bar{w}_{\ell}(x)\\
    &=\sum_{\ell=1}^{M_{\varepsilon}}D_{v}^qf(\optz(x_{\ell}^{\varepsilon}))\bar{w}_{\ell}(x)
  \end{align*}
  This further implies for all $v\in\mathcal{B}$ and all
  $x\in\mathcal{X}$ that
  \begin{equation*}
    \inf_{z\in\mathcal{D}}D_v^qf(z)=\inf_{z\in\mathcal{D}}D_v^qf(z)
    \sum_{\ell=1}^{M_{\varepsilon}}\bar{w}_{\ell}(x)\leq D_{v}^q\bar{g}_{\varepsilon}(x)\leq
    \sup_{z\in\mathcal{D}}D_v^qf(z) \sum_{\ell=1}^{M_{\varepsilon}}\bar{w}_{\ell}(x)=\sup_{z\in\mathcal{D}}D_v^qf(z).
  \end{equation*}
  Hence, we have shown that
  \begin{equation}
    \label{eq:minmax_dominated_approx}
    \bar{g}_{\varepsilon}
    \triangleleft_{\mathcal{D}}^q f.
  \end{equation}
  \textit{Conclude by constructing $g^{\operatorname{up}}_n$:} Define
  for all $n\in\mathbb{N}$ the functions
  \begin{equation*}
    g^{\operatorname{up}}_n\coloneqq \bar{g}_{\varepsilon_n}
    \quad\text{with}\quad \varepsilon_n\coloneqq \frac{1}{n(K_1+K_2)}.
  \end{equation*}
  Then, with \eqref{eq:step1_approximation} and
  \eqref{eq:step2_approximation} it holds that
  \begin{align}
    \sup_{x\in\mathcal{X}}|g^{\operatorname{up}}_n(x)-B_{q,f,\mathcal{D}}^{\operatorname{up}}(x)|
    &=\sup_{x\in\mathcal{X}}|\bar{g}_{\varepsilon_n}(x)-B_{q,f,\mathcal{D}}^{\operatorname{up}}(x)|\nonumber\\
    &\leq
    \sup_{x\in\mathcal{X}}|\bar{g}_{\varepsilon_n}(x)-g_{\varepsilon_n}(x)|
    + \sup_{x\in\mathcal{X}}|g_{\varepsilon_n}(x)-B_{q,
      f,\mathcal{D}}^{\operatorname{up}}(x)|\nonumber\\
    &\leq \varepsilon_n K_2 + \varepsilon_n K_1 =
      \frac{1}{n}.\label{eq:convergence_of_hn}
  \end{align}
  By \eqref{eq:minmax_dominated_approx} it further holds for all
  $n\in\mathbb{N}$ that
  $g^{\operatorname{up}}_n \triangleleft_{\mathcal{D}}^q f$, which
  completes the proof of part~(ii) and hence also of
  Theorem~\ref{thm:taylor_extrapolation_bounds}.
\end{proof}

\subsection{Proof of
  Proposition~\ref{thm:worst_case_prediction}}\label{proof:worst_case_prediction}

\begin{proof}
  Fix $x\in\mathcal{X}$ throughout the proof. For all
  $Q\in\mathcal{Q}_0$, denote by
  $\Psi_Q:\mathcal{X}\rightarrow\mathbb{R}$ the conditional
  expectation function corresponding to the Markov kernel $Q$. Then
  for all $Q\in\mathcal{Q}_0$ there exists a mean-zero random variable
  $U_x$ such that $Y_x=\Psi_Q(x) + U_x$.

  We now first construct an upper bound on the right-hand side of
  \eqref{eq:worst_case_optimal_result}. To this end, fix
  $Q\in\mathcal{Q}_0$, then it holds that
  $\Psi_Q(x)\in\left[B^{\operatorname{lo}}_{\Psi_{0},\Din}(x),\,B^{\operatorname{up}}_{\Psi_{0},\Din}(x)\right]$. Hence,
  we get that
  \begin{equation*}
    \left|\Psi_Q(x)-\frac{1}{2}\left(B^{\operatorname{lo}}_{\Psi_{0},\Din}(x)+B^{\operatorname{up}}_{\Psi_{0},\Din}(x)\right)\right|
    \leq \frac{1}{2}\left(B^{\operatorname{up}}_{\Psi_{0},\Din}(x)-B^{\operatorname{lo}}_{\Psi_{0},\Din}(x)\right).
  \end{equation*}
  Using this bound and $\mathbb{E}_Q[U_x]=0$, we further get
  \begin{align*}
    \mathbb{E}_{Q}[(Y_x-f^*(x))^2]
    &=(\Psi_Q(x)-f^*(x))^2 + \mathbb{E}_{Q}[U_x^2]\\
    &=\left(\Psi_Q(x)-\frac{1}{2}\left(B^{\operatorname{lo}}_{\Psi_{0},\Din}(x)+B^{\operatorname{up}}_{\Psi_{0},\Din}(X)\right)\right)^2
      + \mathbb{E}_{Q}[U_x^2]\\
    &\leq\frac{1}{4}\left(B^{\operatorname{up}}_{\Psi_{0},\Din}(x)-B^{\operatorname{lo}}_{\Psi_{0},\Din}(x)\right)^2
      + 
      \mathbb{E}_{Q}[U_x^2].
  \end{align*}
  Since $Q\in\mathcal{Q}_0$ was arbitrary this implies
  \begin{equation}
    \label{eq:final_upperbound}
    \sup_{Q\in\mathcal{Q}_0}\mathbb{E}_Q[(Y_x-f^*(x))^2]
    \leq\frac{1}{4}\left(B^{\operatorname{up}}_{\Psi_{0},\Din}(x)-B^{\operatorname{lo}}_{\Psi_{0},\Din}(x)\right)^2
    + \sup_{Q\in\mathcal{Q}_0}\mathbb{E}_{Q}[U_x^2].
  \end{equation}
  
  Next, we derive a lower bound. To this end, fix an arbitrary
  $f\in C^0(\mathcal{X})$ and $Q\in\mathcal{Q}_0$. The triangle
  inequality implies
  \begin{equation}
    \label{eq:triangle_uplo}
    \left|B^{\operatorname{up}}_{\Psi_{0},\Din}(x)-f(x)\right|+\left|B^{\operatorname{lo}}_{\Psi_{0},\Din}(x)-f(x)\right|
    \geq\left|(B^{\operatorname{up}}_{\Psi_{0},\Din}(x)-B^{\operatorname{lo}}_{\Psi_{0},\Din}(x)\right|.
  \end{equation}
  By Theorem~\ref{thm:taylor_extrapolation_bounds} it holds that
  $B^{\operatorname{lo}}_{\Psi_{0},\Din},
  B^{\operatorname{up}}_{\Psi_{0},\Din}\in C^0(\mathcal{X})$ and hence
  there exist
  $Q_{\operatorname{lo}},Q_{\operatorname{up}}\in\mathcal{Q}_0$ such
  that
  $\Psi_{Q_{\operatorname{lo}}}\equiv
  B^{\operatorname{lo}}_{\Psi_{0},\Din}$ and
  $\Psi_{Q_{\operatorname{up}}}\equiv
  B^{\operatorname{up}}_{\Psi_{0},\Din}$.  This, implies that
  \begin{align}
    \sup_{Q\in\mathcal{Q}_0}\left(\Psi_Q(x)-f(x)\right)
    &\geq\max\left((B^{\operatorname{up}}_{\Psi_{0},\Din}(x)-f(x))^2,
      (B^{\operatorname{lo}}_{\Psi_{0},\Din}(x)-f(x))^2\right)\nonumber\\
    &\geq\frac{1}{2}\left((B^{\operatorname{up}}_{\Psi_{0},\Din}(x)-f(x))^2 +
      (B^{\operatorname{lo}}_{\Psi_{0},\Din}(x)-f(x))^2\right)\nonumber\\
    &\geq\frac{1}{4}\left(|B^{\operatorname{up}}_{\Psi_{0},\Din}(x)-f(x)| +
      |B^{\operatorname{lo}}_{\Psi_{0},\Din}(x)-f(x)|\right)^2\nonumber\\
    &\geq\frac{1}{4}\left(B^{\operatorname{up}}_{\Psi_{0},\Din}(x) - B^{\operatorname{lo}}_{\Psi_{0},\Din}(x)\right)^2,\label{eq:lower_bound_quarter}
  \end{align}
  where for the third inequality we used $(a+b)^2\leq 2a^2 + 2b^2$ and
  \eqref{eq:triangle_uplo} for the last inequality. Therefore together
  with \eqref{eq:lower_bound_quarter} it holds that
  \begin{equation}
    \label{eq:final_lowerbound}
    \sup_{Q\in\mathcal{Q}_0}\mathbb{E}[(Y_x-f(x))^2]\geq\frac{1}{4}\left(B^{\operatorname{up}}_{\Psi_{0},\Din}(x)-B^{\operatorname{lo}}_{\Psi_{0},\Din}(x)\right)^2
    + \sup_{Q\in\mathcal{Q}_0}\mathbb{E}_Q[U_x].
  \end{equation}
  Combining \eqref{eq:final_upperbound} and
  \eqref{eq:final_lowerbound} completes the proof of
  Proposition~\ref{thm:worst_case_prediction}.
\end{proof}

\subsection{Proof of
  Proposition~\ref{thm:confidence_intervals}}\label{proof:confidence_intervals}

\begin{proof}
  Since $\Psi_0$ is assumed to be $q$-th derivative extrapolating it holds by Theorem~\ref{thm:taylor_extrapolation_bounds} that $\Psi_0\in[B^{\operatorname{lo}}_{\Psi_0,
   \Din}(x), B^{\operatorname{up}}_{\Psi_0,
   \Din}(x)]$. We can use this together with standard probability inequalities to get
  \begin{align*}
  &\mathbb{P}(\Psi_0(x)\in\widehat{\operatorname{C}}_{n;
      \alpha}^{\operatorname{conf}}(x))\\
  &\quad=\mathbb{P}\left(\Psi_0(x)\in\left[\widehat{G}_{n}^{\operatorname{lo}}(\tfrac{\alpha}{2},x),\, \widehat{G}_{n}^{\operatorname{up}}(1-\tfrac{\alpha}{2},x)\right]\right)\\
  &\quad=\mathbb{P}\left(\Psi_0(x)\geq \widehat{G}_{n}^{\operatorname{lo}}(\tfrac{\alpha}{2},x),\, \Psi_0(x)\leq \widehat{G}_{n}^{\operatorname{up}}(1-\tfrac{\alpha}{2},x)\right)\\
  &\quad\geq\mathbb{P}\left(B^{\operatorname{lo}}_{\Psi_0,
   \Din}(x)> \widehat{G}_{n}^{\operatorname{lo}}(\tfrac{\alpha}{2},x),\, B^{\operatorname{up}}_{\Psi_0,
      \Din}(x)\leq \widehat{G}_{n}^{\operatorname{up}}(1-\tfrac{\alpha}{2},x)\right)\\
  &\quad\geq 1 - \mathbb{P}\left(B^{\operatorname{lo}}_{\Psi_0,
   \Din}(x)\leq \widehat{G}_{n}^{\operatorname{lo}}(\tfrac{\alpha}{2},x)\right) - \mathbb{P}\left(B^{\operatorname{up}}_{\Psi_0,
   \Din}(x)> \widehat{G}_{n}^{\operatorname{up}}(1-\tfrac{\alpha}{2},x)\right)\\
   &\quad= - \mathbb{P}\left(B^{\operatorname{lo}}_{\Psi_0,
   \Din}(x)\leq \widehat{G}_{n}^{\operatorname{lo}}(\tfrac{\alpha}{2},x)\right) +\mathbb{P}\left(B^{\operatorname{up}}_{\Psi_0,
   \Din}(x)\leq\widehat{G}_{n}^{\operatorname{up}}(1-\tfrac{\alpha}{2},x)\right).
  \end{align*}
  Next, taking the $\liminf$ as $n$ goes to infinity on both sides together with the assumptions on $\widehat{G}^{\star}_n$, we get that
  \begin{equation*}
      \liminf_{n\rightarrow\infty}\mathbb{P}(\Psi_0(x)\in\widehat{\operatorname{C}}_{n;
      \alpha}^{\operatorname{conf}}(x))\geq 1-\alpha.
  \end{equation*}
  This completes the proof of Proposition~\ref{thm:confidence_intervals}.
\end{proof}

\subsection{Proof of
  Proposition~\ref{thm:prediction_intervals}}\label{proof:prediction_intervals}

\begin{proof}
  By Theorem~\ref{thm:taylor_extrapolation_bounds}, since
  $\mathcal{T}^{\alpha/2}_{0}$ and $\mathcal{T}^{1-\alpha/2}_{0}$ are
  both $q$-th derivative extrapolating, it holds that
  \begin{equation*}
    \left[\mathcal{T}_{0}^{\alpha/2}(x),\, \mathcal{T}_{0}^{1-\alpha/2}(x)\right]\subseteq\operatorname{C}_{\alpha}^{\operatorname{pred}}(x).
  \end{equation*}
  Using the definition of the conditional quantile function, we
  therefore get that
  \begin{equation*}
    \mathbb{P}_{Q_0(x,
      \cdot)}(Y_x\in\operatorname{C}_{\alpha}^{\operatorname{pred}}(x))\geq
    \mathbb{P}_{Q_0(x,
      \cdot)}(Y_x\in [\mathcal{T}_{0}^{\alpha/2}(x),\, \mathcal{T}_{0}^{1-\alpha/2}(x)])=
    1-\alpha.
  \end{equation*}
  This completes the proof of
  Proposition~\ref{thm:prediction_intervals}. 
\end{proof}

\subsection{Proof of
  Theorem~\ref{thm:consistency}}\label{proof:consistency}

\begin{proof}
  Recall, that the directional derivative can be expressed in
  multi-index notation as
  $D_v^{\ell}f(z)=\sum_{\balpha=\ell}\partial^{\balpha}f(z)v^{\balpha}\frac{\ell!}{\balpha!}$. Furthermore,
  by compactness of $\mathcal{X}$ it holds that
  $K\coloneqq \sup_{x,y\in\mathcal{X}}\|x-y\|_2<\infty$.  Throughout,
  the proof we fix $x\in\mathcal{X}$. We begin by defining for all
  $w,z\in\mathcal{X}$ the functions
  \begin{equation*}
    F_1(w)\coloneqq \sum_{\ell=0}^{q-1}\sum_{|\balpha|=\ell}\partial^{\balpha}\Phi_{0}(w)\frac{(x-w)^{\balpha}}{\balpha!}
    \quad\text{and}\quad
    F_2(w,z)\coloneqq \sum_{|\balpha|=q}\partial^{\balpha}\Phi_{0}(z)\frac{(x-w)^{\balpha}}{\balpha!}.
  \end{equation*}
  Then, it holds, using the multi-index expression of the directional
  derivative, that
  \begin{equation*}
    B^{\operatorname{up}}_{\Phi_{0},\Din}(x)=\inf_{w\in\Din}\left(F_1(w)+\sup_{z\in\Din}F_2(w,z)\right).
  \end{equation*}
  Similarly define
  for all $w,z\in\mathcal{X}$ the functions
  \begin{equation*}
    \widehat{F}_1(w)\coloneqq \sum_{\ell=0}^{q-1}\sum_{|\balpha|=\ell}\partial^{\balpha}\widehat{\Phi}_n(w)\frac{(x-w)^{\balpha}}{\balpha!}
    \quad\text{and}\quad
    \widehat{F}_2(w,z)\coloneqq \sum_{|\balpha|=q}\partial^{\balpha}\widehat{\Phi}_n(z)\frac{(x-w)^{\balpha}}{\balpha!}.
  \end{equation*}
  Then, it holds that
  \begin{equation*}
    \widehat{B}^{\operatorname{up}}_n(x)=\min_{i\in\{1,\dots,n\}}\left(\widehat{F}_1(X_i)+\max_{k\in\{1,\dots,n\}}\widehat{F}_2(X_i,X_k)\right).
  \end{equation*}
  Lastly, define for all $w\in\Din$ the functions
  \begin{equation*}
    G(w)\coloneqq F_1(w)+\sup_{z\in\Din}F_2(w,z)
    \quad\text{and}\quad
    \widehat{G}(w)\coloneqq \widehat{F}_1(w)+\max_{k}\widehat{F}_2(w,X_k).
  \end{equation*}
  For the remaining proof we make use of the following properties:
  \begin{itemize}
  \item[(i)] $F_1$ is Lipschitz with constant $L_1>0$.
  \item[(ii)] There exists a constant $L_2>0$ such that for all
    $z\in\mathcal{X}$ the function $w\mapsto F_2(w, z)$ is
    Lipschitz with constant $L_2$.
  \item[(iii)] There exists a constant $L_3>0$ such that for all
    $w\in\mathcal{X}$ the function $z\mapsto F_2(w, z)$ is Lipschitz
    with constant $L_3$.
  \item[(iv)] For all functions $f:\mathcal{D}\rightarrow\mathbb{R}$
    and $f:\mathcal{D}\rightarrow\mathbb{R}$ it holds that
    \begin{equation*}
      |\sup_{z\in\mathcal{D}} f(z) - \sup_{z\in\mathcal{D}} g(z)|\leq \sup_{z\in\mathcal{D}}|f(z)-g(z)|
      \quad\text{and}\quad
      |\inf_{z\in\mathcal{D}} f(z) - \inf_{z\in\mathcal{D}} g(z)|\leq \sup_{z\in\mathcal{D}}|f(z)-g(z)|.
    \end{equation*}
  \end{itemize}
  Property~(i) holds since any continuously differentiable function on
  a compact set is Lipschitz. To prove property (ii) we use that by
  compactness of $\mathcal{X}$ there exists $\bar{K}^1\in\mathbb{R}$
  such that for all $\balpha\in\mathbb{N}^d$ with $|\balpha|=q$, it
  holds that
  $\sup_{z\in\mathcal{X}}|\partial^{\balpha}\Phi_0(z)|<\bar{K}^1$. Moreover,
  for all $\balpha\in\mathbb{N}^d$ with $|\balpha|=q$, the function
  $w\mapsto (x-w)^{\balpha}$ is Lipschitz with constant
  $\bar{L}_{\balpha}^1$ (since it is differentiable and $\mathcal{X}$
  is compact). Then, for all $z,w,w'\in\mathcal{X}$ we can apply the
  triangle inequality and bound each term to get
  \begin{equation*}
    |F_2(w, z) - F_2(w', z)|\leq\sum_{|\balpha|=q}\frac{\bar{K}^1\bar{L}^1_{\balpha}}{\balpha!}\|w-w'\|_2,
  \end{equation*}
  which proves (ii) since the resulting constant does not depend on
  $z$. For (iii), we additionally use that since $\Phi_0$ is
  $(q+1)$-differentiable it holds for all $\balpha\in\mathbb{N}^d$
  with $|\balpha|=q$ that $\partial^{\balpha}\Phi_0$ is Lipschitz with
  constant $\bar{L}_{\balpha}^2$ (since $\mathcal{X}$ is
  compact). Then, for $w,z,z'\in\mathcal{X}$ it holds using the
  triangle inequality and bounding the terms that
  \begin{equation*}
    |F_2(w, z) - F_2(w, z')|\leq\sum_{|\balpha|=q}\frac{K^q\bar{L}^2_{\balpha}}{\balpha!}\|z-z'\|_1,
  \end{equation*}
  which proves (iii) since the resulting constant does not depend on
  $w$. Finally, for (iv) observe that for all $w\in\mathcal{D}$ it holds
  \begin{equation*}
    f(w)\leq \sup_{z\in\mathcal{D}}|f(z)-g(z)| + g(w).
  \end{equation*}
  Now taking either the sup or inf on both sides, rearranging and then
  swapping the roles of $f$ and $g$ proves the result.

  We are now ready to prove the main result in three steps:
  \begin{itemize}
  \item Step 1: Show that
    \begin{equation}
      \label{eq:convergence_of_F1}
      \sup_{w\in\Din}\left|F_1(w)-\widehat{F}_1(w)\right|\overset{P_0}{\longrightarrow}0
      \quad\text{as }n\rightarrow\infty.
    \end{equation}
  \item Step 2: Show that
    \begin{equation}
      \label{eq:convergence_of_F2}
      \sup_{w\in\Din}\left|\sup_{z\in\Din}F_2(w,
        z)-\max_{k}\widehat{F}_2(w,
        X_k)\right|\overset{P_0}{\longrightarrow}0
      \quad\text{as }n\rightarrow\infty.
    \end{equation}
  \item Step 3: Conclude the proof using $G$ and $\widehat{G}$.
  \end{itemize}

  \noindent \textit{Step 1:} To see this, we use the triangle
  inequality and bound the resulting terms as follows
  \begin{align*}
    &\sup_{w\in\Din}\left|F_1(w)-\widehat{F}_1(w)\right|\\
    &\quad=\sup_{w\in\Din}\left|\sum_{\ell=0}^{q-1}\sum_{|\balpha|=\ell}\left(\partial^{\balpha}\Phi_{0}(w)-\partial^{\balpha}\widehat{\Phi}_n(w)\right)\frac{(x-w)^{\balpha}}{\balpha!}\right|\\
    &\quad\leq\sum_{\ell=0}^{q-1}\sum_{|\balpha|=\ell}\sup_{w\in\Din}\left|\partial^{\balpha}\Phi_{0}(w)-\partial^{\balpha}\widehat{\Phi}_n(w)\right|\frac{K^q}{\balpha!}.
  \end{align*}
  Since each of the summands converges in probability to zero by
  assumption, this proves \eqref{eq:convergence_of_F1}.

  \noindent \textit{Step 2:} We begin by fixing $w\in\Din$. Then it
  holds that
  \begin{align}
    &\left|\sup_{z\in\Din}F_2(w, z)-\max_{k}\widehat{F}_2(w, X_k)\right|\nonumber\\
    &\quad\leq
      \left|\max_{k}\widehat{F}_2(w,X_k)-\max_{k}F_2(w,X_k)\right| +
      \left|\sup_{z\in\Din}F_2(w, z)-\max_{k}F_2(w, X_k)\right|\nonumber\\
    &\quad\leq \max_{k}\left|\widehat{F}_2(w,X_k)-F_2(w,X_k)\right| +
      \left|\sup_{z\in\Din}F_2(w, z)-\max_{k}F_2(w, X_k)\right|\nonumber\\
    &\quad\leq
      \sup_{z\in\Din}\left|\widehat{F}_2(w,z)-F_2(w,z)\right| +
      \left|\sup_{z\in\Din}F_2(w, z)-\max_{k}F_2(w, X_k)\right|,\label{eq:F2bound}
  \end{align}
  where for the second inequality, we used property~(iv). Next,
  we bound each summand separately. First, using the triangle
  inequality and bounding the terms we get
  \begin{align}
    \sup_{z\in\Din}\left|\widehat{F}_2(w, z)-F_2(w, z)\right|
    &=\sup_{z\in\Din}\left|\sum_{|\balpha|=q}\left(\partial^{\balpha}\widehat{\Phi}_n(z)-\partial^{\balpha}\Phi_{0}(z)\right)\frac{(x-w)^{\balpha}}{\balpha!}\right|\nonumber\\
    &\leq\sum_{|\balpha|=q}\sup_{z\in\Din}\left|\partial^{\balpha}\widehat{\Phi}_n(z)-\partial^{\balpha}\Phi_{0}(z)\right|\frac{K^q}{\balpha!}.\label{eq:F2bound_part1}
  \end{align}
  Second, fix $w\in\Din$, then since $\mathcal{X}$ is compact there
  exists $z_w\in\Din$ (since $\Din$ is closed) such that
  $\sup_{z\in\Din}F_2(w,z)=F_2(w,z_w)$. Next, recall
  $\Lambda_n=\sup_{z\in\Din}\min_{k}\|X_k-z\|_2$ and use that
  $X_1,\ldots,X_n\in\Din$ to get
  \begin{align}
    \left|\sup_{z\in\Din}F_2(w, z)-\max_{k}F_2(w,
    X_k)\right|
    &=F_2(w, z_w)-\max_{k}F_2(w,
      X_k)\nonumber\\
    &\leq F_2(w,
      z_w)-\inf_{\substack{z\in\Din:\\\|z -z_w\|_2\leq\Lambda_n}}F_2(w,z)\nonumber\\
    &\leq F_2(w, z_w)-(F_2(w,z_w)-\Lambda_n L_3)\nonumber\\
    &= \Lambda_n L_3,\label{eq:dense_points_sampling}
  \end{align}
  where fore the last inequality we used that $F_2$ is Lipschitz with
  constant $L_3$ in the second argument by property~(iii). Since the
  bound does not depend on $w$ this implies
  \begin{equation}
    \label{eq:F2bound_part2}
    \sup_{w\in\Din}\left|\sup_{z\in\Din}F_2(w, z)-\max_{k}F_2(w,
      X_k)\right|\leq \Lambda_n L_3.
  \end{equation}
  Using \eqref{eq:F2bound_part1} and \eqref{eq:F2bound_part2} to further bound \eqref{eq:F2bound} we get
  \begin{align*}
    &\sup_{w\in\Din}\left|\sup_{z\in\Din}F_2(w, z)-\max_{k}\widehat{F}_2(w, X_k)\right|\\
    &\quad\leq
      \sum_{|\balpha|=q}\sup_{z\in\Din}\left|\partial^{\balpha}\widehat{\Phi}_n(z)-\partial^{\balpha}\Phi_{0}(z)\right|\frac{K^q}{\balpha!}
      + \Lambda_n L_3.
  \end{align*}
  Since each summand converges in probability to zero by assumption, this implies
  \eqref{eq:convergence_of_F2}.

  \noindent\textit{Step 3:}
  Since by the triangle inequality it holds that
  \begin{equation*}
    \sup_{w\in\Din}\left|G(w)-\widehat{G}(w)\right|
    \leq \sup_{w\in\Din}\left|F_1(w)-\widehat{F}_1(w)\right| + \sup_{w\in\Din}\left|\sup_{z\in\Din}F_2(w, z)-\max_{k}\widehat{F}_2(w, X_k)\right|,
  \end{equation*}
  we get by \eqref{eq:convergence_of_F1} and
  \eqref{eq:convergence_of_F2} that
  \begin{equation}
    \label{eq:convergence_of_G}
    \sup_{w\in\Din}\left|G(w)-\widehat{G}(w)\right|\overset{P_0}{\longrightarrow}0
    \quad\text{as }n\rightarrow\infty.
  \end{equation}
  Similar to the argument in \eqref{eq:F2bound}, we get
  \begin{align}
    \left|\inf_{w\in\Din}G(w)-\min_k\widehat{G}(X_k)\right|
    &\leq
      \left|\min_{k}\widehat{G}(X_k)-\min_{k}G(X_k)\right| +
      \left|\inf_{w\in\Din}G(w)-\min_{k}G(X_k)\right|\nonumber\\
    &\leq \max_{k}
      \left|\widehat{G}(X_k)-G(X_k)\right| +
      \left|\inf_{w\in\Din}G(w)-\min_{k}G(X_k)\right|\nonumber\\
    &\leq
      \sup_{w\in\Din}\left|\widehat{G}(w)-G(w)\right| +
      \left|\inf_{w\in\Din}G(w)-\min_{k}G(X_k)\right|,\label{eq:Gbound_almost}
  \end{align}
  where for the second inequality we used property~(iv).  Moreover,
  for all $x,y\in\Din$ we can use properties~(i),~(ii) and~(iv) to get
  \begin{align*}
    \left|G(x)-G(y)\right|
    &\leq \left|F_1(x)-F_1(y)\right| +
      \left|\sup_{z\in\Din}F_2(x,z)-\sup_{z\in\Din}F_2(y,z)\right|\\
    &\leq L_1\left\|x-y\right\|_2 +
      \sup_{z\in\Din}\left|F_2(x,z)-F_2(y,z)\right|\\
    &\leq L_1\left\|x-y\right\|_2 + L_2\left\|x-y\right\|_2.
  \end{align*}
  Hence, letting $w^*\in\Din$ be such that
  $\inf_{w\in\Din}G(w)=G(w^*)$, we get by a similar argument as in
  \eqref{eq:dense_points_sampling} that
  \begin{align*}
    \left|\inf_{w\in\Din}G(w)-\min_{k}G(X_k)\right|
    &=\min_{k}G(X_k)-G(w^*)\nonumber\\
    &\leq \sup_{\substack{w\in\Din:\\\|w-w^*\|_2\leq\Lambda_n}}G(w)-G(w^*)\nonumber\\
    &\leq G(w^*)+\Lambda_n (L_1+L_2)-G(w^*)\nonumber\\
    &= \Lambda_n (L_1+L_2).
  \end{align*}
  Using this in \eqref{eq:Gbound_almost}, leads to
  \begin{equation*}
    \left|\inf_{w\in\Din}G(w)-\min_k\widehat{G}(X_k)\right|
    \leq \sup_{w\in\Din}\left|\widehat{G}(w)-G(w)\right| +
    \Lambda_n (L_1+L_2).
  \end{equation*}
  Since both summands converge in probability to zero by \eqref{eq:convergence_of_G} and
  assumption, this implies
  \begin{equation*}
    \left|B^{\operatorname{up}}_{\Phi_{0},\Din}(x)-\widehat{B}^{\operatorname{up}}_n(x)\right|
    =\left|\inf_{w\in\Din}G(w)-\min_k\widehat{G}(X_k)\right|\overset{P_0}{\longrightarrow}0
    \quad\text{as }n\rightarrow\infty,
  \end{equation*}
  which completes the proof of Theorem~\ref{thm:consistency}.
\end{proof}

\section{Auxiliary results}\label{sec:auxiliary_results}

\begin{lemma}
  \label{thm:uniform_continuity_of_bounds}
  Assume $\mathcal{X}$ is compact, let $f\in C^q(\mathcal{X})$ and
  $\mathcal{D}\subseteq\mathcal{X}$ non-empty closed. Then, there
  exists a constant $C>0$ such that for all $\varepsilon>0$ it holds
  that
  \begin{equation*}
    \sup_{\substack{x,y\in\mathcal{X}:\\ \|x-y\|_2<\varepsilon}}\big| B_{q, f,
      \mathcal{D}}^{\operatorname{up}}(x)- B_{q, f,
      \mathcal{D}}^{\operatorname{up}}(y)\big|<C\,\varepsilon.
  \end{equation*}
\end{lemma}

\begin{proof}
  First, introduce for all $x_0\in\mathcal{D}$ the functions
  $B^{x_0}:\mathcal{X}\rightarrow\mathbb{R}$ defined for all
  $x\in\mathcal{X}$ by
  \begin{equation*}
    B^{x_0}(x)\coloneqq \sum_{\ell=0}^{q-1}D_{\overline{v}(x_0, x)}^{\ell}f(x_0)\cdot\frac{\|x-x_0\|_2^{\ell}}{\ell!} + \sup_{z\in\mathcal{D}}D^q_{\overline{v}(x_0,x)}f(z)\cdot \frac{\|x-x_0\|_2^{q}}{q!}.
  \end{equation*}
  Fix $\varepsilon>0$ and $x^*,y^*\in\mathcal{X}$ with
  $\|x^*-y^*\|<\varepsilon$. Since $\mathcal{D}$ is compact, we can
  define a function $\optx:\mathcal{X}\rightarrow\mathcal{D}$ such
  that for all $x\in\mathcal{X}$ it holds that
  $\optx(x)\in\mathcal{D}$ satisfies
  \begin{equation*}
    B_{q,f,\mathcal{D}}^{\operatorname{up}}(x)=B^{\optx(x)}(x).
  \end{equation*}
  By construction this implies that
  \begin{equation*}
    \left| B_{q, f,\mathcal{D}}^{\operatorname{up}}(x^*) - B_{q,
        f,\mathcal{D}}^{\operatorname{up}}(y^*)\right|
    =\left| B^{\optx(x^*)}(x^*) -
      B^{\optx(y^*)}(y^*)\right|.
  \end{equation*}
  Furthermore, since the points $\optx(x^*)$ and
  $\optx(y^*)$ are chosen as minimizers, it holds that
  \begin{equation*}
    \left| B^{\optx(x^*)}(x^*) -
      B^{\optx(y^*)}(y^*)\right|\leq
    \begin{cases}
      \big| B^{\optx(x^*)}(x^*) -
      B^{\optx(x^*)}(y^*)\big|\quad
      &\text{if } B^{\optx(x^*)}(x^*)\leq B^{\optx(y^*)}(y^*)\\
      \big| B^{\optx(y^*)}(x^*) -
      B^{\optx(y^*)}(y^*)\big|\quad &\text{otherwise}.
    \end{cases}
  \end{equation*}
  Without loss of generality we now assume
  $B^{\optx(x^*)}(x^*)\leq B^{\optx(y^*)}(y^*)$ and
  bound the expression in this case. The second case follows
  analogous with the role of $x^*$ and $y^*$ switched. First, by the
  triangle inequality we get
  \begin{align}
    &\left| B^{\optx(x^*)}(x^*) -
      B^{\optx(x^*)}(y^*)\right|\nonumber\\
    &\quad=\bigg| \sum_{\ell=0}^{q-1}\frac{1}{\ell!}\left(D_{\overline{v}(\optx(x^*),
      x^*)}^{\ell}f(\optx(x^*))\|x^*-\optx(x^*)\|_2^{\ell}-D_{\overline{v}(\optx(x^*),
      y^*)}^{\ell}f(\optx(x^*))\|y^*-\optx(x^*)\|_2^{\ell}\right)\nonumber \\
    &\qquad\qquad + \frac{1}{q!}\left(\sup_{z\in\mathcal{D}}D^q_{\overline{v}(\optx(x^*),
      x^*)}f(z)
      \|x^*-\optx(x^*)\|_2^{q}-\sup_{z\in\mathcal{D}}D^q_{\overline{v}(\optx(x^*),
      y^*)}f(z)
      \|y^*-\optx(x^*)\|_2^{q}\right)\bigg|\nonumber\\
    &\quad\leq\sum_{\ell=0}^{q-1}\frac{1}{\ell!}\bigg|D_{\overline{v}(\optx(x^*),
      x^*)}^{\ell}f(\optx(x^*))\|x^*-\optx(x^*)\|_2^{\ell}-D_{\overline{v}(\optx(x^*),
      y^*)}^{\ell}f(\optx(x^*))\|y^*-\optx(x^*)\|_2^{\ell}\bigg| \nonumber \\
    &\qquad\qquad + \frac{1}{q!}\bigg| \sup_{z\in\mathcal{D}}D^q_{\overline{v}(\optx(x^*),
      x^*)}f(z)
      \|x^*-\optx(x^*)\|_2^{q}-\sup_{z\in\mathcal{D}}D^q_{\overline{v}(\optx(x^*),
      y^*)}f(z)
      \|y^*-\optx(x^*)\|_2^{q}\bigg|.\label{eq:main_bound_individual_summands}
  \end{align}
  We now consider the summands individually. Firstly, for all
  $\ell\in\{1,\ldots, q-1\}$, it holds that
  \begin{align}
    &\bigg|D_{\overline{v}(\optx(x^*),
      x^*)}^{\ell}f(\optx(x^*))\|x^*-\optx(x^*)\|_2^{\ell}-D_{\overline{v}(\optx(x^*),
      y^*)}^{\ell}f(\optx(x^*))\|y^*-\optx(x^*)\|_2^{\ell}\bigg|\nonumber\\
    &\quad\leq \bigg|D_{\overline{v}(\optx(x^*),
      x^*)}^{\ell}f(\optx(x^*))\bigg| \,
      \bigg|\|x^*-\optx(x^*)\|_2^{\ell} -
      \|y^*-\optx(x^*)\|_2^{\ell}\bigg|\nonumber \\
    &\qquad\qquad + \Big|D_{\overline{v}(\optx(x^*),
      x^*)}^{\ell}f(\optx(x^*)) - D_{\overline{v}(\optx(x^*),
      y^*)}^{\ell}f(\optx(x^*))\Big| \|y^*-\optx(x^*)\|_2^{\ell}.\label{eq:lterm_bound}
  \end{align}
  We separate two cases: (i) $\|x^*-\optx(x^*)\|_2< 2\varepsilon$ and
  (ii) $\|x^*-\optx(x^*)\|_2\geq 2\varepsilon$. For case (i), it holds
  by the triangle inequality that
  $\|y^*-\optx(x^*)\|_2\leq \|y^*-x^*\|_2 + \|x^*-\optx(x^*)\|_2\leq 4
  \varepsilon$. Combined with \eqref{eq:lterm_bound} and using that
  all directional derivatives of $f$ up to order $q$ are bounded by a
  constant $K>0$ this implies in case (i) that
  \begin{equation}
    \label{eq:lterm_bound_casei}
    \bigg|D_{\overline{v}(\optx(x^*),
      x^*)}^{\ell}f(\optx(x^*))\cdot\|x^*-\optx(x^*)\|_2^{\ell}-D_{\overline{v}(\optx(x^*),
      y^*)}^{\ell}f(\optx(x^*))\cdot\|y^*-\optx(x^*)\|_2^{\ell}\bigg|\leq
    14 K \varepsilon.
  \end{equation}
  For case (ii), consider
  $\mathcal{X}_{\varepsilon}\coloneqq\{x\in\mathcal{X}\mid
  \inf_{x_0\in\mathcal{D}}\|x-x_0\|_2\geq \varepsilon\}$ which is compact
  since $\mathcal{X}$ is compact. Next, observe that for all
  $x_0\in\mathcal{D}$ the function $x\mapsto \|x-x_0\|_2^{\ell}$ is
  Lipschitz continuous on $\mathcal{X}_{\varepsilon}$ with a constant
  $0<L_{x_0}^{\ell}<\infty$, which implies that
  \begin{equation}
    \label{eq:lterm_bound_caseii_1}
    \bigg|\|x^*-\optx(x^*)\|_2^{\ell} -
    \|y^*-\optx(x^*)\|_2^{\ell}\bigg|\leq \sup_{x_0\in\mathcal{D}}L_{x_0}^{\ell}\|x^*-y^*\|_2.
  \end{equation}
  Since $\mathcal{D}$ is compact the supremum is attained and
  $\bar{L}^{\ell}\coloneqq\sup_{x_0\in\mathcal{D}}L_{x_0}^{\ell}<\infty$.
  Furthermore, expressing the directional derivative in terms of
  partial derivatives (using multi-index notation) and letting $M>0$
  be an upper bound on all partial derivatives of $f$ up to order $q$,
  it holds that
  \begin{align}
    &\big| D^{\ell}_{\overline{v}(\optx(x^*),
      x^*)}f(\optx(x^*)) -
      D^{\ell}_{\overline{v}(\optx(x^*), y^*)}f(\optx(x^*))\big|\nonumber\\
    &\quad=\bigg|\sum_{\balpha\in\mathbb{N}^d:
      |\balpha|=\ell}\partial^{\balpha}f(\optx(x^*)) \tfrac{\ell!}{\balpha!}\left(\frac{x^*-\optx(x^*)}{\|x^*-\optx(x^*)\|_2}\right)^{\balpha}-\sum_{\balpha\in\mathbb{N}^d:
      |\balpha|=\ell}\partial^{\balpha}f(\optx(x^*))
      \tfrac{\ell!}{\balpha!}\left(\frac{y^*-\optx(x^*)}{\|y^*-\optx(x^*)\|_2}\right)^{\balpha}\bigg|\nonumber\\
    &\quad\leq \sum_{\balpha\in\mathbb{N}^d:
      |\balpha|=\ell}M \tfrac{\ell!}{\balpha!}\left|\left(\frac{x^*-\optx(x^*)}{\|x^*-\optx(x^*)\|_2}\right)^{\balpha}-\left(\frac{y^*-\optx(x^*)}{\|y^*-\optx(x^*)\|_2}\right)^{\balpha}\right|.\label{eq:term_without_sup}
  \end{align}
  Next, observe that for all $\balpha\in\mathbb{N}^d$ with
  $|\balpha|=\ell$ and all $x_0\in\mathcal{D}$ it holds that the
  function $x\mapsto \left(\frac{x-x_0}{\|x-x_0\|}\right)^{\balpha}$
  is Lipschitz continuous on $\mathcal{X}_{\varepsilon}$ since it is
  differentiable everywhere on $\mathcal{X}_{\varepsilon}$ and has bounded
  derivatives. Denote by $N_{x_0}^{\ell}$ the corresponding Lipschitz
  constant, then
  \begin{equation}
    \label{eq:lterm_bound_caseii_2}
    \big| D^{\ell}_{\overline{v}(\optx(x^*),
      x^*)}f(\optx(x^*)) -
    D^{\ell}_{\overline{v}(\optx(x^*), y^*)}f(\optx(x^*))\big|\leq \Big(\textstyle\sum_{\balpha\in\mathbb{N}^d:
      |\balpha|=\ell}M \tfrac{\ell!}{\balpha!}\Big)\sup_{x_0\in\mathcal{D}}N_{x_0}^{\ell}\|x^*-y^*\|_2.
  \end{equation}
  Again, since $\mathcal{D}$ is compact the supremum is attained and
  $\bar{N}^{\ell}\coloneqq\left(\sum_{\balpha\in\mathbb{N}^d:
      |\balpha|=\ell}M
    \tfrac{\ell!}{\balpha!}\right)\sup_{x_0\in\mathcal{D}}N_{x_0}^{\ell}<\infty$. Hence,
  combining \eqref{eq:lterm_bound}, \eqref{eq:lterm_bound_caseii_1}
  and \eqref{eq:lterm_bound_caseii_2} we get in case (ii) that
  \begin{equation}
    \label{eq:lterm_bound_caseii_final}
    \bigg|D_{\overline{v}(\optx(x^*),
      x^*)}^{\ell}f(\optx(x^*))\cdot\|x^*-\optx(x^*)\|_2^{\ell}-D_{\overline{v}(\optx(x^*),
      y^*)}^{\ell}f(\optx(x^*))\cdot\|y^*-\optx(x^*)\|_2^{\ell}\bigg|\leq
    K \bar{L}^{\ell} \varepsilon + \bar{N}^{\ell} C \varepsilon,
  \end{equation}
  where $C\coloneqq\sup_{x,y\in\mathcal{X}}\|x-y\|<\infty$ (since $\mathcal{X}$
  is compact). Combining the bound \eqref{eq:lterm_bound_casei} from case (i) and the
  bound \eqref{eq:lterm_bound_caseii_final} from case (ii) implies
  \begin{align}
    &\bigg|D_{\overline{v}(\optx(x^*),
      x^*)}^{\ell}f(\optx(x^*))\cdot\|x^*-\optx(x^*)\|_2^{\ell}-D_{\overline{v}(\optx(x^*),
      y^*)}^{\ell}f(\optx(x^*))\cdot\|y^*-\optx(x^*)\|_2^{\ell}\bigg|\nonumber\\
    &\quad\leq
    \max\left(14 K, K \bar{L}^{\ell} + \bar{N}^{\ell} C\right)\varepsilon.
          \label{eq:lterm_bound_final}
  \end{align}
  The only summand in \eqref{eq:main_bound_individual_summands} that
  remains to be bounded is the one involving supremum terms. For this
  term the exact same bounds as above apply except for the bound in
  \eqref{eq:term_without_sup}. To bound this term, use compactness of
  $\mathcal{D}$ to define the function
  $\optz:\mathcal{X}\rightarrow\mathcal{D}$ which satisfies for all
  $x\in\mathcal{X}$ that
  \begin{equation*}
    \sup_{z\in\mathcal{D}}D^{q}_{\bar{v}(\optx(x^*), x)}f(z)=D^{q}_{\bar{v}(\optx(x^*), x)}f(\optz(x)).
  \end{equation*}
  Then, using the triangle inequality we get
  \begin{align*}
    &\big|\sup_{z\in\mathcal{D}} D^{q}_{\overline{v}(\optx(x^*),
      x^*)}f(z) -
      \sup_{z\in\mathcal{D}} D^{q}_{\overline{v}(\optx(x^*),
      y^*)}f(z)\big|\\
    &\quad=\big|D^{q}_{\overline{v}(\optx(x^*),
      x^*)}f(\optz(x^*)) -
      D^{q}_{\overline{v}(\optx(x^*),
      y^*)}f(\optz(y^*))\big|\\
    &\quad\leq\big|D^{q}_{\overline{v}(\optx(x^*),
      x^*)}f(\optz(x^*)) -
      D^{q}_{\overline{v}(\optx(x^*),
      y^*)}f(\optz(x^*))\big| + \big|D^{q}_{\overline{v}(\optx(x^*),
      y^*)}f(\optz(y^*)) -
      D^{q}_{\overline{v}(\optx(x^*),
      x^*)}f(\optz(y^*))\big|\\
    &\quad\qquad\qquad + \big|D^{q}_{\overline{v}(\optx(x^*),
      y^*)}f(\optz(x^*)) -
      D^{q}_{\overline{v}(\optx(x^*),
      x^*)}f(\optz(y^*))\big|.
  \end{align*}
  Now for the last term we can observe that if
  $D^{q}_{\overline{v}(\optx(x^*), y^*)}f(\optz(x^*))\leq
  D^{q}_{\overline{v}(\optx(x^*), x^*)}f(\optz(y^*))$, we can use that
  by definition of $\optz$ it holds that
  $D^{q}_{\overline{v}(\optx(x^*), x^*)}f(\optz(y^*))\leq
  D^{q}_{\overline{v}(\optx(x^*), x^*)}f(\optz(x^*))$ and hence
  \begin{equation*}
    \big|D^{q}_{\overline{v}(\optx(x^*),
      y^*)}f(\optz(x^*)) -
    D^{q}_{\overline{v}(\optx(x^*),
      x^*)}f(\optz(x^*))\big|.
  \end{equation*}
  Similarly if $D^{q}_{\overline{v}(\optx(x^*), y^*)}f(\optz(x^*))\geq
  D^{q}_{\overline{v}(\optx(x^*), x^*)}f(\optz(y^*))$ we get
  \begin{equation*}
    \big|D^{q}_{\overline{v}(\optx(x^*),
      y^*)}f(\optz(y^*)) -
    D^{q}_{\overline{v}(\optx(x^*),
      x^*)}f(\optz(y^*))\big|.
  \end{equation*}
  Hence, combined we get that
  \begin{align*}
    &\big|\sup_{z\in\mathcal{D}} D^{q}_{\overline{v}(\optx(x^*),
      x^*)}f(z) -
      \sup_{z\in\mathcal{D}} D^{q}_{\overline{v}(\optx(x^*),
      y^*)}f(z)\big|\\
    &\quad\leq 2\big|D^{q}_{\overline{v}(\optx(x^*),
      x^*)}f(\optz(x^*)) -
      D^{q}_{\overline{v}(\optx(x^*),
      y^*)}f(\optz(x^*))\big| + 2\big|D^{q}_{\overline{v}(\optx(x^*),
      y^*)}f(\optz(y^*)) -
      D^{q}_{\overline{v}(\optx(x^*),
      x^*)}f(\optz(y^*))\big|.
  \end{align*}
  Using the same argument as for \eqref{eq:lterm_bound_caseii_2} this
  results in
  \begin{equation*}
    \big|\sup_{z\in\mathcal{D}} D^{q}_{\overline{v}(\optx(x^*),
      x^*)}f(z) -
      \sup_{z\in\mathcal{D}} D^{q}_{\overline{v}(\optx(x^*),
      y^*)}f(z)\big|\leq 4\bar{N}^{q}\|x^*-y^*\|.
  \end{equation*}
  In total, we get the following bound for the supremum term
  \begin{align}
    &\bigg|\sup_{z\in\mathcal{D}}D^q_{\overline{v}(\optx(x^*),
      x^*)}f(z)\|x^*-\optx(x^*)\|_2^{q}-\sup_{z\in\mathcal{D}}D^q_{\overline{v}(\optx(x^*),
      y^*)}f(z)\|y^*-\optx(x^*)\|_2^{q}\bigg|\nonumber\\
    &\quad\leq
    \max\left(14 K, K \bar{L}^{q}  + 4\bar{N}^{q} C\right)\varepsilon.
      \label{eq:supterm_bound_final}
  \end{align}
  Finally, using the bounds \eqref{eq:lterm_bound_final} and
  \eqref{eq:supterm_bound_final} in
  \eqref{eq:main_bound_individual_summands} leads to
  \begin{align}
    &\left| B^{\optx(x^*)}(x^*) -
      B^{\optx(x^*)}(y^*)\right|\nonumber\\
    &\quad\leq \varepsilon\underbrace{\left(\sum_{\ell=0}^{q-1}\frac{1}{\ell!} \max\left(14 K, K \bar{L}^{\ell}+ \bar{N}^{\ell} C\right) + \max\left(14 K , K \bar{L}^{q} + 4\bar{N}^{q} C
      \right)\right)}_{\eqqcolon c}. \label{eq:final_bound}
  \end{align}
  Since, $c>0$ only depends on $K$, $\bar{L}^{\ell}$,
  $\bar{N}^{\ell}$, $C$ and $q$, we have shown for all $\varepsilon>0$ and
  all $x^*,y^*\in\mathcal{X}$ with $\|x^*-y^*\|_2<\varepsilon$ that
  \begin{equation*}
    \left| B_{q, f,\mathcal{D}}^{\operatorname{up}}(x^*) - B_{q,
        f,\mathcal{D}}^{\operatorname{up}}(y^*)\right|\leq c\, \varepsilon,
  \end{equation*}
  which completes the proof of
  Lemma~\ref{thm:uniform_continuity_of_bounds}.
\end{proof}

\begin{lemma}
  \label{thm:properties_first_approximation}
  Assume $\mathcal{X}$ is compact, let $f\in C^q(\mathcal{X})$ and
  $\mathcal{D}\subseteq\mathcal{X}$ non-empty closed.  For all
  $z\in\mathcal{X}$ let $B_{z}:\mathbb{R}^d\rightarrow\mathbb{R}$ be
  the functions defined in \eqref{eq:B_z_delta} in the proof of
  Theorem~\ref{thm:taylor_extrapolation_bounds}.  Then, the following
  statements hold:
  \begin{itemize}
  \item[(i)] For all $z\in\mathcal{X}$ the function $B_z$ is
    infinitely often continuously differentiable and for all
    $x\in\mathbb{R}$ and all $\bbeta\in\mathbb{N}^d$ with
    $|\bbeta|\leq q$ it holds that
    \begin{equation*}
      \partial^{\bbeta}B_z(x)=\sum_{|\balpha|<
        q}\partial^{\balpha}f(\optx(z))\tfrac{(x-\optx(z))^{\balpha-\bbeta}}{(\balpha-\bbeta)!}\mathds{1}(\bbeta\leq\balpha)
        +\sum_{|\balpha|=q}\partial^{\balpha}f(\optz(z))\tfrac{(x-\optx(z))^{\balpha-\bbeta}}{(\balpha-\bbeta)!}\mathds{1}(\bbeta\leq\balpha)
    \end{equation*}
    In particular, for $|\bbeta|=q$ this implies
    \begin{equation*}
      \partial^{\bbeta}B_z(x)=\partial^{\bbeta}f(\optz(z)).
    \end{equation*}
  \item[(ii)] There exists a constant $C>0$ such that for all $\varepsilon>0$
    and all $z\in\mathcal{X}$ it holds that
    \begin{equation*}
      \sup_{\substack{x,y\in\mathcal{X}:\\ \|x-y\|_2<\varepsilon}}\big|B_z(x)-B_z(y)\big|<C\,\varepsilon.
    \end{equation*}
  \end{itemize}
\end{lemma}

\begin{proof}
  Since $B_z$ is a polynomial of order $q$ it is infinitely often
  continuously differentiable. Moreover, a direct computation shows
  for all $x\in\mathbb{R}^d$ and all $\bbeta\in\mathbb{N}^d$ with
  $|\bbeta|\leq q$ that
  \begin{equation*}
    \label{eq:Bz_derivative_formula}
    \partial^{\bbeta}B_z(x)=\sum_{|\balpha|<
      q}\partial^{\balpha}f(\optx(z))\tfrac{(x-\optx(z))^{\balpha-\bbeta}}{(\balpha-\bbeta)!}\mathds{1}(\bbeta\leq\balpha)
    +\sum_{|\balpha|=q}\partial^{\balpha}f(\optz(z))\tfrac{(x-\optx(z))^{\balpha-\bbeta}}{(\balpha-\bbeta)!}\mathds{1}(\bbeta\leq\balpha).
  \end{equation*}
  For $|\bbeta|=q$ all summands are zero except the one with
  $\balpha=\bbeta$, hence
  \begin{equation*}
    \partial^{\bbeta}B_z(x)=\partial^{\bbeta}f(\optz(z)).
  \end{equation*}
  Next, we prove (ii). Fix $x,y,z\in\mathcal{X}$. Then, it holds that
  \begin{align}
    \big|B_z(x)-B_z(y)\big|
    &\leq
    \sum_{|\balpha|<q}\tfrac{1}{\balpha!}\big|\partial^{\balpha}f(\optx(z))\big|\cdot
      \big|(x-\optx(z))^{\balpha} - (y-\optx(z))^{\balpha}\big|\nonumber\\
    &\qquad\qquad +
    \sum_{|\balpha|=q}\tfrac{1}{\balpha!}\big|\partial^{\balpha}f(\optz(z))\big|\cdot
    \big|(x-\optx(z))^{\balpha} - (y-\optx(z))^{\balpha}\big|.\label{eq:Bx_bound1}
  \end{align}
  Next, observe that for all $x_0\in\mathcal{X}$ the function
  $x\mapsto (x-x_0)^{\balpha}$ is Lipschitz continuous with constant
  $L_{x_0}^{\balpha}$ (using that $\mathcal{X}$ is compact). Hence,
  \begin{equation*}
    \big|(x-\optx(z))^{\balpha} - (y-\optx(z))^{\balpha}\big|\leq
    L_{\optx(z)}^{\balpha}\|x - y\|_2\leq \bar{L}^{\balpha}\|x - y\|_2,
  \end{equation*}
  where $\bar{L}^{\balpha}\coloneqq \sup_{x\in\mathcal{X}}L_{x}^{\balpha}<\infty$ since
  $\mathcal{X}$ is compact. Together with \eqref{eq:Bx_bound1} and
  since all partial derivatives of $f$ up to order $q$ are bounded by
  a constant $K>0$ this yields
  \begin{align*}
    \big|B_z(x)-B_z(y)\big|
    \leq
    \underbrace{\left(\sum_{|\balpha|<q}\tfrac{1}{\balpha!}K\cdot
      \bar{L}^{\balpha} +
    \sum_{|\balpha|=q}\tfrac{1}{\balpha!}K\cdot
    \bar{L}^{\balpha}\right)}_{\eqqcolon C<\infty} \|x-y\|_2.
  \end{align*}
  As $x,y,z\in\mathcal{X}$ were arbitrary and $C$ does not depend on
  them, this implies for all $z\in\mathcal{X}$ and all $\varepsilon>0$
  that
  \begin{equation}
    \label{eq:Bx_bound2}
    \sup_{\substack{x,y\in\mathcal{X}:\\ \|x-y\|_2<\varepsilon}}\big|B_z(x)-B_z(y)\big|<C\,\varepsilon.
  \end{equation}
  This completes the proof of
  Lemma~\ref{thm:properties_first_approximation}.
\end{proof}

\end{document}